\documentclass[envcountsame]{llncs}

\usepackage{amsmath,amssymb} 
\let\originallhook\lhook
\usepackage{mnsymbol}
\let\lhook\originallhook

\usepackage{varioref}
\usepackage[pdfpagelabels=true,linktocpage,hidelinks,bookmarks,unicode,colorlinks=true]{hyperref} 
\usepackage{xcolor} 
\usepackage{listings} 
\usepackage{lipsum}	
\usepackage[pangram]{blindtext}
\usepackage[capitalize]{cleveref}
\usepackage{mathtools}
\usepackage[normalem]{ulem}
\usepackage{placeins}
\usepackage{cprotect}
\usepackage{wrapfig}
\usepackage{subcaption}
\usepackage{fancyvrb}
\usepackage{boxedminipage}
\usepackage{array}
\usepackage{stmaryrd}
\usepackage[shortlabels]{enumitem}
\usepackage{fancyvrb}
	
\usepackage{graphicx}
\PassOptionsToPackage{usenames,dvipsnames,svgnames,table}{xcolor}
\usepackage{ wasysym }
\usepackage{tikz}
\usepackage{pgfplots}
\usetikzlibrary{arrows}
\usetikzlibrary{patterns}
\usetikzlibrary{shapes,arrows,calc,arrows}
\usetikzlibrary{automata,arrows,decorations.pathmorphing,positioning,fit}
\usetikzlibrary{trees,shapes}
\usetikzlibrary{decorations.pathreplacing}
\usetikzlibrary{calc}
\usetikzlibrary{intersections}
\usetikzlibrary{backgrounds}
\usetikzlibrary{fadings}
\usetikzlibrary{tikzmark}

\usepackage{calc}
\usepackage{tkz-euclide}
\usepackage{catchfilebetweentags}
\usepackage{afterpage}
\usepackage{url}
\usepackage{todonotes}
\usepackage{xstring}
\usepackage{nameref}
\usepackage{alphalph}
\usepackage{relsize}

\usepackage[utf8]{inputenc} % unicode encoding
\usepackage{algorithm2e} 
\usepackage{scrextend}
\usepackage{refcount} %set counter values from references

\usepackage{ntheorem}
\usepackage{forloop}

\renewcommand{\emptyset}{\varnothing}

\newcommand{\code}[1]{\texttt{#1}}
\newcommand{\tool}[1]{\textsf{#1}}
\newcommand{\sfC}{\tool{C}}
\newcommand{\LLVM}{\tool{LLVM}\xspace}

\newcommand{\aprove}{\tool{AProVE}}

\newcommand{\PP}{\mathcal{P}}
\newcommand{\Blks}{\mathit{Blks}}
\newcommand{\Ids}{\mathcal{V}_{\PP}}

\newcommand{\Pos}{p}
\newcommand{\GG}{\mathcal{G}}
\newcommand{\PPos}{\mathit{Pos}}

\newcommand{\bigast}{\mathop{\mathlarger{\mathlarger{\mathlarger{*}}}}}

\newcommand{\pointstosl}{\hookrightarrow}

\newcommand{\SL}{\mathit{SL}}
\newcommand{\FOL}{\mathit{FO}}

\newcommand{\N}{\mathbb{N}}
\newcommand{\Nplus}{\mathbb{N}_{>0}}
\newcommand{\Z}{\mathbb{Z}}

\newcommand{\Prog}{\mathcal{P}}

\newcommand{\ERROR}{\mathit{ERR}}

\newcommand{\rSC}[1]{Sect.~\ref{#1}}
\newcommand{\rF}[1]{Fig.~\ref{#1}}

\newcommand{\oldcomment}[1]{}

\newcommand{\oldcom}[1]{}

\newcommand{\codevar}[1]{\ensuremath{\mathtt{#1}}}
\newcommand{\codevarx}[1]{\protect{\ensuremath{\mathtt{#1}}}}

\newcommand{\pointsto}[1][]{\hookrightarrow_{#1}} %"points to" predicate used in set \PT 
\newcommand{\pointstorec}[2][]{\xhookrightarrow{#2}_{#1}} %recursive "points to" predicate used in set \PT 
\newcommand{\domain}{\mathop{domain}} %Domain of a function
\newcommand{\listpred}{\mathit{li}}
\newcommand{\listinv}{\mathit{l}}
\newcommand{\listsize}{\mathit{bs}}

\newcommand{\LI}{\mathit{LI}}
\newcommand{\PT}{\mathit{PT}}

\newcommand{\KB}{\mathit{KB}}

\newcommand{\LV}{\mathit{LV}}

\newcommand{\stateformula}[1]{\langle{#1}\rangle}

\newcommand{\alloc}[2]{\ensuremath{\llbracket{}#1,\,#2\rrbracket}}

\newcommand{\QFIA}{\mathit{QF\_IA}}
\newcommand{\AL}{\mathit{AL}}

\newcommand{\Vsym}{\mathcal{V}_{\mathit{sym}}}

\newcommand{\partialfunctionmap}[0]{\rightharpoonup}

\newcommand{\sizeOf}{\mathit{size}}

\newcommand{\slassignment}{\mathit{as}}
\newcommand{\slmemory}{\mathit{mem}}

\newcommand{\cmpFor}{\code{cmpF}}
\newcommand{\bodyFor}{\code{bodyF}}
\newcommand{\cmpWhile}{\code{cmpW}}
\newcommand{\bodyWhile}{\code{bodyW}}
\newcommand{\initPtr}{\code{initPtr}}

\newcommand{\GraphInitEntryZero}{A}
\newcommand{\GraphInitCmpZero}{B}
\newcommand{\GraphInitCmpOne}{C}
\newcommand{\GraphInitCmpOneRefTrue}{D}
\newcommand{\GraphInitCmpOneRefFalse}{E}
\newcommand{\GraphInitCmpTwoTrue}{F}
\newcommand{\GraphInitBodyZero}{G}
\newcommand{\GraphInitBodyOne}{H}
\newcommand{\GraphInitBodySix}{J}
\newcommand{\GraphInitBodyTen}{K}
\newcommand{\GraphInitCmpZeroA}{L}
\newcommand{\GraphInitCmpZeroB}{M}
\newcommand{\GraphInitCmpZeroGen}{O}
\newcommand{\GraphInitBodyStoreStart}{P}
\newcommand{\GraphInitBodyStoreRes}{Q}
\newcommand{\GraphInitCmpZeroC}{R}
\newcommand{\GraphTravEntryZero}{S}
\newcommand{\GraphTravCmpZeroTwo}{T}
\newcommand{\GraphTravBodyGetelemStart}{U}
\newcommand{\GraphTravBodyGetelemStartP}{U'}
\newcommand{\GraphTravBodyGetelemRes}{V}
\newcommand{\GraphTravBodyGetelemResP}{V'}
\newcommand{\GraphTravCmpZeroThree}{W}
\newcommand{\GraphTravBodyGetelemStartFinal}{X}
\newcommand{\GraphTravBodyGetelemResFinal}{Y}
\newcommand{\GraphTravCmpZeroFour}{W'}

\newcommand{\StateA}{s}
\newcommand{\StateB}{{s'}}
\newcommand{\StateGen}{{\overline{s}}}

\newcommand{\seqs}

\renewcommand{\arraystretch}{1.2}

\renewcommand{\implies}{\Rightarrow}

\DeclareMathAlphabet
{\mathttit}{OT1}{lmtt}{l}{it}

\newlist{caselist}{enumerate}{10}
\setlist[caselist]{wide,labelindent=0pt,itemsep=1ex,topsep=1ex,parsep=0pt,leftmargin=0pt,%
 labelwidth=*,labelsep=1ex,align=caselabel,label*=.\arabic*,}
\setlist[caselist,1]{label=\arabic*}
\newcommand\thecaselabeltext{}
\newcommand{\thecaselabelword}{Case}
\SetLabelAlign{caselabel}{\uline{\thecaselabelword~#1\thecaselabeltext}.\hfil}
\crefname{caselisti}{Case}{Cases}
\crefname{theorem}{Thm.}{Thms.}
\Crefname{theorem}{Theorem}{Theorems}
\crefname{lemma}{Lem.}{Lems.\ }
\Crefname{lemma}{Lemma}{Lemmas}
\crefname{subformula}{subformula}{subformulas}
\Crefname{subformula}{Subformula}{Subformulas}
\crefname{corollary}{Cor.}{Corollaries}
\Crefname{corollary}{Corollary}{Corollaries}
\crefname{definition}{Def.}{Def.}
\Crefname{definition}{Definition}{Definitions}
\crefname{proofcondition}{Cond.}{Conditions}
\Crefname{proofcondition}{Condition}{Definitions}
\crefname{appendix}{App.}{App.}
\Crefname{appendix}{Appendix}{Appendixes}
\crefname{page}{page}{pages}
\Crefname{page}{Page}{Pages}
\creflabelformat{subformula}{#2(#1)#3}
\crefname{condition}{Cond.}{Conditions}
\Crefname{condition}{Condition}{Definitions}
\crefname{section}{Sect.}{FIXME}
\Crefname{section}{Section}{FIXME}
\crefname{eqstatement}{}{}
\Crefname{eqstatement}{}{}
\newcounter{i}
\forLoop{1}{10}{i}{\crefalias{caselist\roman{i}}{caselisti}}

\newlist{labeledenumerate}{enumerate}{10}
\setlist[labeledenumerate]{leftmargin=*,labelindent=1em,topsep=5pt}

\tikzstyle{eval-edge}=[-stealth]
\tikzstyle{gen-edge}=[-stealth]
\tikzstyle{int-gen-edge}=[-stealth]
\tikzstyle{ca-edge}=[-stealth]
\tikzstyle{fs-edge}=[-stealth]
\tikzstyle{omit-edge}=[-stealth, dotted]

\usetikzlibrary{shapes,positioning,calc}
\usepackage{wrapfig}
\usetikzlibrary{shapes,arrows,calc}
\tikzstyle{ptr}  = [draw, -latex']
\tikzstyle{ptrvar} = [rectangle, draw, text height=3mm, text width=3mm, text centered, node distance=3cm, inner sep=0pt]
\tikzstyle{data} = [rectangle split, rectangle split horizontal, rectangle split parts=2, draw, text centered, text width=0.4cm]
\tikzstyle{mem16} = [rectangle split, rectangle split horizontal, rectangle split parts=16, draw, text centered, text width=0.4cm]

\tikzset{
    invisible/.style={opacity=0},
    visible on/.style={alt={#1{}{invisible}}},
    alt/.code args={<#1>#2#3}{%
      \alt<#1>{\pgfkeysalso{#2}}{\pgfkeysalso{#3}} % \pgfkeysalso doesn't change the path
    },
  }
\tikzset{onslide/.code args={<#1>#2}{%
  \only<#1>{\pgfkeysalso{#2}} 
}}

\pgfdeclaremetadecoration{middlezigzag}{straight}{
    \state{straight}[switch if less than=\pgfmetadecorationsegmentlength to final,
                  width=\pgfmetadecoratedpathlength/2 - \pgfmetadecorationsegmentlength/2 ,
                  next state=zigzag] {
        \decoration{curveto}
    }
    \state{zigzag}[width=\pgfmetadecorationsegmentlength,
                   next state=final] {
        \decoration{zigzag}
    }
    \state{final}{
        \decoration{curveto}
        \beforedecoration{\pgfpathmoveto{\pgfpointmetadecoratedpathfirst}}
    }
}

\tikzset{
    middle zigzag/.style={
        decorate,
        decoration={
            middlezigzag,
            meta-segment length=#1,
            segment length=0.5cm,
			amplitude=1.5pt
        }
    },
    middle zigzag/.default=1cm
}

\newcounter{cor-memory-safety}
\newcounter{thm-soundness-graph-construction-main-text}
\newcounter{thm-termination}

\makeatletter
\RequirePackage[bookmarks,unicode,colorlinks=true]{hyperref}%
   \def\@citecolor{blue}%
   \def\@urlcolor{blue}%
   \def\@linkcolor{blue}%

\def\orcidID#1{\smash{\href{http://orcid.org/#1}{\protect\raisebox{-1.25pt}{\protect\includegraphics{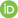}}}}}
\makeatother

\newcommand{\report}[1]{}
\newcommand{\paper}[1]{}
\newcommand{\arxiv}[1]{#1}

 \arxiv{
  \setlength{\textwidth}{125mm}
   \setlength{\textheight}{198mm}
 }

\title{Proving Termination of \tool{C} Programs with Lists\thanks{funded by the Deutsche Forschungsgemeinschaft (DFG, German Research Foundation) -
235950644 (Project GI 274/6-2)}} 
\author{Jera Hensel\orcidID{0000-0003-2852-9830} \and J\"urgen Giesl\orcidID{0000-0003-0283-8520}}
\institute{LuFG Informatik 2, RWTH Aachen University, Germany}
\begin{document}

\maketitle

\begin{abstract}There are many techniques and tools to prove termination
of \sfC{} programs, but up to now 
these tools were not very powerful
for fully automated
termination proofs of 
programs whose
termination depends on recursive data structures like lists.
We present the
first approach that extends powerful techniques for termination analysis of \sfC{} programs (with
memory allocation and explicit pointer arithmetic) to lists.
\end{abstract}
\section{Introduction}
\label{sect:introduction}

In \cite{ASV19,JLAMP18,TACAS22,LLVM-JAR17}, we intro\-duced an approach for
automatic termination analysis of
\sfC{} that also handles programs whose termination relies on the relation
between
allocated memory addresses and the data stored at such addresses.
This approach is implemented in our tool \aprove{} \cite{JAR17-AProVE}.
Instead of analyzing \sfC{} directly, \aprove{} compiles the program to \LLVM{} code using \tool{Clang} \cite{Clang}.
Then it constructs a (finite) symbolic execution graph (SEG) such that every program run
corresponds to a path through
the SEG.
\aprove{} proves memory safety during the construction of the SEG to ensure
absence of undefined behavior (which would also
allow non-termination). Afterwards,
the SEG is transformed into an integer transition system (ITS) such
that all paths through the SEG (and hence, the \sfC{} program)
are terminating if the ITS is terminating. To analyze termination of the
 ITS, 
\aprove{} applies standard techniques and calls 
 the tools \tool{T2} \cite{T2} and
 \tool{LoAT} \cite{loat,LoATIJCAR22} to detect non-termi\-nation of ITSs.
However, like other termination tools  for \sfC{},
\begin{wrapfigure}[14]{r}{5.3cm}
\vspace*{-.7cm}
\scriptsize
\hspace*{-.2cm}\begin{minipage}{5.2cm}
  \begin{Verbatim}[commandchars=\\\{\}]
struct list \{
  unsigned int value;
  struct list* next;   \};

int main() \{
  \textcolor{black!55}{// initialize length}
  unsigned int n = nondet_uint();
  \textcolor{black!55}{// initialize list of length n}
  struct list* tail = NULL;
  struct list* curr;
  for (unsigned int k = 0; k < n; k++) \{
    curr = malloc(sizeof(struct list));
    curr->value = nondet_uint();
    curr->next = tail;
    tail = curr;                       \}
  \textcolor{black!55}{// traverse list}
  struct list* ptr = tail;
  while(ptr != NULL) \{
    ptr = *((struct list**)((void*)ptr +
          offsetof(struct list, next)));\}\}
  \end{Verbatim}
\end{minipage}
\end{wrapfigure}
\noindent
  up to now
   \aprove{} supported dynamic data structures only in a very restricted way.

In this paper, we introduce a novel technique to analyze
\sfC{} programs on lists. In
the\linebreak program on the right,
 \code{nondet\_uint}  returns a random unsigned integer. The \code{for}
loop creates a list of \code{n} random numbers if $\code{n}> 0$.\linebreak The
\code{while} loop traverses this list
via poin\-ter arithmetic: Starting with
\code{tail}, it com\-putes the address of the \code{next} field of the\linebreak
current element
by adding the offset of the \code{next} field within a \code{list}  to the address of the
current  \code{list} and   dereferencing the com\-puted \pagebreak address (i.e., the content of the 
\code{next} field).
This is done by \code{offsetof},  defined in
the \sfC{} library \code{stddef.h}.\footnote{Note that
\code{ptr + n} increases \code{ptr} by
  \code{n} times the size of the type \code{*ptr}.
 As we want to increase \code{ptr} by a number of bytes
 and \code{ptr} is not an \code{i8} pointer, we first  cast \code{ptr} to \code{void*}. Then \code{((void*)ptr + offsetof(struct list, next))}  contains the \code{next} pointer, so we cast our
  computed address to \code{struct list**} before dereferencing it.}
Since the list is acyclic and the \code{next}
poin\-ter of its last element is the null pointer, list
traversal always termi\-nates.
Of course, the \code{while} loop could also traverse the list\linebreak via \code{ptr =
  ptr->next}, but in \sfC{},  memory accesses
 can be combined with pointer arithmetic. This example contains both the
 access via \code{curr->next} (when initializing the list) and pointer
  arithmetic (when traversing the list).

We present a new general technique
to infer \emph{list invariants} via symbolic execution,
which express
all properties that are crucial for
memory safety and termination.
In our example, the list invariant contains the 
information that dereferencing the \code{next} pointer
in the \code{while} loop
is  safe and that one finally reaches
the null pointer.
In general, our novel list invariants allow us to abstract from detailed information about
lists (e.g., about their intermediate elements)
such that abstract states with ``similar'' lists can be merged and generalized during the
symbolic execution in order to obtain finite SEGs. At the
same time, list invariants 
express enough information about the lists (e.g., their length, their start address, etc.)
such that memory safety and termination can still
be proved.

We define the abstract states used for symbolic execution in Sect.\ \ref{sect:domain}.
In Sect.\ \ref{sect:symbexec}, after recapitulating the construction of
SEGs, we adapt our techniques for merging and generalizing states from \cite{LLVM-JAR17}
to infer list invariants.
Moreover, we
adapt\linebreak those rules for symbolic 
execution that are affected by introducing list invariants.
\rSC{sect:Proving Termination} discusses the generation of ITSs and the soundness of our approach. 
 Sect.\ \ref{sect:evaluation}
gives an overview on related work. Moreover, we
evaluate the implementation of our approach in the tool
\aprove{} using benchmark sets from
\emph{SV-COMP} \cite{SVCOMP22} and the \emph{Termination Competition}
\cite{TermComp}.
All proofs can be found
in\paper{ \cite{report}.}\arxiv{ App.\ \ref{app:proofs}.}

\vspace*{-.05cm}

\paragraph{Limitations}
To ease the presentation, in this paper we treat integer types as unbounded.
Moreover, we assume that a program consists of a
single non-recursive function
and that values may be stored at any address.
Our approach can also deal with bitvectors, data
alignments, and programs with arbitrary many (possibly recursive) functions, see
\cite{ASV19,JLAMP18,LLVM-JAR17} for details.
However,  so far only lists without sharing can be handled by our new technique. Extending it to more
general recursive data structures is one of the main challenges for future
work.

\section{Abstract States for Symbolic Execution}\label{sect:domain}

The \LLVM{} code
for the \code{for} loop is given on the next page. It is equivalent to the code produced by
\tool{Clang} without optimizations on a 64-bit computer.
We explain it in detail in \cref{sect:symbexec}.
To ease readability, we omitted instructions and keywords that are irrelevant for our
presentation, renamed variables, and wrote \code{list}\linebreak instead of \code{struct.list}.
Moreover, 
we gave the \sfC{} instructions (\textcolor{black!55}{in gray}) before 
the corresponding \LLVM{} code.
The code consists of several \emph{basic blocks} including \cmpFor{}
 and \bodyFor{} 
(corresponding  to the loop \pagebreak
comparison and body).

\begin{wrapfigure}[20]{r}{6cm}
\begin{boxedminipage}{6cm}
\scriptsize
\begin{Verbatim}[commandchars=\\\{\}]
list = type \{ i32, list* \}

define i32 @main() \{ ...
\cmpFor:
  \tiny{\textcolor{black!55}{k < n}}
  0: k = load i32, i32* k_ad
  1: kltn = icmp ult i32 k, n
  2: br i1 kltn, label \bodyFor, label \initPtr
\bodyFor:
  \tiny{\textcolor{black!55}{curr = malloc(sizeof(struct list));}}
  0: mem = call i8* @malloc(i64 16)
  1: curr = bitcast i8* mem to list*
  \tiny{\textcolor{black!55}{curr->value = nondet_uint();}}
  2: nondet = call i32 @nondet_uint()
  3: curr_val = getelementptr list,
                list* curr, i32 0, i32 0
  4: store i32 nondet, i32* curr_val
  \tiny{\textcolor{black!55}{curr->next = tail;}}
  5: tail = load list*, list** tail_ptr
  6: curr_next = getelementptr list,
                 list* curr, i32 0, i32 1
  7: store list* tail, list** curr_next
  \tiny{\textcolor{black!55}{tail = curr;}}
  8: store list* curr, list** tail_ptr
  \tiny{\textcolor{black!55}{k++}}
  9: kinc = add i32 k, 1
  10:store i32 kinc, i32* k_ad
  11:br label \cmpFor
  ...                                     \}
\end{Verbatim}
\end{boxedminipage}
\vspace{-.05cm}
\end{wrapfigure}
We now recapitulate  the \emph{abstract\linebreak states}
of
\cite{LLVM-JAR17}  used for symbolic
execu\-tion and extend  them by a component $\LI$ for list
invariants, i.e., they  have the form
$((\codevar{b}, i), \LV, \AL, \PT, \LI, \KB)$. The first component is a \emph{program position}
(\codevar{b}, $i$), indicating that instruction\linebreak $i$ of block \codevar{b} is executed
next. $\PPos \subseteq (\Blks \times \N)$ is the set of all program positions, and
$\Blks$ are all basic blocks.

The second component is a partial injective
function $\LV\colon\Ids\partialfunctionmap\Vsym$, which maps \emph{\underline{l}ocal program \underline{v}ariables} $\Ids$ of the program
$\PP$ to
an infinite set  $\Vsym$ of symbolic variables with $\Vsym
\cap \Ids = \emptyset$.
 We  identify $\LV$ with the set of equations $\{ \code{x} = \LV(\code{x}) \mid
\code{x} \in \domain(\LV)\}$ and we
often extend $\LV$ to a function from $\Ids \uplus \N$ to $\Vsym \uplus \N$ by defining
$\LV(n) = n$ for all $n \in \N$.

The third component of each state is a set $\AL$ of (bytewise) allocations  $\alloc{v_1}{v_2}$ with
$v_1, v_2 \in \Vsym$, which indicate that $v_1 \le v_2$
and that all addresses between $v_1$ and $v_2$ have been allocated. 
We require any two  entries $\alloc{v_1}{v_2}$ and $\alloc{w_1}{w_2}$ from $\AL$
with
$v_1 \neq w_1$ or $v_2 \neq w_2$ to be disjoint. 

The fourth and fifth components  $\PT$ and $\LI$ model the memory contents. $\PT$ contains
``\underline{p}oints-\underline{t}o'' entries of the form $v_1 \pointsto[\codevar{ty}] v_2$ where $v_1,v_2 \in \Vsym$ and
$\codevar{ty}$ is an \LLVM{} type, meaning that the address $v_1$ of type $\codevar{ty}$
\underline{p}oints \underline{t}o $v_2$. 
In contrast, the set $\LI$ of \emph{\underline{l}ist \underline{i}nvariants} (which is new
compared to \cite{LLVM-JAR17}) does not 
describe pointwise memory contents but contains invariants
$v_{\mathit{ad}} \pointstorec[\codevar{ty}]{v_\ell} [(\mathit{off}_i: \codevar{ty}_i:\linebreak
  {v_i..\hat{v}_i})]_{i=1}^n$ where $n\!\in\!\N_{>0}$,
  $v_{\mathit{ad}},v_\ell,v_i,\hat{v}_i \in
\Vsym$, $\mathit{off}_i \in \N$
for all $1 \leq i \leq n$, $\codevar{ty}$ and
$\codevar{ty}_i$ are \LLVM{} types for all $1 \leq i \leq n$, and there is exactly
one ``recursive field''
$1 \leq j \leq n$ such that $\code{ty}_j = \code{ty*}$.\footnote{
Soundness of our approach is not affected if there are other recursive fields, but our
symbolic execution technique for list traversal
on list invariants in Sect.\ \ref{sect:List Traversal}
can only be applied if  the traversal is done along field $j$.
}
Such an invariant represents a
\code{struct} \code{ty} with $n$ fields that corresponds to a recursively defined list of
length $v_\ell$.
Here, $v_{\mathit{ad}}$ points to the first list element, the $i$-th field starts at address $v_{\mathit{ad}} + \mathit{off}_i$ (i.e., with
offset $\mathit{off}_i$)\footnote{The field offsets can be computed using the data
  layout string in the \LLVM{} program.}
and has type $\codevar{ty}_i$,
and the values of the $i$-th fields of
the first and last list\linebreak element
are $v_i$ and $\hat{v}_i$, respectively.
For example, the following list invariant \eqref{example list invariant} represents all lists of length $x_{\ell}$
and type \code{list} whose elements store a 32-bit integer in their first field and the pointer to the next element in their second field\linebreak
with offset $8$. The first list element starts at address $x_\code{mem}$,
the second  starts at ad\-dress $x_\code{next}$, and the last element
contains the null pointer.
Moreover, the first  ele-\linebreak ment stores the integer value $x_\code{nd}$ and the last list
element stores  the \pagebreak integer $\hat{x}_\code{nd}$. 
\begin{equation}
  \label{example list invariant}
x_{\code{mem}} \pointstorec[\code{list}]{x_{\ell}} [(0: \code{i32}:
  {x_\code{nd}}..\hat{x}_{\code{nd}}), (8: \code{list*}: x_{\code{next}}..0)]
\end{equation}
\noindent
For example, this invariant represents the list with the allocation
$\alloc{x_{\code{mem}}}{x_{\code{mem}}+15}$, where the first four bytes store the integer $5$
and  the last eight bytes store the pointer $x_{\code{next}}$,
and the
allocation $\alloc{x_{\code{next}}}{x_{\code{next}}+15}$, where  the first four bytes
store the integer $2$ and the last eight bytes store the null pointer (i.e., the address
$0$). Here,  we have $x_{\ell} = 2$.
 Sect.\ \ref{sect:List Traversal} will show that the expressiveness of our list invariants is indeed needed to prove termination
of programs that traverse a
list.

The last component of a state is a \emph{\underline{k}nowledge \underline{b}ase} $\KB$
of quantifier-free first-order formulas that express
integer arithmetic properties of $\Vsym$. We identify \emph{sets} of first-order formulas $\{\varphi_1, \ldots, \varphi_m\}$ with their conjunction $\varphi_1 \wedge \ldots \wedge \varphi_m$.

A special state $\ERROR$ is reached if we cannot prove absence of undefined
beha-\linebreak  vior (e.g., if memory safety might be violated
by dereferencing the null pointer).

As an example, the following abstract state \eqref{example state}
represents concrete states at the beginning of the
block \cmpFor{}, where the program variable \code{curr} is assigned the symbolic variable
$x_{\code{mem}}$, the allocation $\alloc{x_{\code{k\_ad}}}{x_{\code{k\_ad}}^\mathit{end}}$
consisting of $4$ bytes stores the value $x_{\code{kinc}}$,
and $x_{\code{mem}}$ points to
the first element of a list of length $x_\ell$ (equal to $x_{\code{kinc}}$)
that satisfies the list invariant \eqref{example list invariant}. (This state will later
be obtained during the symbolic execution, see State $O$ in \cref{fig:ForLoopMerging} in
Sect.\ \ref{sect:inferring_inv}.)
\begin{equation}
  \label{example state}
\hspace*{-5cm}  \parbox{6cm}{
 \begin{tikzpicture}[node distance = \ydist and \xdist]
\scriptsize
\tikzstyle{state}=[inner sep=2pt, font=\scriptsize, draw]
\node[state, align=left] (1) 
{$(\cmpFor, 0),\;
  \{ \code{curr} = x_{\code{mem}}, \,
     \code{kinc} = x_{\code{kinc}}, \,
     ...\}, \;
  \{ \alloc{x_{\code{k\_ad}}}{x_{\code{k\_ad}}^\mathit{end}}, \,
     ...\}, \;
  \{ x_{\code{k\_ad}} \pointsto[\code{i32}] x_{\code{kinc}}, \,
     ...\}, \;$\\
 $\{ x_{\code{mem}} \pointstorec[\code{list}]{x_{\ell}} [(0: \code{i32}: x_{\code{nd}}..\hat{x}_{\code{nd}}), (8: \code{list*}: x_{\code{next}}..0)]\}, \;
  \{ x_{\code{k\_ad}}^\mathit{end} = x_{\code{k\_ad}} + 3, \,
     x_{\ell} = x_{\code{kinc}}, \,
     ...\}
 $
};
\end{tikzpicture}}
\end{equation}

A state $s= (\Pos, \LV, \AL, \PT, \LI, \KB)$ is \emph{represented by a formula}
$\stateformula{s}$ which contains $\KB$ and
encodes 
$\AL$, $\PT$, and $\LI$ in first-order logic. This allows us to use
standard SMT solving for all reasoning during the construction of the SEG.
Moreover, $\stateformula{s}$ is
also used for the generation of the ITS afterwards.
The encoding of $\AL$ and $\PT$ is as in
\cite{LLVM-JAR17}, see\arxiv{ App.\ \ref{Separation Logic Semantics of Abstract
    States}:}\paper{ \cite{report}:}
$\stateformula{s}$  contains formulas which express that allocated addresses are positive, that allocations represent
disjoint memory areas, that equal addresses point to equal values, and that addresses
are different if they point to different values.
For each element of $\LI$, we add the following new formulas to $\stateformula{s}$ which express
that the list length $v_\ell$ is
$\geq 1$ and the ad-\linebreak{}dress $v_{\mathit{ad}}$ of the first element is not null. If
 $v_\ell = 1$, 
then the values $v_i$ and $\hat{v}_i$ 
of the fields in the first and the last element are equal.
If  $v_\ell  \geq 2$, then the \code{next} pointer $v_j$
 in the first element must not be null. Finally, if there is a field whose
 values
$v_k$ and $\hat{v}_k$ 
differ in the first and the last element, then the length $v_\ell$ must be
$\geq 2$.\vspace*{-.1cm}
\[
\mbox{ $\begin{array}{@{\hspace*{-.15cm}}l}
     \{v_\ell \geq 1 \wedge v_{\mathit{ad}} \geq 1 
  \mid
(v_{\mathit{ad}} \pointstorec[\codevar{ty}]{v_\ell}
  [(\mathit{off}_i: \codevar{ty}_i: {v_i..\hat{v}_i})]_{i=1}^n) \in \LI \} \; \cup \\
 \{\bigwedge_{i=1}^n v_i = \hat{v}_i 
   \mid (v_{\mathit{ad}} \pointstorec[\codevar{ty}]{v_\ell}
  [(\mathit{off}_i: \codevar{ty}_i: {v_i..\hat{v}_i})]_{i=1}^n) \in \LI \mbox{ and } \models \, 
   \stateformula{s} \implies v_\ell = 1\} \; \cup\\
 \mbox{\small $\{v_j \geq 1 
   \mid
(v_{\mathit{ad}} \pointstorec[\codevar{ty}]{v_\ell}
       [(\mathit{off}_i: \codevar{ty}_i: {v_i..\hat{v}_i})]_{i=1}^n) \in \LI
       \mbox{ with }  \codevar{ty}_j = \codevar{ty*} \mbox{ and } \models \, 
   \stateformula{s} \implies v_\ell \geq 2\} \; \cup$}\\
 \mbox{\small $\{v_\ell \geq 2 
   \mid (v_{\mathit{ad}} \pointstorec[\codevar{ty}]{v_\ell}
  [(\mathit{off}_i: \codevar{ty}_i: {v_i..\hat{v}_i})]_{i=1}^n) \in \LI \mbox{ and }
  \exists k\!\in\!\Nplus, k \leq n, \, \mbox{s.t.} \, \models \, 
   \stateformula{s} \implies v_k \neq \hat{v}_k\}$}
  \end{array}$}
\]

In \emph{concrete}
 states $c$, all values of variables and
memory contents are deter\-mined uniquely.
To ease the formalization, we assume
that all integers are unsigned and refer to \cite{JLAMP18} for the general
case.
 So
 for all $v \in \Vsym(c)$ (i.e., all $v \in \Vsym$\linebreak occurring in $c$)
 we have
 $\models \stateformula{c} \implies v = n$ for some
  $n \in \N$. Moreover, 
 here $\PT$ only contains information about
 allocated addresses   and $\LI = \varnothing$ since the abstract knowledge in list invariants
 is unnecessary
 if all memory contents \pagebreak are known.

 For instance,  all concrete states  $((\cmpFor,0),\LV,\AL, \PT,\emptyset,\KB)$
 represented by the state
 \eqref{example state}
contain $\ell$ allocations of 16 bytes for some $\ell  \geq 1$,
where in the first four bytes a  32-bit integer is stored and in the last eight
bytes the address of the next allocation (or 0, in case of the last allocation) is stored.

See\arxiv{ App.\ \ref{Separation Logic Semantics of Abstract States}}\paper{ \cite{report}} for a formal definition to
determine 
 which concrete states are represented by a state
$s$. To this end, as in \cite{LLVM-JAR17} we define a \emph{separation logic} formula $\stateformula{s}_{\SL}$
which also encodes the knowledge contained in the memory
components of states. To extend this formula to list invariants, we use 
a fragment similar to
\emph{quantitative} separation logic \cite{QSL}, extending conventional separation logic
by list predicates.
 For any state $s$, we have $\models  \stateformula{s}_{\SL}
\implies \stateformula{s}$, i.e., $\stateformula{s}$ is a weakened version of
$\stateformula{s}_{\SL}$ that we use for symbolic execution and
the termination proof.

\vspace*{-.2cm}

\section{Symbolic Execution with List Invariants}
\label{sect:symbexec}

\vspace*{-.1cm}

We first recapitulate the construction of SEGs.
Then, \cref{sect:inferring_inv} extends the technique
for \emph{merging} and generalization of
states from \cite{LLVM-JAR17} to infer list invariants.
Finally,
we adapt  the rules for symbolic execution to list invariants in
\cref{sect:adapting_inv}.

\footnotesize
\newcommand{\ydist}{0.3cm}
\newcommand{\xdist}{0.5cm}
\newcommand{\setfont}{\scriptsize}
\newcommand{\labelxshift}{-1pt}

\begin{figure}[t]
\centering
\begin{tikzpicture}[node distance = \ydist and \xdist]
\scriptsize
\def\widetwidth{6.2cm}
\def\smalltwidth{4.75cm}
\def\fulllinewidhth{11cm}
\def\edgenodedist{0.09cm}
\def\FirstIndentwidth{3cm}
\def\SecondIndentwidth{2cm}
\tikzset{invisible/.style={opacity=0}}
\tikzstyle{state}=[
						   %minimum size=10pt,
                           %fill=white,
                           %shape=rectangle,
                           %text=black,
                           inner sep=2pt,
                           font=\scriptsize,
                           draw]

\node[state, label={[xshift=\labelxshift]2:$\GraphInitEntryZero$}] (1) 
{$(\code{entry}, 0),\;
  \emptyset, \;
  \emptyset, \;
  \emptyset, \;
  \emptyset, \;
  \emptyset
 $
};

\node[state, below=of 1, align=left, label={[xshift=\labelxshift]2:$\GraphInitCmpZero$}] (2) 
{$(\cmpFor, 0),\;
  \{ \code{n} = v_{\code{n}}, \,
     \code{tail\_ptr} = v_{\code{tp}}, \,
     \code{k\_ad} = v_{\code{k\_ad}}, \,
     ...\}, \;
  \{ \alloc{v_{\code{tp}}}{v_{\code{tp}}^\mathit{end}}, \,
     \alloc{v_{\code{k\_ad}}}{v_{\code{k\_ad}}^\mathit{end}} \},
 $\\
 $
  \{ v_{\code{tp}} \pointsto[\code{list*}] 0, \,
     v_{\code{k\_ad}} \pointsto[\code{i32}] 0 \}, \;
  \emptyset, \;
  \{ v_{\code{tp}}^\mathit{end} = v_{\code{tp}} + 7, \,
     v_{\code{k\_ad}}^\mathit{end} = v_{\code{k\_ad}} + 3, \,
     ... \}
 $
};

\node[state, below=of 2, align=left, label={[xshift=\labelxshift]2:$\GraphInitCmpOne$}] (3) 
{$(\cmpFor, 1),\;
  \{ \code{k} = 0, \,
     ...\}, \;
  \AL^{\GraphInitCmpZero}, \;
  \PT^{\GraphInitCmpZero}, \;
  \emptyset, \;
  \KB^{\GraphInitCmpZero}
 $
};

\node[state, below left=.3cm and -2cm of 3, align=left, label={[xshift=\labelxshift]2:$\GraphInitCmpOneRefTrue$}] (3a) 
{$(\cmpFor, 1),\;
  \{ \code{k} = 0, \,
     ...\}, \;
  \AL^{\GraphInitCmpZero},
 $\\
 $
  \PT^{\GraphInitCmpZero}, \;
  \emptyset, \;
  \{ v_\code{n} > 0, \,
     ... \}
 $
};

\node[state, below right=.3cm and -2cm of 3, align=left, label={[xshift=\labelxshift]2:$\GraphInitCmpOneRefFalse$}] (3b) 
{$(\cmpFor, 1),\;
  \{ \code{k} = 0, \,
     ...\}, \;
  \AL^{\GraphInitCmpZero},
 $\\
 $
  \PT^{\GraphInitCmpZero}, \;
  \emptyset, \;
  \{ v_\code{n} \leq 0, \,
     ... \}
 $
};

\node[state, below right=.3cm and -4cm of 3a, align=left, label={[xshift=\labelxshift]2:$\GraphInitCmpTwoTrue$}] (4a) 
{$(\cmpFor, 2),\;
  \{ \code{kltn} = 1, \,
     ...\}, \;
  \AL^{\GraphInitCmpZero}, \;
  \PT^{\GraphInitCmpZero}, \;
  \emptyset, \;
  \KB^{\GraphInitCmpOneRefTrue}
 $
};

\node[below=of 3b] (4b) 
{$...$};

\node[state, below right=.3cm and -4.5cm of 4a, align=left, label={[xshift=\labelxshift]2:$\GraphInitBodyZero$}] (5) 
{$(\bodyFor, 0),\;
  \LV^{\GraphInitCmpTwoTrue}, \;
  \AL^{\GraphInitCmpZero}, \;
  \PT^{\GraphInitCmpZero}, \;
  \emptyset, \;
  \KB^{\GraphInitCmpOneRefTrue}
 $
};

\node[state, below right=.3cm and -6.5cm of 5, align=left, label={[xshift=\labelxshift]2:$\GraphInitBodyOne$}] (6) 
{$(\bodyFor, 1),\;
  \{ \code{mem} = v_{\code{mem}}, \,
     ...\}, \;
  \{ \alloc{v_{\code{mem}}}{v_{\code{mem}}^\mathit{end}}, \,
     ... \}, \;
  \PT^{\GraphInitCmpZero}, \;
  \emptyset, \;
  \{ v_{\code{mem}}^\mathit{end} = v_{\code{mem}} + 15, \,
     ... \}
 $
};

\node[state, below=of 6, align=left, label={[xshift=\labelxshift]2:$\GraphInitBodySix$}] (11) 
{$(\bodyFor, 7),\;
  \{ \code{curr} = v_{\code{mem}}, \,
     \code{nondet} = v_{\code{nd}}, \,
     \code{curr\_val} = v_{\code{mem}}, \,
     \code{tail} = 0, \,
     \code{curr\_next} = v_{\code{cn}}, \,
     ...\},$\\
 $\AL^{\GraphInitBodyOne}, \;
  \{ v_{\code{mem}} \pointsto[\code{i32}] v_{\code{nd}}, \,
     ... \}, \;
  \emptyset, \;
  \{ v_{\code{cn}} = v_{\code{mem}} + 8, \,
     ... \}
 $
};

\node[state, below=of 11, align=left, label={[xshift=\labelxshift]2:$\GraphInitBodyTen$}] (15) 
{$(\bodyFor, 11),\;
  \{ \code{kinc} = 1, \;
     ...\}, \,
  \AL^{\GraphInitBodyOne}, \;
  \{ v_{\code{cn}} \pointsto[\code{list*}] 0, \,
     v_{\code{tp}} \pointsto[\code{list*}] v_{\code{mem}}, \,
     v_{\code{k\_ad}} \pointsto[\code{i32}] 1, \,
     ... \}, \;
  \emptyset, \;
  \KB^{\GraphInitBodySix}
 $
};

\draw[omit-edge] (1)  --  (2);
\draw[eval-edge] (2)  --  (3);
\draw[eval-edge] (3)  --  (3a);
\draw[eval-edge] (3)  --  (3b);
\draw[eval-edge] (3a) --  (4a);
\draw[eval-edge] (3b) --  (4b);
\draw[eval-edge] (4a) --  (5);
\draw[eval-edge] (5)  --  (6);
\draw[omit-edge] (6)  --  (11);
\draw[omit-edge] (11) --  (15);

\end{tikzpicture}
\vspace*{-.1cm}
\caption{SEG for the First Iteration of the \code{for} Loop}
\label{fig:ForLoopFirstIteration}
\vspace*{-.5cm}
\end{figure}
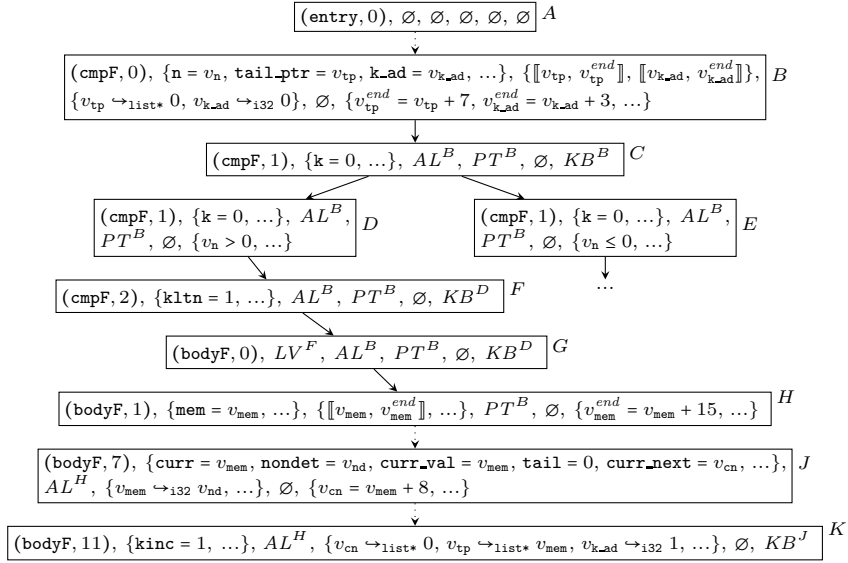
\normalsize
Our symbolic execution starts with a state $\GraphInitEntryZero$
at the first instruction of the first block
(called \code{entry} in
our example).
Fig.\ \ref{fig:ForLoopFirstIteration} shows 
the first iteration of the \code{for} loop.
Dotted arrows indicate that we omitted some symbolic execution steps. For every
state, we perform symbolic execution by applying 
the corresponding inference rule as in \cite{LLVM-JAR17}
  to compute
its successor state(s) and repeat this until all paths end in return states.
We call an SEG with this property \emph{complete}.

As an example, we recapitulate the inference rule for the \code{load} instruction in the case where a value is loaded from allocated and initialized memory.
It loads the value of type \code{ty} that is stored at the address \code{ad} to the program
variable \code{x}.
Let $\sizeOf(\codevar{ty})$ denote \vspace*{-.3cm}\pagebreak the size of \codevar{ty} in bytes for any \LLVM{} type \codevar{ty}.
If we can prove that there is an
allocation $\alloc{v_1}{v_2}$ containing all addresses $\LV(\code{ad}), \ldots,
\LV(\code{ad})+\sizeOf(\code{ty})-1$ and
there exists an entry $(w_1 \pointsto[\code{ty}] w_2) \in \PT$ such that $w_1$ is equal to
the address $\LV(\code{ad})$ loaded from, then we
transform the state $s$ at position $p = (\codevar{b},i)$ to a state $s'$ at position
$p^+ = (\codevar{b},i+1)$. In $s'$, a fresh symbolic variable $w$ is assigned to \code{x}
and $w = w_2$ is added to
$\KB$.
We write $\LV[\code{x} := w]$ for the function where
$\LV[\code{x} := w](\code{x}) = w$ and
$\LV[\code{x} := w](\code{y}) = \LV(\code{y})$
for all $\code{y} \neq \code{x}$.

\vspace*{.15cm}
\noindent
\fbox{
\begin{minipage}{11.8cm}
\label{rule:load}
\small
\mbox{\small \textbf{\hspace*{-.2cm}\code{load} from initialized allocated memory ($p\!:\!\!$ ``\code{x = load ty, ty* ad}'', $\!\codevar{x},\codevar{ad} \in \Ids\!$)}}\\
\vspace*{-.2cm}\\
\centerline{$\frac{\parbox{6.6cm}{\centerline{
$s = (\Pos, \; \LV, \; \AL, \; \PT, \; \LI, \; \KB)$}\vspace*{.1cm}}}
{\parbox{7.4cm}{\vspace*{.1cm} \centerline{
$s' = (\Pos^+, \; \LV[\code{x}:=w], \; \AL, \; \PT, \; \LI, \; \KB \cup \{w = w_2\})$}}}\;\;\;\;$
\mbox{if  $w \in \Vsym$ is fresh and}} \vspace*{-.15cm}
{\small
\begin{itemize}
\item[$\bullet$] there is $\alloc{v_1}{v_2} \in \AL$ with $\models \, \stateformula{s} \implies (v_1 \leq \LV(\code{ad}) \wedge \LV(\code{ad})+\sizeOf(\code{ty})-1 \leq v_2)$
\item[$\bullet$] there are $w_1, w_2 \in \Vsym$ with $\models \, \stateformula{s} \implies (\LV(\code{ad}) = w_1)$ and $(w_1 \pointsto[\code{ty}] w_2) \in \PT$
\end{itemize}}
\end{minipage}}
\vspace*{.15cm}

In our example,
the \code{entry} block comprises the first three lines of the \sfC{} program and
the initialization of the pointer to the loop variable $\code{k}$:
First,  a non-deterministic unsigned integer is assigned to \code{n}, i.e.,
 $(\code{n} = v_\code{n}) \in \LV^{\GraphInitCmpZero}$, where
 $v_\code{n}$ is not restricted.
  Moreover,  memory for the
pointers \code{tail\_ptr} and \code{k\_ad} is allocated and they
point to \code{tail = NULL} and \code{k = 0}, respectively
($\code{tail\_ptr} = v_{\code{tp}}$
and $\code{k\_ad} = v_{\code{k\_ad}}$ 
with $(v_{\code{tp}} \pointsto[\code{list*}] 0),
(v_{\code{k\_ad}} \pointsto[\code{i32}] 0)
\in \PT^\GraphInitCmpZero$). For
simplicity, in \cref{fig:ForLoopFirstIteration}
we use concrete values directly instead of introducing fresh variables
for them. Since we assume a 64-bit architecture, 
\code{tail\_ptr}'s allocation contains 8 bytes. For the integer value of \code{k},
only 4 bytes are allocated. Alignments and pointer sizes depend on the memory layout and are given
in the \LLVM{} program.

State $\GraphInitCmpOne$ results from $\GraphInitCmpZero$ by evaluating the \code{load}
instruction at
$(\cmpFor{}, 0)$, see the above \code{load} rule. 
Thus, there is an
\emph{evaluation edge}
 from
 $\GraphInitCmpZero$ to
$\GraphInitCmpOne$.

The next instruction is an \underline{i}nteger
\underline{c}o\underline{mp}arison whose Boolean return value depends on whether the
\underline{u}nsigned value of \code{k} is \underline{l}ess \underline{t}han the one
of \code{n}. If we cannot decide the validity of a comparison, we refine
the state into two successor states.
Thus, the states $\GraphInitCmpOneRefTrue$ and $\GraphInitCmpOneRefFalse$ 
(with
$(v_\code{n} > 0) \in \KB^\GraphInitCmpOneRefTrue$ and
 $(v_\code{n} \leq 0) \in\KB^\GraphInitCmpOneRefFalse$) are
reached by \emph{refinement edges} from State $\GraphInitCmpOne$. Evaluating
$\GraphInitCmpOneRefTrue$ yields $\code{kltn} = 1$ in $\GraphInitCmpTwoTrue$. Therefore, the 
\underline{br}anch instruction leads to the block \bodyFor{} in State
$\GraphInitBodyZero$. State $\GraphInitCmpOneRefFalse$ is evaluated to a state with
$\code{kltn} = 0$. This path branches to the block \initPtr{} and terminates quickly as
\code{tail\_ptr} points to an empty list.

The instruction at $(\bodyFor{}, 0)$ allocates 16 bytes of memory
starting at $v_{\code{mem}}$
in State $\GraphInitBodyOne$.
The next instruction casts the pointer to the allocation from \code{i8*} to
\code{list*} and assigns it to \code{curr}. 
Now the allocated area can be treated as a list element.
Then \code{nondet\_uint()} is invoked to assign a 32-bit integer to \code{nondet}.
The \code{getelementptr} instruction computes the address of the integer field of the list
element by indexing this field (the second \code{i32 0}) based on the start address
(\code{curr}). The first index (\code{i32 0}) specifies that a field of
\code{*curr} itself is computed and not of another list stored after \code{*curr}. Since 
the address of the integer value of the list element coincides with the start address of
the list element, this instruction assigns $v_\code{mem}$ to \code{curr\_val}. 
Afterwards, the value of \code{nondet} is stored at \code{curr\_val} ($v_{\code{mem}}
\pointsto[\code{i32}] v_\code{nd}$), the value \code{0} stored at $v_\code{tp}$ is loaded
to \code{tail}, and a second \code{getelementptr} instruction computes the address of the
recursive field of the current list element ($v_{\code{cn}} = v_{\code{mem}} + 8$) and
assigns it to \code{curr\_next}, leading to 
state $\GraphInitBodySix$.
In the path to $\GraphInitBodyTen$,  the values of \code{tail} and \code{curr} are stored
at \code{curr\_next} and \code{tail\_ptr},  respectively ($v_{\code{cn}}
\pointsto[\code{list*}] 0$, $v_{\code{tp}} \pointsto[\code{list*}]
v_{\code{mem}}$). Finally, the incremented value of \code{k} is assigned to \code{kinc}
and stored at \code{k\_ad} ($v_{\code{k\_ad}} \pointsto[\code{i32}] 1$).

To ensure a finite graph construction, when a program position is reached for the second
time, we try to merge the states at this position to a \emph{generalized}
state. However, this is only meaningful if the domains of the $\LV$ functions
 of the two states coincide (i.e., the states consider the same program
variables). Therefore, after the branch from the loop body back to  \cmpFor{}
(see State $\GraphInitCmpZeroA$ in Fig.\ \ref{fig:ForLoopSecondIteration}), we evaluate the loop a
second time and reach $\GraphInitCmpZeroB$. Here, a second list element with 
\def\distB{0.85cm}
\begin{wrapfigure}[5]{r}{4.2cm}
\vspace*{-0.85cm}
\begin{tikzpicture}[node distance=2cm, auto]
\hspace*{-0.15cm}
%---------- A ----------
\node (CmpZeroA) at (-0.67,0.07) {\small{$\GraphInitCmpZeroA:$}};
\node[label={[label distance=-0.55cm]0:{\scriptsize{$\;\;v_\codevar{mem}$}}}]  (vmem) at (0,0) {};
\node[data, right=0.35cm of vmem] (L1A) {\scriptsize{$\;\;v_\codevar{nd}$} \nodepart{second} \scriptsize{0}};
   
\node[above=0.17cm of $(L1A.west)!0.5!(L1A.text split)$] (o1L1A) {\tiny{\code{value}}};
\node[above=0.17cm of $(L1A.east)!0.5!(L1A.text split)$] (o2L1A) {\tiny{\code{next}}};

\path[ptr]  ($(vmem)+(0.25,0)$) --++(right:2mm)  |- (L1A.text west);

%---------- B ----------
\node (CmpZeroB) at ($(-0.67,0.07)+(0,-\distB)$) {\small{$\GraphInitCmpZeroB:$}};
\node[] (distanceB) at (0,-\distB) {};
\node[label={[label distance=-0.55cm]0:{\scriptsize{$\;\;w_\codevar{mem}$}}}]  (wmem) at (distanceB) {};
\node[data, right=0.35cm of wmem] (L1B) {\scriptsize{$\;\;w_\codevar{nd}$} \nodepart{second}};
\node[data, right=0.2cm of L1B]   (L2B) {\scriptsize{$\;\;v_\codevar{nd}$} \nodepart{second} \scriptsize{0}};
   
\node[above=0.17cm of $(L1B.west)!0.5!(L1B.text split)$] (o1L1B) {\tiny{\code{value}}};
\node[above=0.17cm of $(L1B.east)!0.5!(L1B.text split)$] (o2L1B) {\tiny{\code{next}}};
\node[above=0.17cm of $(L2B.west)!0.5!(L2B.text split)$] (o1L2B) {\tiny{\code{value}}};
\node[above=0.17cm of $(L2B.east)!0.5!(L2B.text split)$] (o2L2B) {\tiny{\code{next}}};

\path[ptr]  ($(wmem)+(0.25,0)$) --++(right:2mm)  |- (L1B.text west);
\draw[fill] ($(L1B.east)!0.5!(L1B.text split)$) circle (0.05);
\draw[ptr]  ($(L1B.east)!0.5!(L1B.text split)$) --++(right:2mm) |- (L2B.text west);

\end{tikzpicture}
\vspace*{-.1cm}
\end{wrapfigure}
value
$w_\code{nd}$ and a \code{next} pointer $w_\code{cn}$ pointing to\linebreak $v_\code{mem}$ has been stored at a new
allocation
$\llbracket{}w_\code{mem},\linebreak w_\code{mem}^\mathit{end}\rrbracket$. Now, \code{curr} points to
the new element and \code{k} has been incremented again, so \code{k\_ad} points to
2.
  
\footnotesize

\begin{figure}[t]
\centering
\begin{tikzpicture}[node distance = \ydist and \xdist]
\scriptsize
\def\widetwidth{6.2cm}
\def\smalltwidth{4.75cm}
\def\fulllinewidhth{11cm}
\def\edgenodedist{0.2cm}
\def\FirstIndentwidth{3cm}
\def\SecondIndentwidth{2cm}
\tikzset{invisible/.style={opacity=0}}
\tikzstyle{state}=[
						   %minimum size=10pt,
                           %fill=white,
                           %shape=rectangle,
                           %text=black,
                           inner sep=2pt,
                           font=\scriptsize,
                           draw]

\node[state, align=left, label={[xshift=-\labelxshift]2:$\GraphInitCmpZeroA$}] (1) 
{$(\cmpFor, 0),\;
  \{
   \code{n} = v_{\code{n}}, \,
   \code{tail\_ptr} = v_{\code{tp}}, \,
   \code{mem} = v_{\code{mem}}, \,
   \code{curr} = v_{\code{mem}}, \,
   \code{nondet} = v_{\code{nd}}, \,
   \code{curr\_val} = v_{\code{mem}},$\\ 
   $\,\code{curr\_next} = v_{\code{cn}}, \,
\code{k} = 0, \,
   \code{kinc} = 1, \,
   ...\},\;
  \{
   \alloc{v_{\code{tp}}}{v_{\code{tp}}^\mathit{end}}, \,
   \alloc{v_{\code{k\_ad}}}{v_{\code{k\_ad}}^\mathit{end}}, \,
   \alloc{v_{\code{mem}}}{v_{\code{mem}}^\mathit{end}}
   \},$\\
 $\{
   v_{\code{tp}} \pointsto[\code{list*}] v_{\code{mem}}, \,
   v_{\code{k\_ad}} \pointsto[\code{i32}] 1, \,
   v_{\code{mem}} \pointsto[\code{i32}] v_{\code{nd}}, \,
   v_{\code{cn}} \pointsto[\code{list*}] 0
   \}, \;
  \emptyset,$\\
 $\{ v_{\code{n}} > 0, \,
    v_{\code{k\_ad}}^\mathit{end} = v_{\code{k\_ad}} + 3, \,
   v_{\code{tp}}^\mathit{end} = v_{\code{tp}} + 7, \,
   v_{\code{mem}}^\mathit{end} = v_{\code{mem}} + 15, \,
   v_{\code{cn}} = v_{\code{mem}} + 8, \,
   ... \}$
};

\node[state, below=of 1, align=left, label={[xshift=-\labelxshift]2:$\GraphInitCmpZeroB$}] (2) 
{$(\cmpFor, 0),\;
  \{   
   \code{n} = v_{\code{n}}, \,
   \code{tail\_ptr} = v_{\code{tp}}, \,
   \code{mem} = w_{\code{mem}}, \,
   \code{curr} = w_{\code{mem}}, \,
   \code{nondet} = w_{\code{nd}}, \,
   \code{curr\_val} = w_{\code{mem}},$\\
   $\,\code{curr\_next} = w_{\code{cn}}, \,
\code{k} = 1, \,
   \code{kinc} = 2, \,
   ...\},\;
  \{
   \alloc{v_{\code{tp}}}{v_{\code{tp}}^\mathit{end}}, \,
   \alloc{v_{\code{k\_ad}}}{v_{\code{k\_ad}}^\mathit{end}}, \,
   \alloc{v_{\code{mem}}}{v_{\code{mem}}^\mathit{end}}, \,
   \alloc{w_{\code{mem}}}{w_{\code{mem}}^\mathit{end}}
   \},$\\
 $\{
   v_{\code{tp}} \pointsto[\code{list*}] w_{\code{mem}}, \,
   v_{\code{k\_ad}} \pointsto[\code{i32}] 2, \,
   v_{\code{mem}} \pointsto[\code{i32}] v_{\code{nd}}, \,
   v_{\code{cn}} \pointsto[\code{list*}] 0, \,
   w_{\code{mem}} \pointsto[\code{i32}] w_{\code{nd}}, \,
   w_{\code{cn}} \pointsto[\code{list*}] v_{\code{mem}}
   \}, \;
  \emptyset,$\\
 $\{v_{\code{n}}> 1, 
  v_{\code{k\_ad}}^\mathit{end} = v_{\code{k\_ad}}\!+\!3, 
   v_{\code{tp}}^\mathit{end} = v_{\code{tp}}\!+\!7,  
   v_{\code{mem}}^\mathit{end} = v_{\code{mem}}\!+\!15, 
   v_{\code{cn}} = v_{\code{mem}}\!+\!8, 
   w_{\code{mem}}^\mathit{end} = w_{\code{mem}}\!+\!15, 
   w_{\code{cn}} = w_{\code{mem}}\!+\!8, 
   ... \}$
};

\draw[omit-edge] (1) -- (2);

\end{tikzpicture}
\paper{\vspace*{-.5cm}}
\caption{Second Iteration of the \code{for} Loop}
\label{fig:ForLoopSecondIteration}
\vspace*{-.5cm}
\end{figure}
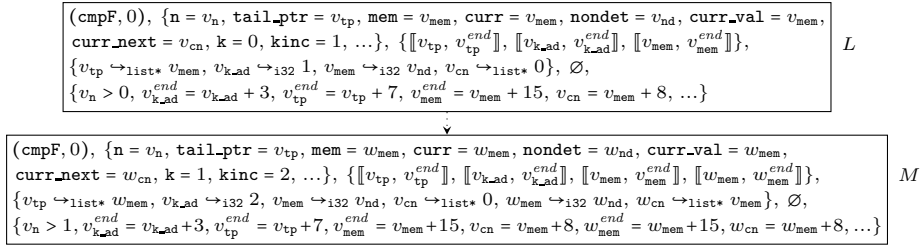
\normalsize

\subsection{\mbox{Inferring List Invariants and Generalization of States}}
\label{sect:inferring_inv}

As mentioned, our goal is to merge $\GraphInitCmpZeroA$ and
  $\GraphInitCmpZeroB$ to a more general state $\GraphInitCmpZeroGen$ that repre\-sents all states which are
represented by  $\GraphInitCmpZeroA$ or
  $\GraphInitCmpZeroB$.
The challenging part during\linebreak generalization is to find loop invariants
  automatically that always hold at this po\-sition and provide sufficient information to
  prove termination of the loop.
For $\GraphInitCmpZeroGen$,\linebreak we can neither use the information
that \code{curr} points to a struct whose \code{next} field contains the null pointer (as
  in $\GraphInitCmpZeroA$), nor that its \code{next} field points to another struct
  whose \code{next} field contains the null pointer (as in $\GraphInitCmpZeroB$).
  
With the approach of \cite{LLVM-JAR17},
 when merging states like $\GraphInitCmpZeroA$ and
 $\GraphInitCmpZeroB$
where a list has\linebreak different lengths, 
the merged state would only contain those list elements
that are allocated in both states (often this is only the first element).
Then elements which are the null pointer
in one but not in 
the other state are lost. Hence, proving memory safety (and thus, also 
termination) fails when the list is traversed afterwards, since now there might be
\code{next} pointers to non-allocated memory. 

We solve this problem by introducing \emph{list invariants}. In our
example, we will\linebreak infer an invariant stating that \code{curr} points to a list of
length $x_\ell \geq 1$. This invariant also implies that all struct fields are allocated
  and that there is no sharing.

To this end, we  adapt the merging heuristic from  \cite{LLVM-JAR17}.
To merge two states $s$ and $s'$ at the same program position with $\domain(\LV^s) =
\domain(\LV^{s'})$, we introduce a fresh symbolic variable $x_\codevar{var}$ for each program
variable \codevar{var} and use instantiations $\mu_\StateA$ and $\mu_\StateB$ \pagebreak which map
$x_\codevar{var}$ to the corresponding symbolic variables of $\StateA$ and
$\StateB$. For the merged state $\StateGen$, we set
$\LV^\StateGen(\codevar{var}) = x_\codevar{var}$. Moreover, we identify corresponding
variables that
only occur in the memory components and extend  $\mu_\StateA$ and $\mu_\StateB$ accordingly.
In a second step, we check which constraints from the memory components and the knowledge
base hold in both states in order to find invariants that we can add to the memory
components and the knowledge base of $\StateGen$. For example, if
$\alloc{\mu_\StateA(x)}{\mu_\StateA(x^{end})} \in \AL^\StateA$ and
$\alloc{\mu_\StateB(x)}{\mu_\StateB(x^{end})} \in \AL^\StateB$ for $x, x^{end} \in \Vsym$,
then $\alloc{x}{x^{end}}$ is added to $\AL^\StateGen$.  To extend this heuristic
to lists, we have to regard several memory entries together. 
If there is an
$\codevar{ad} \in \Ids$ such that $\mu_\StateA(x_\codevar{ad}) = v_1^\mathit{start}$ and
$\mu_\StateB(x_\codevar{ad}) = w_1^\mathit{start}$ both point to lists of type \codevar{ty}
but of different lengths $\ell_\StateA \neq
\ell_\StateB$ with $\ell_\StateA, \ell_\StateB \geq 1$, then we create a list invariant.

For a state $s$ we say that $v_1^\mathit{start}$ \emph{points to a list of
type \codevar{ty} with $n$ fields and length $\ell_\StateA$ with allocations
    $\alloc{v^\mathit{start}_k}{v^\mathit{end}_k}$ and
values $v_{k,i}$}
(for $1 \leq k \leq \ell_\StateA$ and $1 \leq i \leq n$)
if the following
conditions $(a)-(d)$ hold: 
\begin{itemize}
	\item[$(a)$] \codevar{ty} is an \LLVM{} struct type with subtypes $\codevar{ty}_i$
          and field offsets $\mathit{off}_i \in \N$ for all $1 \leq i \leq n$ such that there exists exactly one $1 \leq j \leq n$ with $\codevar{ty}_j = \codevar{ty*}$.
	\item[$(b)$] There exist pairwise different
          $\alloc{v^\mathit{start}_k}{v^\mathit{end}_k} \in \AL^\StateA$ for all $1 \leq k
          \leq \ell_\StateA$ and $\models \, \stateformula{\StateA} \implies
v^\mathit{end}_k = v^\mathit{start}_k + \sizeOf(\codevar{ty})-1$.
	\item[$(c)$] For all $1 \leq k \leq \ell_\StateA$ and $1 \leq i \leq n$ there
          exist $v^\mathit{start}_{k,i}, v_{k,i} \in \Vsym$ with
          $\models \, \stateformula{\StateA} \implies
          v^\mathit{start}_{k,i} = v^\mathit{start}_k + \mathit{off}_i$ and
          $(v^\mathit{start}_{k,i} \pointsto[\codevar{ty}_i] v_{k,i}) \in \PT^\StateA$.
	\item[$(d)$] For all $1 \leq k < \ell_\StateA$ we have $\models \,
          \stateformula{\StateA} \implies v_{k,j} = v^\mathit{start}_{k+1}$. 
\end{itemize}
Condition $(a)$ states that \code{ty} is a list type with $n$ fields, where
the pointer to the next element is  in the $j$-th field. In $(b)$ we ensure that each list element has a unique
allocation of the correct size where $v_1^\mathit{start}$ is the start address of the
first allocation.
  Condition $(c)$ requires that for the $k$-th element, the initial address
plus\linebreak the $i$-th offset points to a value $v_{k,i}$ of type $\codevar{ty}_i$. Finally, $(d)$
states that the recur\-sive field of each element indeed points to the initial
address of the next element.

Then, for fresh $x_\ell,x_i,\hat{x}_i \in \Vsym$, we add
the following list invariant to $\LI^\StateGen$.

\vspace*{-.2cm}

\begin{equation}
  \label{exampleListInvariantMergedState}
x_{\codevar{ad}} \pointstorec[\codevar{ty}]{x_\ell} [(\mathit{off}_i: \codevar{ty}_i:
  {x_i..\hat{x}_i})]_{i=1}^n
\end{equation}

To ensure
that the allocations expressed by the list invariant are disjoint from all
allocations in $\AL^\StateGen$, we do not use
the list allocations $\alloc{v^\mathit{start}_k}{v^\mathit{end}_k}$ to infer
generalized allocations in $\AL^\StateGen$.
Similarly, to create
$\PT^\StateGen$, we only use entries 
$v \hookrightarrow_{\codevar{ty}} w$ from
$\PT^\StateA$
and $\PT^\StateB$ where $v$ is disjoint from the list
addresses,
i.e., where $\models \stateformula{s} \implies v < v_k^\mathit{start} \lor v >
v_k^\mathit{end}$ holds for all $1 \leq k \leq \ell_s$ and analogously for $s'$.
Moreover, we add formulas to $\KB^\StateGen$ stating that $(A)$ the length $x_\ell$ of
the list is at least the smaller length of the merged lists, $(B)$ $x_\ell$
is equal to all variables $x$ which result from merging variables $v$ and $w$ that
are equal to the lengths $\ell_\StateA$ and $\ell_\StateB$ in $\StateA$ and $\StateB$, and
$(C)$   the symbolic variable $x_i$ for the value of the $i$-th field of the first list
element is equal
to all variables $x$ with $\mu_\StateA(x) = v_{1,i}$ and $\mu_\StateB(x) = w_{1,i}$ where
$v_{1,i}$ and $w_{1,i}$ are the values of the $i$-th field of the first list element in
$s$ and $s'$
(and
analogously for the values $\hat{x}_i$ of the last list element):
\begin{itemize}
	\item[$(A)\!\!$] $\min(\ell_\StateA,\ell_\StateB) \leq x_\ell$
	\item[$(B)\!\!$] {\small $\bigwedge_{x \in \mu^{-1}_\StateA(v) \cap \mu^{-1}_\StateB(w)} x_\ell = x$ for all $v,w \in \Vsym$ with $\models \, \stateformula{\StateA} \implies v = \ell_\StateA$ and $\models \, \stateformula{\StateB} \implies w = \ell_\StateB$}
	\item[$(C)\!\!$] {\small $\bigwedge_{x \in \mu^{-1}_\StateA(v_{1,i}) \cap \mu^{-1}_\StateB(w_{1,i})} x_i = x$\; and \;$\bigwedge_{x \in \mu^{-1}_\StateA(v_{\ell_{\StateA},i}) \cap \mu^{-1}_\StateB(w_{\ell_{\StateB},i})} \hat{x}_i = x$\; for all $1 \leq i \leq n$}
\pagebreak\end{itemize}

To identify the variables in the list invariant 
\eqref{exampleListInvariantMergedState}
of $\StateGen$ with the corresponding values in $\StateA$ and $\StateB$, the instantiations $\mu_\StateA$ and $\mu_\StateB$ are extended such that $\mu_\StateA(x_\ell) = \ell_\StateA$, $\mu_\StateB(x_\ell) = \ell_\StateB$, $\mu_\StateA(x_i) = v_{1,i}$, $\mu_\StateB(x_i) = w_{1,i}$, $\mu_\StateA(\hat{x}_i) = v_{\ell_{\StateA},i}$, and $\mu_\StateB(\hat{x}_i) = w_{\ell_{\StateB},i}$ for all $1 \leq i \leq n$.
Similarly, if there already exist list invariants in $\StateA$ and $\StateB$, for each pair of corresponding variables a new variable is introduced and mapped to its origin by $\mu_\StateA$ and $\mu_\StateB$.
This adaption of the merging heuristic only concerns the result of merging but
not the rules \emph{when} to merge two states. Thus, 
the same reasoning as in \cite{LLVM-JAR17} can be used to prove soundness and termination of
merging.

\footnotesize

\begin{figure}[t]
\centering
\begin{tikzpicture}[node distance = \ydist and \xdist]
\scriptsize
\def\widetwidth{6.2cm}
\def\smalltwidth{4.75cm}
\def\fulllinewidhth{11cm}
\def\edgenodedist{0.2cm}
\def\FirstIndentwidth{3cm}
\def\SecondIndentwidth{2cm}
\tikzset{invisible/.style={opacity=0}}
\tikzstyle{state}=[
						   %minimum size=10pt,
                           %fill=white,
                           %shape=rectangle,
                           %text=black,
                           inner sep=2pt,
                           font=\scriptsize,
                           draw]

\node[state, align=left, label={[xshift=-\labelxshift]91:$\GraphInitCmpZeroA$}] (1) {};

\node[state, below=of 1, align=left, label={[xshift=-\labelxshift]269:$\GraphInitCmpZeroB$}] (2) {};

\node[state, right=of 1, yshift=-.23cm, xshift=.5cm, align=left, label={[xshift=-\labelxshift]175:$\GraphInitCmpZeroGen$}] (3) 
{$(\cmpFor, 0),\;
  \{
     \code{n} = x_{\code{n}}, \,
     \code{tail\_ptr} = x_{\code{tp}}, \,
   \code{mem} = x_{\code{mem}}, \,
   \code{curr} = x_{\code{mem}}, \,
   \code{nondet} = x_{\code{nd}}, \,
   \code{curr\_val} = x_{\code{mem}},$\\
$\,
   \code{curr\_next} = x_{\code{cn}}, \, \code{k} = x_{\code{k}}, \,
   \code{kinc} = x_{\code{kinc}}, \,
    ...\},\;
  \{
   \alloc{x_{\code{tp}}}{x_{\code{tp}}^\mathit{end}}, \,
   \alloc{x_{\code{k\_ad}}}{x_{\code{k\_ad}}^\mathit{end}}
   \},$\\
 $\{
   x_{\code{tp}} \pointsto[\code{list*}] x_{\code{mem}}, \,
   x_{\code{k\_ad}} \pointsto[\code{i32}] x_{\code{kinc}}
  \}, \;
  \{
   x_{\code{mem}} \pointstorec[\code{list}]{x_{\ell}}
     [(0: \code{i32}: x_{\code{nd}}..\hat{x}_{\code{nd}}),
      (8: \code{list*}: x_{\code{next}}..0)]
   \},$\\
 $\{
      x_{\code{n}} > x_{\code{k}}, \,
x_{\code{k\_ad}}^\mathit{end} = x_{\code{k\_ad}} + 3, \,
   x_{\code{tp}}^\mathit{end} = x_{\code{tp}} + 7, \,
%   x_{\code{mem}_\mathit{end}} = x_{\code{mem}} + 15, \,
   x_{\code{cn}} = x_{\code{mem}} + 8, \,
   x_{\code{kinc}} = x_{\code{k}} + 1, \,
   1 \leq x_\ell, \,
   x_{\ell} = x_{\code{kinc}}, \,
   ... \}$
};

\draw[omit-edge] (1) -- (2);

\coordinate (merge-p1) at ($(1.east)+(\edgenodedist*1.6,0)$);
\coordinate (merge-p2) at (merge-p1 |- 2.east);
\draw[-,rounded corners=3pt] (1.east) -- (merge-p1) -- (merge-p2) -- (2.east);

\coordinate (gen-p1) at ($(1.south east)+(\edgenodedist*1.6,-\ydist/2)$);
\draw[gen-edge] (gen-p1) --  (3.west);

\end{tikzpicture}
\vspace*{-.1cm}
\caption{Merging of States}
\label{fig:ForLoopMerging}
\vspace*{-.6cm}
\end{figure}
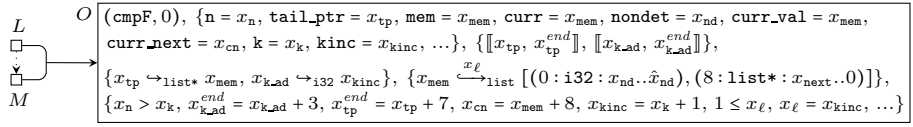
\normalsize

In our example, $\GraphInitCmpZeroA$
and $\GraphInitCmpZeroB$ 
contain lists of length $\ell_\GraphInitCmpZeroA =
1$ and $\ell_\GraphInitCmpZeroB = 2$. To
ease the presentation, we re-use variables that are known to be equal instead of
introducing fresh variables. If $x_\code{mem}$ is the variable for the program variable
\code{curr}, we have $\mu_\GraphInitCmpZeroA(x_\codevar{mem}) = v_\code{mem}$ and
$\mu_\GraphInitCmpZeroB(x_\codevar{mem}) = w_\code{mem}$. Indeed, $v_\code{mem}$
resp.\ $w_\code{mem}$ points to a list with values $v_{k,i}$ resp. $w_{k,i}$ as defined in
$(a)$--$(d)$: For the type \code{list} with $n=2$, $\codevar{ty}_1 = \codevar{i32}$,
$\codevar{ty}_2 = \codevar{list*}$, $\mathit{off}_1 = 0$, $\mathit{off}_2 = 8$, and
$j = 2$ (see $(a)$), we have 
$\alloc{v_\codevar{mem}}{{v_{\codevar{mem}}^\mathit{end}}} \in \AL^\GraphInitCmpZeroA$ and
$\alloc{v_\codevar{mem}}{{v_{\codevar{mem}}^\mathit{end}}}$, $
\alloc{w_\codevar{mem}}{{w_{\codevar{mem}}^\mathit{end}}} \in \AL^\GraphInitCmpZeroB$, all
consisting of $\sizeOf(\codevar{list}) = 16$ bytes, see $(b)$.
We have $(v_{\code{mem}} \pointsto[\code{i32}] v_{\code{nd}}), (v_{\code{cn}} \pointsto[\code{list*}] 0)
\in \PT^\GraphInitCmpZeroA$ with $(v_{\code{cn}} = v_{\code{mem}} + 8) \in
\KB^\GraphInitCmpZeroA$ and $(v_{\code{mem}} \pointsto[\code{i32}] v_{\code{nd}}), (v_{\code{cn}}
\pointsto[\code{list*}] 0), (w_{\code{mem}} \pointsto[\code{i32}] w_{\code{nd}}), (w_{\code{cn}}
\pointsto[\code{list*}] v_{\code{mem}}) \in \PT^\GraphInitCmpZeroB$ with $(v_{\code{cn}} =
v_{\code{mem}} + 8), (w_{\code{cn}} = w_{\code{mem}} + 8) \in \KB^\GraphInitCmpZeroB$
(see  $(c)$), so the first list element in $\GraphInitCmpZeroB$ points to the second
one (see  $(d)$). 
Therefore, when merging $\GraphInitCmpZeroA$ and $\GraphInitCmpZeroB$ to a new state
$\GraphInitCmpZeroGen$ (see \cref{fig:ForLoopMerging}), the lists are merged to a list invariant of variable length
$x_\ell$ and we add the formulas $(A)$ $1 \leq x_\ell$ and $(B)$ $x_\ell = x_\code{kinc}$
to $\KB^\GraphInitCmpZeroGen$. 
By $(C)$, the \codevar{i32} value of the first element is identified with
$x_\code{nd}$, since $\mu_\GraphInitCmpZeroA(x_\code{nd})$ is equal to the first value of
the first list element in $\GraphInitCmpZeroA$ and $\mu_\GraphInitCmpZeroB(x_\code{nd})$ is
equal to the first value of the first list element in $\GraphInitCmpZeroB$. Similarly, the
values of the last list elements are identified
  with $0$, as in $\GraphInitCmpZeroA$ and $\GraphInitCmpZeroB$.

After merging $s$ and $s'$ to a generalized state $\overline{s}$, we
continue symbolic execution from $\overline{s}$.
 The next time we reach the same program position, we might have to merge
the corresponding states again. As described in \cite{LLVM-JAR17}, we use
a heuristic for constructing the SEG which ensures
that after a finite number of iterations, a state is reached
that only represents concrete states that are also represented by an \emph{already existing} (more
general) state in the SEG. Then symbolic execution can\linebreak continue from this more general state
instead. So with this heuristic,
the construction always ends in a complete SEG or  an SEG containing the state $\ERROR$.

We formalized the concept of ``generalization''  by a symbolic
execution rule in \cite{LLVM-JAR17}.
Here, the state $\overline{s}$ is a generalization of $s$ if the conditions
$(g1)-(g6)$ hold.

Condition $(g1)$ prevents cycles consisting
only of refinement and generalization edges in the graph.
Condition $(g2)$ states that the instantiation $\mu \colon
\Vsym(\overline{s}) \to \Vsym(s) \cup \Z$ maps symbolic variables from the more general state $\overline{s}$ to their counterparts from the more specific state $s$ such that they correspond to the same program variable.
Conditions  $(g3)$--$(g6)$ ensure that all knowledge present in $\overline{\KB}$, $\overline{\AL}$, $\overline{\PT}$, and $\overline{\LI}$ still holds in $s$ with the applied instantiation.

\vspace*{.12cm}
\noindent
\fbox{
\begin{minipage}{11.8cm}
\label{rule:generalization}
\small
\mbox{\small \textbf{\hspace*{-.15cm}generalization with instantiation $\mu$}}\\
\vspace*{-.2cm}\\
\centerline{$\frac{\parbox{4.6cm}{\centerline{
$s = (\Pos, \; \LV, \; \AL, \; \PT, \; \LI, \; \KB)$}\vspace*{.1cm}}}
{\parbox{3.2cm}{\vspace*{.1cm} \centerline{
$\overline{s} = (\Pos, \; \overline{\LV}, \; \overline{\AL}, \;  \overline{\PT}, \; \overline{\LI}, \; \overline{\KB})$}}}\;\;\;\;$
\mbox{if}} \vspace*{-.3cm}
{\small
\begin{itemize}
\item[$(g1)$] $s$ has an incoming evaluation edge
\item[$(g2)$] $\domain(\LV) = \domain(\overline{\LV})$  and
$\LV(\code{var}) = \mu(\overline{\LV}(\code{var}))$ for all $\code{var} \in \Ids$ where $\LV$ and $\overline{\LV}$ are defined
\item[$(g3)$] $\models \stateformula{s} \implies \mu(\overline{\KB})$
\item[$(g4)$] if $\alloc{x_1}{x_2} \in \overline{\AL}$, then $\alloc{v_1}{v_2} \in \AL$ with $\models \stateformula{s} \implies  v_1 = \mu(x_1) \wedge v_2 = \mu(x_2)$
\item[$(g5)$] if $(x_1 \hookrightarrow_{\codevar{ty}} x_2) \in \overline{\PT}$,\\then $(v_1 \hookrightarrow_{\codevar{ty}} v_2) \in \PT$ with $\models \stateformula{s} \implies v_1 = \mu(x_1) \wedge v_2 = \mu(x_2)$ 
\item[$(g6)$] \label{bullet-point-generalization-list-inv} if $(x_{\codevar{ad}} \pointstorec[\codevar{ty}]{x_\ell} [(\mathit{off}_i: \codevar{ty}_i: {x_i..\hat{x}_i})]_{i=1}^n) \in \overline{\LI}$,\\then either $(v_{\codevar{ad}} \pointstorec[\codevar{ty}]{v_\ell} [(\mathit{off}_i: \codevar{ty}_i: {v_i..\hat{v}_i})]_{i=1}^n) \in \LI$ with
    \begin{itemize}
        \item $\models \stateformula{s} \implies v_\codevar{ad} = \mu(x_\codevar{ad}) \wedge v_\ell = \mu(x_\ell)$ and
        \item $\models \stateformula{s} \implies v_i = \mu(x_i) \wedge \hat{v}_i = \mu(\hat{x}_i)$ for all $1 \leq i \leq n$,
    \end{itemize}
    or $v^\mathit{start}_1$ points to a list of type \codevar{ty} and length $\ell$
with allocations
    $\alloc{v^\mathit{start}_k}{v^\mathit{end}_k}$ and
       values $v_{k,i}$ (for $1 \leq k \leq \ell, 1 \leq i \leq n$) such that
    \begin{itemize}
        \item $\models \stateformula{s} \implies v^\mathit{start}_1 = \mu(x_\codevar{ad}) \wedge \ell = \mu(x_\ell)$,
        \item $\models \stateformula{s} \implies v_{1,i} = \mu(x_i) \wedge v_{\ell,i} = \mu(\hat{x}_i)$ for all $1 \leq i \leq n$, and
        \item if $(z_1 \hookrightarrow_{\codevar{ty}} z_2) \in \overline{\PT}$,\\then $\models \stateformula{s} \implies \mu(z_1) < v^\mathit{start}_k \lor \mu(z_1) > v^\mathit{end}_k$ for all $1 \leq k \leq \ell$.
    \end{itemize}
\end{itemize}}
\end{minipage}}
\vspace*{.12cm}

Condition $(g6)$ is new compared to \cite{LLVM-JAR17} and takes list
invariants into account.
So  for every list invariant $\overline{\listinv}$ of $\overline{s}$
there is either a corresponding list invariant $\listinv$ in $s$ such that lists
represented by $\listinv$ in $s$ are also represented by $\overline{\listinv}$ in
$\overline{s}$, or there is a concrete list in $s$ that is represented by
$\overline{\listinv}$ in $\overline{s}$. The last condition of the latter case ensures
that disjointness between
the memory domains of $\overline{\PT}$ and $\overline{\LI}$ is preserved.
See\arxiv{ App. \ref{app:proofs}}\paper{ \cite{report}}
for
the soundness proof of the extended generalization\linebreak rule, i.e., that every concrete state
represented by $s$ is also represented by $\overline{s}$.

Our merging technique always yields generalizations according to this
rule, i.e., the edges from $\GraphInitCmpZeroA$ and $\GraphInitCmpZeroB$
to $\GraphInitCmpZeroGen$ in \cref{fig:ForLoopMerging} are  generalization
edges.
Here, one\linebreak chooses
$\mu_\GraphInitCmpZeroA$ and
$\mu_\GraphInitCmpZeroB$
such that $\mu_\GraphInitCmpZeroA(x_\codevar{mem}) =
v_\codevar{mem}$, $\mu_\GraphInitCmpZeroA(x_\ell) = 1$,
$\mu_\GraphInitCmpZeroA(x_\codevar{nd}) = v_\codevar{nd}$,
$\mu_\GraphInitCmpZeroA(\hat{x}_\codevar{nd})\linebreak = v_\codevar{nd}$,
$\mu_\GraphInitCmpZeroA(x_\codevar{next}) = 0$,
$\mu_\GraphInitCmpZeroB(x_\codevar{mem}) = w_\codevar{mem}$,
$\mu_\GraphInitCmpZeroB(x_\ell) = 2$,
$\mu_\GraphInitCmpZeroB(x_\codevar{nd}) = w_\codevar{nd}$,
$\mu_\GraphInitCmpZeroA(\hat{x}_\codevar{nd}) = v_\codevar{nd}$, and
$\mu_\GraphInitCmpZeroB(x_\codevar{next}) = v_\codevar{mem}$. In both cases, all
conditions of the second case of $(g6)$ with $\ell_\GraphInitCmpZeroA = 1$ and
$\ell_\GraphInitCmpZeroB = 2$ are satisfied. With
$\mu_\GraphInitCmpZeroA(x_\codevar{kinc}) = 1$
resp. $\mu_\GraphInitCmpZeroB(x_\codevar{kinc}) = 2$, we also have $\models
\stateformula{\GraphInitCmpZeroA} \implies \mu_\GraphInitCmpZeroA(x_\ell) =
\mu_\GraphInitCmpZeroA(x_\codevar{kinc})$ resp. $\models \stateformula{\GraphInitCmpZeroB}
\implies \mu_\GraphInitCmpZeroB(x_\ell) = \mu_\GraphInitCmpZeroB(x_\codevar{kinc})$.

\vspace*{-.1cm}

\subsection{Adapting List Invariants}
\label{sect:adapting_inv}

To handle and modify list invariants, three of our symbolic execution rules have to be
changed.
 \cref{sect:List Extension} presents a variant of the \code{store}
rule where the list invariant is \emph{extended} by an
element. In \cref{sect:List Traversal},  we adapt the  \code{load}  rule 
to load values from the first list element
and
we present a variant of the \code{getelementptr} 
rule for list \emph{traversal}.
Soundness of our new rules is proved in\arxiv{ App.\  \ref{app:proofs}}\paper{ \cite{report}}.
For all other instructions, the symbolic execution rules from \cite{LLVM-JAR17} remain
unchanged.

\vspace*{-.25cm}

\subsubsection{List Extension}
\label{sect:List Extension}

\footnotesize

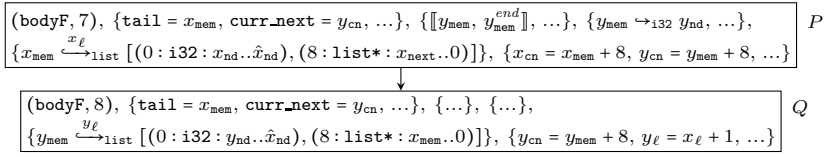
\begin{figure}[t]
\centering
\begin{tikzpicture}[node distance = \ydist and \xdist]
\scriptsize
\def\widetwidth{6.2cm}
\def\smalltwidth{4.75cm}
\def\fulllinewidhth{11cm}
\def\edgenodedist{0.2cm}
\def\FirstIndentwidth{3cm}
\def\SecondIndentwidth{2cm}
\tikzset{invisible/.style={opacity=0}}
\tikzstyle{state}=[
						   %minimum size=10pt,
                           %fill=white,
                           %shape=rectangle,
                           %text=black,
                           inner sep=2pt,
                           font=\scriptsize,
                           draw]

\node[state, align=left, label={[xshift=-\labelxshift]2:$\GraphInitBodyStoreStart$}] (1) 
{$(\bodyFor, 7), \;
  \{
   \code{tail} = x_{\code{mem}}, \,
%   \code{curr} = y_{\code{mem}}, \,
%   \code{curr\_val} = y_{\code{mem}}, \,
   \code{curr\_next} = y_{\code{cn}}, \,
%   \code{k} = x_{\code{kinc}}, \,
%   \code{kinc} = x_{\code{kinc}}, \,
   ...\}, \;
  \{
   \alloc{y_{\code{mem}}}{y_{\code{mem}}^\mathit{end}}, \,
   ...\}, \;
  \{
   y_{\code{mem}} \pointsto[\code{i32}] y_{\code{nd}}, \,
   ...
  \},$\\
 $\{
   x_{\code{mem}} \pointstorec[\code{list}]{x_{\ell}}
     [(0: \code{i32}: x_{\code{nd}}..\hat{x}_{\code{nd}}),
      (8: \code{list*}: x_{\code{next}}..0)]
  \}, \;
  \{
%   x_{\code{mem}_\mathit{end}} = x_{\code{mem}} + 15, \,
   x_{\code{cn}} = x_{\code{mem}} + 8, \,
%   y_{\code{mem}_\mathit{end}} = y_{\code{mem}} + 15, \,
   y_{\code{cn}} = y_{\code{mem}} + 8, \,
%   x_{\code{kinc}} = x_{\code{k}} + 1, \,
%   x_{\ell} = x_{\code{kinc}}, \,
   ... \}$
};

\node[state, align=left, below=of 1, label={[xshift=-\labelxshift]2:$\GraphInitBodyStoreRes$}] (2) 
{$(\bodyFor, 8), \;
  \{
   \code{tail} = x_{\code{mem}}, \,
%   \code{curr} = y_{\code{mem}}, \,
%   \code{curr\_val} = y_{\code{mem}}, \,
   \code{curr\_next} = y_{\code{cn}}, \,
%   \code{k} = x_{\code{kinc}}, \,
%   \code{kinc} = x_{\code{kinc}}, \,
   ...\}, \;
  \{
   ...\}, \; 
  \{...\},$\\
 $\{
   y_{\code{mem}} \pointstorec[\code{list}]{y_{\ell}}
     [(0: \code{i32}: y_{\code{nd}}..\hat{x}_{\code{nd}}),
      (8: \code{list*}: x_{\code{mem}}..0)]
  \}, \;
  \{
%   y_{\code{mem}_\mathit{end}} = y_{\code{mem}} + 15, \,
   y_{\code{cn}} = y_{\code{mem}} + 8, \,
%   y_{\code{kinc}} = x_{\code{kinc}} + 1, \,
   y_{\ell} = x_{\ell} + 1, \,
%   x_{\ell} = x_{\code{kinc}}, \,
   ... \}$
};

\draw[eval-edge] (1) -- (2);

\end{tikzpicture}
\vspace*{-.1cm}
\caption{Extending a List Invariant}
\label{fig:ForLoopStore}
\vspace*{-.5cm}
\end{figure}
\normalsize

After merging $\GraphInitCmpZeroA$ and $\GraphInitCmpZeroB$,
symbolic execution continues from the more general state
$\GraphInitCmpZeroGen$ in \rF{fig:ForLoopMerging}. Here, the values of \code{k} and \code{kinc} and the length of the
list are not concrete but any
positive (resp.\ non-negative) value with $x_\ell = x_\code{kinc} =
x_\code{k}+1$. The symbolic execution of $\GraphInitCmpZeroGen$ is similar to the steps
from $\GraphInitCmpZero$ to $\GraphInitBodySix$ in 
\cref{sect:symbexec} (see \rF{fig:ForLoopFirstIteration}). First, the value $x_\codevar{kinc}$ stored at \code{k\_ad} is
loaded to \code{k}. To distinguish whether \code{k $<$ n} still holds, the next state is
refined.
From
the refined state with \code{k $<$ n}, we enter the
loop body again. A new block $\alloc{y_{\codevar{mem}}}{y_{\codevar{mem}}^\mathit{end}}$
of 16 bytes is allocated and $y_{\codevar{mem}}$ is assigned to
\code{mem} and \code{curr}.
Then, a new unknown value $y_\code{nd}$ is assigned to \code{nondet}.
The address of the \code{i32} value of the current element (equal to
$y_{\code{mem}}$) is computed by the first \code{getelementptr} instruction of the loop
and the value $y_\codevar{nd}$ of \code{nondet} is stored at it. The second
\code{getelementptr} instruction computes the address $y_{\code{cn}}$ of the
recursive field and results in State $\GraphInitBodyStoreStart$ in
Fig. \ref{fig:ForLoopStore}, where $y_{\code{cn}} = y_{\code{mem}}+8$ is
added to $\KB^\GraphInitBodyStoreStart$. Now, \code{store} sets the
address of the \code{next} field to the head of the list created in the previous iteration.
Since this instruction extends the list by an element, instead of adding
$y_{\code{cn}} \pointsto[\code{list*}] x_{\code{mem}}$ to $\PT^\GraphInitBodyStoreRes$, we
extend the list invariant: The length is set to $y_\ell$ and identified with
$x_\ell+1$ in $\KB^\GraphInitBodyStoreRes$. The pointer $x_{\code{mem}}$ to the first
element is replaced by $y_{\code{mem}}$, while
the first recursive field in the list gets the value
$x_{\code{mem}}$. Since $(y_{\code{mem}} \pointsto[\code{i32}]
y_{\code{nd}}) \in \PT^\GraphInitBodyStoreStart$, $y_\codevar{nd}$ is
the value of the first \code{i32} integer in the list. 
We remove all entries from $\PT^\GraphInitBodyStoreRes$ that are already contained in the new list invariant, e.g., $y_{\code{mem}} \pointsto[\code{i32}] y_{\code{nd}}$.

To formalize this adaption of list invariants,
we introduce a modified
 rule for \code{store}
in addition to the one in
\cite{LLVM-JAR17}.
It handles the case where there 
is a
concrete list  at some address $v^{\mathit{start}}$, \code{pa} points to the $m$-th
field of this list's first element, one wants to store a value $\mathit{t}$ 
at the address \code{pa}, and one already has a list\linebreak invariant $l$ for the ``tail'' of the
list in the $j$-th field 
(if $m \neq j$) resp.\ for
the list at
\begin{wrapfigure}[12]{r}{5cm}
\def\distB{1.15cm}
\def\distC{2.5cm}
\def\distD{3.65cm}
\vspace*{-0.7cm}
\hspace*{-.15cm}\begin{tikzpicture}[node distance=2cm, auto]
\hspace*{-0.15cm}
\node (case1) at (-0.1,0.6) {\scriptsize{\underline{$m \neq j$:}}};
\node (case2) at ($(-0.1,0.6)+(0,-\distC)$) {\scriptsize{\underline{$m = j$:}}};
%---------- A ----------
\node[label={[label distance=-0.65cm]0:{\scriptsize{$v^\mathit{start}$}}}]  (vstart1) at (0,0) {};
\node[data, right=0.42cm of vstart1] (L1A) {\scriptsize{\textcolor{white}{5}} \nodepart{second}};
\node[data, right=0.2cm of L1A]   (L2A) {\scriptsize{2} \nodepart{second}};
\node[right=0.2cm of L2A]   (L3A) {\scriptsize{\ldots}};

\path[ptr]  ($(vstart1)+(0.25,0)$) --++(right:2mm)  |- (L1A.text west);
\draw[fill] ($(L1A.east)!0.5!(L1A.text split)$) circle (0.05);
\draw[ptr]  ($(L1A.east)!0.5!(L1A.text split)$) --++(right:2mm) |- (L2A.text west);
\draw[fill] ($(L2A.east)!0.5!(L2A.text split)$) circle (0.05);
\draw[ptr]  ($(L2A.east)!0.5!(L2A.text split)$) --++(right:2mm) |- (L3A.west);

\node  (pa1) at ($(L1A.148)+(0,0.35)$) {\scriptsize{\code{pa}}};
\node  (vad1) at ($(L2A.148)+(0,0.35)$) {\scriptsize{$v_\mathit{ad}$}};

\path[ptr]  ($(pa1)-(0,0.1)$) -- (L1A.148);
\path[ptr]  ($(vad1)-(0,0.13)$) -- (L2A.148);
\draw[red,decorate,decoration={brace,amplitude=4pt}] ($(L2A.south west)+(2cm,-3pt)$) -- ($(L2A.south west)+(0,-3pt)$) node[midway,below,yshift=-2.5pt]{\scriptsize{$l$}};

%---------- B ----------
\node[] (distanceB) at (0,-\distB) {};
\node[label={[label distance=-0.65cm]0:{\scriptsize{$v^\mathit{start}$}}}]  (vstart2) at (distanceB) {};
\node[data, right=0.42cm of vstart2] (L1B) {\scriptsize{5} \nodepart{second}};
\node[data, right=0.2cm of L1B]   (L2B) {\scriptsize{2} \nodepart{second}};
\node[right=0.2cm of L2B]   (L3B) {\scriptsize{\ldots}};

\path[ptr]  ($(vstart2)+(0.25,0)$) --++(right:2mm)  |- (L1B.text west);
\draw[fill] ($(L1B.east)!0.5!(L1B.text split)$) circle (0.05);
\draw[ptr]  ($(L1B.east)!0.5!(L1B.text split)$) --++(right:2mm) |- (L2B.text west);
\draw[fill] ($(L2B.east)!0.5!(L2B.text split)$) circle (0.05);
\draw[ptr]  ($(L2B.east)!0.5!(L2B.text split)$) --++(right:2mm) |- (L3B.west);

\node  (pa2) at ($(L1B.148)+(0,0.35)$) {\scriptsize{\code{pa}}};
\node  (vad2) at ($(L2B.148)+(0,0.35)$) {\scriptsize{$v_\mathit{ad}$}};

\path[ptr]  ($(pa2)-(0,0.1)$) -- (L1B.148);
\path[ptr]  ($(vad2)-(0,0.13)$) -- (L2B.148);
\draw[red,decorate,decoration={brace,amplitude=4pt,aspect=0.29}] ($(L2B.south west)+(2cm,-3pt)$) -- ($(L1B.south west)+(0,-3pt)$) node[pos=0.29,below,yshift=-2.5pt]{\scriptsize{$l'$}};

%---------- C ----------
\node[] (distanceC) at (0,-\distC) {};
\node[label={[label distance=-0.65cm]0:{\scriptsize{$v^\mathit{start}$}}}]  (vstart3) at (distanceC) {};
\node[data, right=0.42cm of vstart3] (L1C) {\scriptsize{5}  \nodepart{second}};
\node[data, right=0.2cm of L1C]   (L2C) {\scriptsize{2} \nodepart{second}};
\node[right=0.2cm of L2C]   (L3C) {\scriptsize{\ldots}};

\path[ptr]  ($(vstart3)+(0.25,0)$) --++(right:2mm)  |- (L1C.text west);
\draw[fill] ($(L2C.east)!0.5!(L2C.text split)$) circle (0.05);
\draw[ptr]  ($(L2C.east)!0.5!(L2C.text split)$) --++(right:2mm) |- (L3C.west);

\node  (pa3) at ($(L1C.32)+(0,0.35)$) {\scriptsize{\code{pa}}};
\node  (vad3) at ($(L2C.148)+(0,0.35)$) {\scriptsize{$v_\mathit{ad}$}};

\path[ptr]  ($(pa3)-(0,0.1)$) -- (L1C.32);
\path[ptr]  ($(vad3)-(0,0.13)$) -- (L2C.148);
\draw[red,decorate,decoration={brace,amplitude=4pt}] ($(L2C.south west)+(2cm,-3pt)$) -- ($(L2C.south west)+(0,-3pt)$) node[midway,below,yshift=-2.5pt]{\scriptsize{$l$}};

%---------- D ----------
\node[] (distanceD) at (0,-\distD) {};
\node[label={[label distance=-0.65cm]0:{\scriptsize{$v^\mathit{start}$}}}]  (vstart4) at (distanceD) {};
\node[data, right=0.42cm of vstart4] (L1D) {\scriptsize{5} \nodepart{second}};
\node[data, right=0.2cm of L1D]   (L2D) {\scriptsize{2} \nodepart{second}};
\node[right=0.2cm of L2D]   (L3D) {\scriptsize{\ldots}};

\path[ptr]  ($(vstart4)+(0.25,0)$) --++(right:2mm)  |- (L1D.text west);
\draw[fill] ($(L1D.east)!0.5!(L1D.text split)$) circle (0.05);
\draw[ptr]  ($(L1D.east)!0.5!(L1D.text split)$) --++(right:2mm) |- (L2D.text west);
\draw[fill] ($(L2D.east)!0.5!(L2D.text split)$) circle (0.05);
\draw[ptr]  ($(L2D.east)!0.5!(L2D.text split)$) --++(right:2mm) |- (L3D.west);

\node  (pa4) at ($(L1D.32)+(0,0.35)$) {\scriptsize{\code{pa}}};
\node  (vad4) at ($(L2D.148)+(0,0.35)$) {\scriptsize{$v_\mathit{ad}$}};

\path[ptr]  ($(pa4)-(0,0.1)$) -- (L1D.32);
\path[ptr]  ($(vad4)-(0,0.13)$) -- (L2D.148);
\draw[red,decorate,decoration={brace,amplitude=4pt}] ($(L2D.south west)+(2cm,-3pt)$) -- ($(L1D.south west)+(0,-3pt)$) node[midway,below,yshift=-2.5pt]{\scriptsize{$l'$}};

%------ gen path -------
\coordinate (c1) at ($(L3A.east)+(0.2,0)$);
\draw[gen-edge,rounded corners=3pt] (L3A.east) -- (c1) -- node[xshift=5pt,yshift=15pt,rotate=270]{\scriptsize{store 5}} (c1 |- L3B.east) -- (L3B.east);

\coordinate (c2) at ($(L3C.east)+(0.2,0)$);
\draw[gen-edge,rounded corners=3pt] (L3C.east) -- (c2) -- node[xshift=5pt,yshift=21pt,rotate=270]{\scriptsize{store $v_\mathit{ad}$}} (c2 |- L3D.east) -- (L3D.east);
\end{tikzpicture}
\vspace*{-.6cm}
\end{wrapfigure}
the address $\mathit{t}$ (if $m = j$).
In all other cases,
the ordinary \code{store} rule is applied.

More precisely, 
let the list invariant $l$ describe a list of length $v_l$ at the address
$v_\mathit{ad}$. Then $l$ 
is 
replaced by a new list invariant $l'$ which describes the list
at the address $v^{\mathit{start}}$ after
storing $t$ at the address \code{pa}. Irrespective of whether $m \neq j$ or $m = j$,
the resulting list at  $v^{\mathit{start}}$
has the list at $v_\mathit{ad}$ as its ``tail'' and thus, its length  $v_\ell'$ is 
$v_\ell +1$.
We prevent sharing of different
 elements by removing the allocation $\alloc{v^\mathit{start}}{v^\mathit{end}}$ of the
 list and all
 points-to information of pointers
in \pagebreak $\alloc{v^\mathit{start}}{v^\mathit{end}}$.

\noindent
\fbox{
\begin{minipage}{11.8cm}
\label{rule:extension}
\small
\mbox{\small \textbf{\hspace*{-.15cm}list extension ($p:$ ``\code{store ty $t$, ty* pa}'',
    $t \in \Ids \cup \N$, \code{pa} $\in \Ids$)}}\\
\vspace*{-.2cm}\\
\centerline{$\frac{\parbox{6.6cm}{\centerline{
$s = (\Pos, \; \LV, \; \AL, \; \PT, \; \LI, \; \KB)$}\vspace*{.1cm}}}
{\parbox{8cm}{\vspace*{.1cm} \centerline{
$s' = (\Pos^+, \; \LV, \; \AL \backslash \{\alloc{v^\mathit{start}}{v^\mathit{end}}\}, \; \PT', \; \LI\backslash\{\listinv\}\cup\{\listinv'\}, \; \KB')$}}}\;\;\;\;$
\mbox{if}} \vspace*{-.15cm}
{\small
\begin{itemize}
\item[$\bullet$] there is $\listinv = (v_\mathit{ad} \pointstorec[\codevar{lty}]{v_\ell}
  [(\mathit{off}_i: \codevar{lty}_i: {w_i..\hat{w}_i})]_{i=1}^n) \in \LI$ with $\codevar{lty}_j = \codevar{lty*}$
\item[$\bullet$] there is $\alloc{v^\mathit{start}}{v^\mathit{end}} \in \AL$ with $\models \, \stateformula{s} \implies v^\mathit{end} = v^\mathit{start} + \sizeOf(\codevar{lty})-1$
\item[$\bullet$] there exists $1 \leq m \leq n$ such that $\codevar{ty} = \codevar{lty}_m$ and $\models \, \stateformula{s} \implies \LV(\codevar{pa}) = v^\mathit{start} + \mathit{off}_m$
\item[$\bullet$]  $\models
  \, \stateformula{s} \implies v_\mathit{ad} = v_j$ if $m \neq j$ and $\models \,
  \stateformula{s} \implies v_\mathit{ad} = \LV(t)$ if $m = j$
\item[$\bullet$] for all $1 \leq i \leq n$ with $i \neq m$ there exist $v^\mathit{start}_i,v_i \in \Vsym$\\with $\models \, \stateformula{s} \implies v^\mathit{start}_i = v^\mathit{start} + \mathit{off}_i$ and $(v^\mathit{start}_i \pointsto[\codevar{lty}_i] v_i) \in \PT$
\item[$\bullet$] $\PT' = \{(x_1 \hookrightarrow_{\codevar{sy}} x_2) \in \PT \mid\, \models
  \stateformula{s} \implies
(v^\mathit{end} < x_1) \vee (x_1+\sizeOf(\codevar{sy})-1 < v^\mathit{start}) \}$
\item[$\bullet$] $\listinv' = (v^\mathit{start} \pointstorec[\codevar{lty}]{v'_\ell} [(\mathit{off}_i: \codevar{lty}_i: {v_i..\hat{w}_i})]_{i=1}^n)$
\item[$\bullet$] $\KB' = \KB \,\cup\, \{v_m = \LV(t),\, v'_\ell = v_\ell+1\}$, where
   $v_m, v'_\ell$ are fresh
\end{itemize}}
\end{minipage}}

\vspace*{-.3cm}

\subsubsection{List Traversal}
\label{sect:List Traversal}
After the current element $y_\code{mem}$ is stored at $x_\code{tp}$ and the value
$x_\code{kinc}$ of \codevar{k} is incremented to $y_\code{kinc}$ and stored at
$x_\code{k\_ad}$, we reach a state $\GraphInitCmpZeroC$ at\linebreak position $(\cmpFor, 0)$ by the
branch instruction.
However, our already existing state $\GraphInitCmpZeroGen$ is more
general than $\GraphInitCmpZeroC$, i.e., we can draw a generalization edge from
$\GraphInitCmpZeroC$ to $\GraphInitCmpZeroGen$\linebreak
using the generalization rule with the
instantiation $\mu_\GraphInitCmpZeroC$ where $\mu_\GraphInitCmpZeroC(x_\codevar{mem}) = y_\codevar{mem}$,
$\mu_\GraphInitCmpZeroC(x_\codevar{nd}) = y_\codevar{nd}$,
$\mu_\GraphInitCmpZeroC(x_\codevar{cn}) = y_\codevar{cn}$,
$\mu_\GraphInitCmpZeroC(x_\codevar{k}) = x_\codevar{kinc}$,
$\mu_\GraphInitCmpZeroC(x_\codevar{kinc}) = y_\codevar{kinc}$,
$\mu_\GraphInitCmpZeroC(x_\ell) = y_\ell$,
$\mu_\GraphInitCmpZeroC(\hat{x}_\codevar{nd}) = \hat{x}_\codevar{nd}$, and
$\mu_\GraphInitCmpZeroC(x_\codevar{next}) = x_\codevar{mem}$. Thus, the cycle of the first loop closes
here.

\begin{wrapfigure}[7]{r}{7.81cm}
  \vspace*{-.75cm}
\hspace*{-.1cm}\begin{tikzpicture}[node distance = \ydist and \xdist]
\scriptsize
\def\widetwidth{6.2cm}
\def\smalltwidth{4.75cm}
\def\fulllinewidhth{11cm}
\def\edgenodedist{0.3cm}
\def\FirstIndentwidth{3cm}
\def\SecondIndentwidth{2cm}
\tikzset{invisible/.style={opacity=0}}
\tikzstyle{state}=[inner sep=2pt, font=\scriptsize, draw]

\node[state, align=left, label={[xshift=\labelxshift]2:$\GraphTravEntryZero$}] (1) 
{$(\initPtr, 0), \;
  \{\code{tail\_ptr} = x_{\code{tp}}, \,
    ...\}, \;
  \{...\}, \;
  \{
   x_{\code{tp}} \pointsto[\code{list*}] x_{\code{mem}}, \,
   ...\},$\\
 $\{
   x_{\code{mem}} \pointstorec[\code{list}]{x_{\ell}}
     [(0: \code{i32}: x_{\code{nd}}..\hat{x}_{\code{nd}}),
      (8: \code{list*}: x_{\code{next}}..0)]
   \}, \;
  \{
   ... \}$
};

\node[state, below=of 1, align=left, label={[xshift=\labelxshift]2:$\GraphTravCmpZeroTwo$}] (2) 
{$(\cmpWhile, 0),\;
  \{
   \code{ptr} = x_{\code{ptr}}, \,
   \code{curr'} = x_{\code{mem}}, \,
   \code{next\_ptr} = x_{\code{np}},$ \\
$\code{next} = x_{\code{next}}, \,
   ...\}, \;
 \{
   \alloc{x_{\code{ptr}}}{x_\code{ptr}^\mathit{end}}, \,
   ...\}, \;
  \{
   x_{\code{ptr}} \pointsto[\code{list*}] x_{\code{next}}, \,
   ...\},$\\
 $\{
   x_{\code{mem}} \pointstorec[\code{list}]{x_{\ell}}
     [(0: \code{i32}: x_{\code{nd}}..\hat{x}_{\code{nd}}),
      (8: \code{list*}: x_{\code{next}}..0)]
   \},$\\
 $\{
   x_{\code{np}} = x_{\code{mem}} + 8, \,
   ... \}$
};

\draw[omit-edge] (1)  --  (2);

\end{tikzpicture}
\end{wrapfigure}
As mentioned,
in the path from $\GraphInitCmpZeroGen$ to $\GraphInitCmpZeroC$ there is a state at
position $(\cmpFor, 1)$ which is refined (similar to State $\GraphInitCmpOne$). If 
 \code{k $<$ n}
 holds, we reach $\GraphInitCmpZeroC$. The other
path with
\code{k $\not<$ n}
leads out of the

\vspace*{-.15cm}

\begin{wrapfigure}[12]{r}{6.9cm}
\vspace{-0.85cm}
\begin{boxedminipage}{6.8cm}
\scriptsize
\begin{Verbatim}[commandchars=\\\{\}]
\initPtr:
  0: tail' = load list*, list** tail_ptr
  1: store list* tail', list** ptr
  2: br label \cmpWhile

\cmpWhile:
  0: str = load list*, list** ptr
  1: notnull = icmp ne list* str, null
  2: br i1 notnull, label \bodyWhile, label ret

\bodyWhile:
  0: curr' = bitcast list* str to i8*
  1: next_ptr = getelementptr i8, i8* curr', i64 8
  2: next_ptr' = bitcast i8* next_ptr to list**
  3: next = load list*, list** next_ptr'
  4: store list* next, list** ptr
  5: br label \cmpWhile
\end{Verbatim}
\end{boxedminipage}
\vspace{-.05cm}
\end{wrapfigure}
\noindent
\code{for} loop to the
block \initPtr{} followed by the \code{while} loop (see State $\GraphTravEntryZero$ 
and the corresponding \LLVM{} code on the side).
The\linebreak value $x_\codevar{mem}$
at address \codevar{tail\_ptr} is loaded to \code{tail'} and stored at a new pointer variable \codevar{ptr}.
State $\GraphTravCmpZeroTwo$ is reached after the first iteration of the \code{while} loop body. Here, block \cmpWhile{} loads the value $x_\codevar{mem}$ stored at \code{ptr} to
\code{str}. Since it is not the null pointer,  we enter \bodyWhile{}, which
corresponds to the body of the \code{while} loop. First, $x_\codevar{mem}$ is cast to an
\code{i8} pointer. Then  \code{getelementptr} computes a pointer
$x_\codevar{np}$ to the next element by adding 8 bytes to $x_\codevar{mem}$.
After another cast back to a $\code{list*}$ pointer, we load the content of the new
pointer to \code{next}. To this end, we need the following new variant of the \code{load}
rule to load values that are described by a list \pagebreak invariant.

\noindent
\fbox{
\begin{minipage}{11.8cm}
\label{rule:loadlist}
\small
\mbox{\small \textbf{\hspace*{-.15cm}\code{load} from list invariant ($p:$ ``\code{x = load ty, ty* ad\_i}'', $\codevar{x},\codevar{ad\_i} \in \Ids$)}}\\
\vspace*{-.2cm}\\
\centerline{$\frac{\parbox{6.6cm}{\centerline{
$s = (\Pos, \; \LV, \; \AL, \; \PT, \; \LI, \; \KB)$}\vspace*{.1cm}}}
{\parbox{7.4cm}{\vspace*{.1cm} \centerline{
$s' = (\Pos^+, \; \LV[\code{x}:=w], \; \AL, \; \PT, \; \LI, \; \KB \cup \{w = v_i\})$}}}\;\;\;\;$
\mbox{if $w \in \Vsym$ is fresh and}} \vspace*{-.2cm}
{\small
\begin{itemize}
\item[$\bullet$] there is $\listinv = (v_\mathit{ad} \pointstorec[\codevar{ty}]{v_\ell}
  [(\mathit{off}_i: \codevar{ty}_i: {v_i..\hat{v}_i})]_{i=1}^n) \in \LI$
\item[$\bullet$] there exists $1 \leq i \leq n$ such that $\codevar{ty} = \codevar{ty}_i$ and $\models \, \stateformula{s} \implies \LV(\code{ad\_i}) = v_\mathit{ad} + \mathit{off}_i$
\end{itemize}}
\end{minipage}}
\vspace*{.07cm}

With this new \code{load} rule, the content of the new pointer is identified as
$x_\codevar{next}$. It
is loaded to \code{next} and stored at $x_\codevar{ptr}$. Then we return to the
block \cmpWhile{} (State $\GraphTravCmpZeroTwo$). Merging $\GraphTravCmpZeroTwo$ with its
predecessor at the same program position is not possible yet since the domains of the
respective $\LV$
functions do not coincide.
  Now, $x_\codevar{next}$ is loaded to
\code{str} and compared to the null pointer. Since we do not have information about
$x_\codevar{next}$, $\GraphTravCmpZeroTwo$'s successor state is refined
to a state with $x_\codevar{next} = 0$ (which starts a path out of the
loop to a return state), and to a state with $x_\codevar{next} \geq 1$, which reaches
$\GraphTravBodyGetelemStart$ after a few evaluation steps, see \cref{fig:Trav1}.
Now, \code{getelementptr}
computes the pointer $x_\codevar{np}'  = x_\codevar{next} + 8$ to the third element of the
list, which is assigned to  \code{next\_ptr}. $\stateformula{\GraphTravBodyGetelemStart}$ contains
$x_\ell \geq 2$ since the first and the last pointer value are known to be
different ($x_\codevar{next} \neq 0$). This information is crucial for creating a new list
invariant starting at $x_\codevar{next}$, which is used in the next iteration of the
loop. Therefore, if our list invariant did not contain variables for the first and the
last pointer, we could not prove termination of the program.
In such a case where the pointer to the third element
of a list invariant is computed and the length of the list is at least two, we
\emph{traverse} the list invariant to retain the correspondence between the
computed pointer $x'_\code{np}$
and the new list invariant.
In the resulting state $\GraphTravBodyGetelemRes$,
we represent the first list element by an allocation
$\alloc{x_\code{mem}}{x_\code{mem}^\mathit{end}}$ and preserve all knowledge about this
element that was
encoded 
in the list invariant
($x_\code{mem}^\mathit{end} = x_{\code{mem}} + 15$, $x_{\code{mem}}
\pointsto[\code{i32}] x_{\code{nd}}$, $x_{\code{np}} \pointsto[\code{list*}] x_{\code{next}}$).
Moreover, we adapt the list invariant such that it now represents the list at 
$x_\codevar{next}$ (i.e., without its first element) starting with the value 
$x'_\code{nd}$. We also relate the length of the
new list invariant to the length of the former  one ($x'_{\ell} = x_{\ell} -
1$).

\footnotesize

\begin{figure}[t]
\centering
\begin{tikzpicture}[node distance = \ydist and \xdist]
\scriptsize
\def\widetwidth{6.2cm}
\def\smalltwidth{4.75cm}
\def\fulllinewidhth{11cm}
\def\edgenodedist{0.2cm}
\def\FirstIndentwidth{3cm}
\def\SecondIndentwidth{2cm}
\tikzset{invisible/.style={opacity=0}}
\tikzstyle{state}=[
						   %minimum size=10pt,
                           %fill=white,
                           %shape=rectangle,
                           %text=black,
                           inner sep=2pt,
                           font=\scriptsize,
                           draw]

\node[state, align=left, label={[xshift=\labelxshift]2:$\GraphTravBodyGetelemStart$}] (3) 
{$(\bodyWhile, 1),\;
  \{
   \code{ptr} = x_{\code{ptr}}, \,
   \code{curr'} = x_{\code{next}}, \,
   \code{next\_ptr} = x_{\code{np}}, \,
   \code{next} = x_{\code{next}}, \,
   ...\},$\\
 $\{
   \alloc{x_{\code{ptr}}}{x_\code{ptr}^\mathit{end}}, \,
   ...\}, \;
  \{
   x_{\code{ptr}} \pointsto[\code{list*}] x_{\code{next}}, \,
   ...\},$\\
 $\{
   x_{\code{mem}} \pointstorec[\code{list}]{x_{\ell}}
     [(0: \code{i32}: x_{\code{nd}}..\hat{x}_{\code{nd}}),
      (8: \code{list*}: x_{\code{next}}..0)]
   \},$\\
 $\{
%   x_{\code{kinc}} = x_{\code{k}} + 1, \,
%   x_{\ell} = x_{\code{kinc}}, \,
   x_{\code{np}} = x_{\code{mem}} + 8, \,
   x_{\code{next}} \geq 1, \,
   ... \}$
};

\node[state, below=of 3, align=left, label={[xshift=\labelxshift]2:$\GraphTravBodyGetelemRes$}] (4) 
{$(\bodyWhile, 2),\;
  \{
   \code{ptr} = x_{\code{ptr}}, \,
   \code{curr'} = x_{\code{next}}, \,
   \code{next\_ptr} = x'_{\code{np}}, \,
   \code{next} = x_{\code{next}}, \,
   ...\},$\\
 $\{
   \alloc{x_{\code{ptr}}}{x_\code{ptr}^\mathit{end}}, \,
   \alloc{x_{\code{mem}}}{x_\code{mem}^\mathit{end}}, \,
   ...\}, \;
  \{
   x_{\code{ptr}} \pointsto[\code{list*}] x_{\code{next}}, \,
   x_{\code{mem}} \pointsto[\code{i32}] x_{\code{nd}}, \,
   x_{\code{np}} \pointsto[\code{list*}] x_{\code{next}}, \,
   ...\},$\\
 $\{
   x_{\code{next}} \pointstorec[\code{list}]{x'_{\ell}}
     [(0: \code{i32}: x'_{\code{nd}}..\hat{x}_{\code{nd}}),
      (8: \code{list*}: x'_{\code{next}}..0)]
   \},$\\
 $\{
%   x_{\ell} = x_{\code{kinc}}, \,
 x_{\code{np}} = x_{\code{mem}} + 8, \,
   x_\code{mem}^\mathit{end} = x_{\code{mem}} + 15, \,
   x'_{\code{np}} = x_{\code{next}} + 8, \,
   x'_{\ell} = x_{\ell} - 1, \,
   ... \}$
};

\draw[eval-edge] (3)  --  (4);

\end{tikzpicture}
\vspace*{-.1cm}
\caption{Traversing a List Invariant}
\label{fig:Trav1}
\vspace*{-.7cm}
\end{figure}
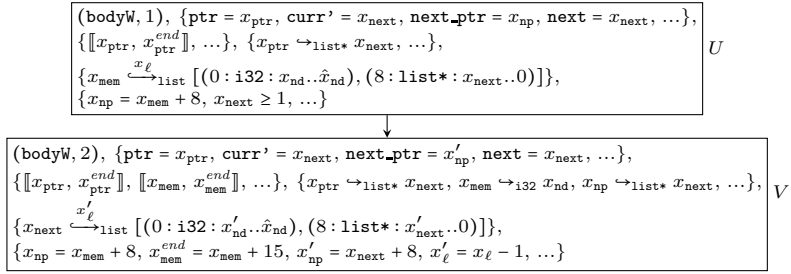
\normalsize

Thus, 
in addition to the rule for  \code{getelementptr} in \cite{LLVM-JAR17}, we
now introduce rules for list traversal via \code{getelementptr}. The rule below
handles the case where the address
calculation is based on the type \code{i8} and the \code{getelementptr} instruction
adds the number of bytes given by the term
$t$ to the address \code{pa}.
Here, the offsets in our list invariants are needed to compute the address of the accessed
field.
We also have similar rules for list traversal via \vspace*{-.3cm}\pagebreak field access (i.e.,
where the next element is accessed using \code{curr'->next} as in the
\code{for} loop) and for the case where we cannot prove that the length $v_{\ell}$ of the
list is at least 2, see\arxiv{ App.\ \ref{app:rules}}\paper{ \cite{report}}.

\vspace*{.05cm}
\noindent
\fbox{
\begin{minipage}{11.8cm}
\label{rule:traversal}
\small
\mbox{\small \textbf{\hspace*{-.2cm}list traversal ($p:$ ``\code{pb = getelementptr i8, i8* pa, i$m$ $t$}'', {\scriptsize $t \in \Ids \cup \N$, $\code{pa},\code{pb}\in \Ids$})}}\vspace*{.2cm}\\
\vspace*{-2mm}
\centerline{$\frac{\parbox{6.6cm}{\centerline{
$s = (\Pos, \; \LV, \; \AL, \; \PT, \; \LI, \; \KB)$}\vspace*{.1cm}}}
{\parbox{10.2cm}{\vspace*{.1cm} \centerline{
$s' = (\Pos^+, \; \LV[\code{pb}:=w^\mathit{start}_j], \; \AL \cup \alloc{v^\mathit{start}}{v^\mathit{end}}, \; \PT', \; \LI\backslash\{\listinv\}\cup\{\listinv'\}, \; \KB')$}}}\;\;\;\;$
\mbox{if}} \vspace*{-.1cm}
{\small
\begin{itemize}
\item[$\bullet$] there is $\listinv = (v_\mathit{ad} \pointstorec[\codevar{ty}]{v_\ell}
  [(\mathit{off}_i: \codevar{ty}_i: {v_i..\hat{v}_i})]_{i=1}^n) \in \LI$ with
$\codevar{ty}_j = \codevar{ty*}$,\\
  $\models
  \, \stateformula{s} \implies \LV(\codevar{pa}) = v_j$,  $\models \, \stateformula{s} \implies \LV(t) = \mathit{off}_j$, and $\models \, \stateformula{s} \implies v_\ell \geq 2$
\item[$\bullet$] $\PT' = \PT \cup \{(v^\mathit{start}_i \hookrightarrow_{\codevar{ty}_i} v_i) \mid 1 \leq i \leq n \}$
\item[$\bullet$] $\listinv' = (w^\mathit{start} \pointstorec[\codevar{ty}]{w_\ell} [(\mathit{off}_i: \codevar{ty}_i: {w_i..\hat{v}_i})]_{i=1}^n)$
\item[$\bullet$] $\KB' = \KB \cup \{v^\mathit{start} = v_\mathit{ad},\, v^\mathit{end} = v^\mathit{start} + \sizeOf(\code{ty}) -1,\, w^\mathit{start} = v_j,\, w_\ell = v_\ell-1,\\ w^\mathit{start}_j = w^\mathit{start} + \mathit{off}_j\} \cup \{v^\mathit{start}_i = v_\mathit{ad} + \mathit{off}_i \mid 1 \leq i \leq n \}$
\item[$\bullet$] $v^\mathit{start}, v^\mathit{end}, v^\mathit{start}_1, \ldots, v^\mathit{start}_n, w^\mathit{start}, w_\ell, w^\mathit{start}_j, w_1, \ldots, w_n \in \Vsym$ are fresh
\end{itemize}}
\end{minipage}}
\vspace*{.05cm}

We continue the symbolic execution of State $\GraphTravBodyGetelemRes$ in our example
and finally obtain a complete SEG with a path
from a state $W$ at the position 
$(\cmpWhile, 0)$ to the next
state $W'$ at this position, and a generalization edge back from $W'$ to $W$ using an
instantiation $\mu_{W'}$. Both $W$ and $W'$
contain a list invariant similar to $T$ where
instead of the length $x_\ell$ in $T$, we have the symbolic variables $z_\ell$ and $z_\ell'$
in $W$ and $W'$, where $\mu_{W'}(z_\ell) = z_\ell'$
(see\arxiv{ App.\ \ref{app:leadingExample}}\paper{ \cite{report}} for more details).

\vspace*{-.6cm}

\footnotesize

\begin{figure}[h]
\centering
\begin{tikzpicture}[node distance = \ydist and \xdist]
\scriptsize
\def\widetwidth{6.2cm}
\def\smalltwidth{4.75cm}
\def\fulllinewidhth{11cm}
\def\edgenodedist{0.3cm}
\def\FirstIndentwidth{3cm}
\def\SecondIndentwidth{2cm}
\tikzset{invisible/.style={opacity=0}}
\tikzstyle{state}=[
						   %minimum size=10pt,
                           %fill=white,
                           %shape=rectangle,
                           %text=black,
                           inner sep=2pt,
                           font=\tiny,
                           draw]

\node[state, align=left, label={[xshift=-.1cm]12:$\GraphTravCmpZeroThree$}] (3) 
{$(\cmpWhile, 0),\;
  \{
   \code{ptr} = z_{\code{ptr}}, \,
   \code{curr'} = z'_{\code{next}}, \,
   \code{next\_ptr} = z''_{\code{np}}, \,
   \code{next} =$\\
   $z''_{\code{next}}, \,
   ...\}, \;
 \{
   \alloc{z_{\code{ptr}}}{z_\code{ptr}^\mathit{end}},
   ...\},
   \{
   z_{\code{ptr}} \pointsto[\code{list*}] z''_{\code{next}}, \,
   ...\},$\\
 $\{
   z_{\code{mem}} \pointstorec[\code{list}]{z_{\ell 1}}
     [(0\!:\!\code{i32}\!:\!z_{\code{nd}}..{z'_{\code{nd}}}),
      (8\!:\!\code{list*}\!:\!z_{\code{next}}..{z'_\code{next}})],$\\
  \;\,$z'_{\code{next}} \pointstorec[\code{list}]{z_{\ell}}
     [(0: \code{i32}: z''_\code{nd}..\hat{z}_\code{nd}),
      (8: \code{list*}: z''_{\code{next}}..0)]
   \},$\\
 $\{
   z''_{\code{np}} = z'_{\code{next}} + 8, \,
 %  z_{\ell 1} \geq 1, \,
   z_{\ell} = z'_{\ell} - 1, \,
 %  z_{\ell} \geq 1, \,
   ... \}$
};

\node[state, right=of 3, align=left, label={[xshift=.1cm]190:$\GraphTravCmpZeroFour$}] (6) 
{$(\cmpWhile, 0),\;
  \{
   \code{ptr} = z_{\code{ptr}}, \,
   \code{curr'} = z''_{\code{next}}, \,
   \code{next\_ptr} = z'''_{\code{np}}, \,
   \code{next} =$\\
  $z'''_{\code{next}}, \,
   ...\}, \;
 \{
   \alloc{z_{\code{ptr}}}{z_\code{ptr}^\mathit{end}},
   ...\},
   \{
   z_{\code{ptr}} \pointsto[\code{list*}] z'''_{\code{next}}, \,
   ...\},$\vspace*{-.1cm}\\
 $\{
   z_{\code{mem}} \pointstorec[\code{list}]{z'_{\ell 1}}
     [(0\!:\!\code{i32}\!:\!z_{\code{nd}}..{z''_\code{nd}}),
      (8\!:\!\code{list*}\!:\!z_{\code{next}}..{z''_\code{next}})],$\vspace*{-.1cm}\\
  \;\,$z''_{\code{next}} \pointstorec[\code{list}]{z'_{\ell}}
     [(0: \code{i32}: z'''_\code{nd}..\hat{z}_\code{nd}),
      (8: \code{list*}: z'''_{\code{next}}..0)]
   \},$\\
 $\{
 %  z_{\ell 1} \geq 1, \,
 %  z_{\ell} \geq 1, \,
   z'''_{\code{np}} = z''_{\code{next}} + 8, \,
   z'_{\ell 1} = z_{\ell 1} + 1, \,
   z'_{\ell} = z_{\ell} - 1, \,
   ... \}$
};

%\draw[omit-edge,bend left] (3)  --  ($(6) + (0,2)$);
\draw[omit-edge] (3) to[out=10,in=170]  (6);
\draw[gen-edge] (6) to[out=190,in=-10]  (3);

%% \coordinate (gen-edge) at ($(6.west)+(0,0)$);
%% \draw[gen-edge,rounded corners=5pt] (6.west) -- (gen-edge) -- (gen-edge |- 3.east) --
%% (3.east);

\end{tikzpicture}
\end{figure}
\normalsize

\vspace*{-.8cm}

\section{Proving Termination}
\label{sect:Proving Termination}

To prove termination of a program $\Prog$,  as
 in \cite{LLVM-JAR17} the cycles of the SEG are
translated to an integer transition system whose termination implies termination of $\Prog$.
 The edges of the SEG are transformed
 into ITS
 transitions whose application
 conditions 
 consist of the state formulas $\stateformula{s}$ and equations to identify corresponding
 symbolic variables of the different states. For evaluation and refinement edges, the 
 symbolic variables do not change. For generalization edges, we use the instantiation $\mu$ to identify corresponding symbolic variables. 
 In our example, the ITS has cyclic transitions of the following form:

\vspace*{-.45cm}

 \[ \mbox{\footnotesize $\begin{array}{rlll}
 O(x_{\code{n}}, x_{\code{k}}, x_{\code{kinc}}, \ldots) &\to^+& R(x_{\code{n}}, x_{\code{k}},
 x_{\code{kinc}}, \ldots)& \quad \mid \quad
 x_{\code{kinc}} = x_{\code{k}} + 1  \wedge  x_{\code{n}} > x_{\code{k}} \wedge  \ldots \\
 R(x_\code{n}, x_\code{k},
 x_\code{kinc}, \ldots) &\to& O(x_\code{n},  x_\code{kinc}, \ldots) \\
 W(z_\ell, z_\ell', \ldots) &\to^+ & W'(z_\ell, z_\ell', \ldots) &\quad \mid \quad z_\ell =
 z_\ell'-1  \wedge z_\ell \geq 1 \wedge\ldots \\  
 W'(z_\ell, z_\ell', \ldots) &\to& W(z_\ell', \ldots) 
 \end{array}$} 
 \]

\vspace*{-.05cm}

\noindent
The first cycle resulting from the generalization edge from
$\GraphInitCmpZeroC$ to $\GraphInitCmpZeroGen$ 
 terminates since \code{k} is increased until it reaches
\code{n}. The generalization edge yields a condition identifying
$x_\codevar{kinc}$ in $\GraphInitCmpZeroC$ with $x_\codevar{k}$
in $\GraphInitCmpZeroGen$, since $\mu_\GraphInitCmpZeroC(x_\codevar{k}) =
x_\codevar{kinc}$.
With the condi\-tions $x_\code{kinc} = x_\code{k} + 1$ and
$x_\code{n} > x_\code{k}$ (from $\KB^\GraphInitCmpZeroGen$), the resulting transitions of
the
ITS are terminating.
 The second cycle from the generalization edge from
$\GraphTravCmpZeroFour$ to
$\GraphTravCmpZeroThree$
 terminates since the length of the list starting with
\code{curr'} decreases. \vspace*{-.2cm}\pagebreak Although there is no program
variable for the length, due to our list invariants the states contain variables for this length, which are also passed
to the ITS.  Thus, the ITS contains the variable
$z_{\ell}$ (where $z_{\ell}$ in $\GraphTravCmpZeroThree$ is identified with 
$z'_{\ell}$ in  $\GraphTravCmpZeroFour$ due to 
$\mu_{\GraphTravCmpZeroFour}(z_{\ell}) = z'_{\ell}$). Since the condition
$z'_{\ell} = z_{\ell} - 1$ is obtained on the path from
 $\GraphTravCmpZeroThree$ to
 $\GraphTravCmpZeroFour$ and
$z_{\ell} \geq 1$ is part of $\stateformula{\GraphTravCmpZeroThree}$ due to the
list invariant with length $z_{\ell}$ in $\LI^\GraphTravCmpZeroThree$,
 the resulting transitions of the  ITS clearly terminate.
 Analogous to \cite[Cor.\ 11 and Thm.\ 13]{LLVM-JAR17}, we obtain the following theorem.
To prove that a complete SEG represents all program paths,
in \cite{LLVM-JAR17} we used the \LLVM{} semantics defined by the \tool{Vellvm} project
\cite{Vellvm}.
One now also has to prove soundness of those symbolic execution rules which 
    were modified due to the new concept of list invariants (i.e., generalization, 
    list extension, and list traversal), see\arxiv{ App.\ \ref{app:proofs}.}\paper{ \cite{report}.}

\setcounter{thm-termination}{\value{theorem}}

\begin{theorem}[Memory Safety and Termination]\label{thm-termination}
  Let $\Prog$ be a program with a complete SEG
  $\GG$.
Since a complete
  SEG does not contain $\ERROR$,  $\Prog$ is memory safe for all concrete states
  represented by the states in $\GG$.\footnote{Our approach can only
\emph{prove} but not \emph{disprove} memory safety, i.e., a SEG with the state $\ERROR$
just means that we failed in showing memory safety.}
  If the ITS corresponding to $\GG$ is terminating, then $\Prog$ is also terminating for all states represented by $\GG$.
\end{theorem}

\section{Conclusion, Related Work, and Evaluation}
\label{sect:evaluation}

We presented a new approach for automated proofs of memory safety and termination of
\tool{C}/\tool{LLVM} programs on lists.  It first constructs a symbolic execution graph (SEG) which
overapproximates all
program runs. Afterwards, an integer transition system (ITS) is generated from this graph whose
termination is proved using standard techniques. The main idea of our new approach is
the
extension of the  states in the SEG by suitable \emph{list
invariants}. We 
developed techniques to infer and  modify list invariants
automatically during the symbolic execution.

During the construction of the SEG, the list invariants 
abstract from a concrete number of memory allocations to a list of allocations of variable
length while preserving knowledge about some of the contents (the values  of the fields of the first and
the last element) and the list shape (the start address of the first element, the list length, and the content of the
last recursive pointer which allows us to distinguish between cyclic and acyclic lists).
They also contain information on
the memory arrangement of the list fields which is needed for  programs that
access fields
via
pointer arithmetic.
The symbolic variables for  the list length and the first and last values of list
elements are preserved when generating an ITS from the SEG.
 Thus, they
  can be used in the termination proof of the ITS (e.g., the variables for list length
can occur in ranking functions).

In \cite{RTA11,CAV12,RTA10} we
developed a technique for
termination analysis of \tool{Java},
based on a program transformation to \emph{integer term rewrite systems} 
instead of ITSs.  This approach does not require specific
list
invariants as
recursive data structures on\linebreak the heap are abstracted to terms. However, these terms are
unsuitable for \sfC{}, since they cannot express 
memory allocations and the connection to 
their contents.

Separation logic predicates for termination of list programs were also used in\linebreak
\cite{Mutant},  but 
their list \pagebreak predicates only consider the list length and
the recursive field,  but no other fields or offsets.  
The tools \tool{Cyclist} \cite{Cyclist} and \tool{HipTNT+} \cite{HipTNT} are integra\-ted
in separation logic systems which also allow to define heap predicates.
However, they require annotations and
hints which parameters of the list predicates are needed as a
termination measure.
The tool \tool{2LS} \cite{2LS} also provides basic support for dynamic data
structures.
But all these approaches are not suitable if
ter\-mination depends on the contents or the shape of data structures combined with\linebreak pointer
arithmetic.
In \cite{David15}, programs can be annotated with 
arithmetic and
structural properties to reason about termination.  In contrast,
our approach does not need hints or annotations, but finds termination arguments fully
automatically.

We implemented our approach in \aprove{} \cite{LLVM-JAR17}.
While \sfC{} programs with  lists are very common, existing tools can hardly prove their
termination. Therefore, 
the current benchmark collections for termination analysis contain almost no list
programs. In 2017, a benchmark
set\footnote{\url{https://github.com/sosy-lab/sv-benchmarks/tree/master/c/termination-memory-linkedlists}}
of 18 typical \sfC{} programs
 on 
lists was added to the \emph{Termination} category of the
\emph{Competition on Software Verification}
(\emph{SV-COMP}) \cite{SVCOMP22}, where 9 of them are terminating.
Two of these 9 programs do not need
list invariants, because they
just  create
a list without operating on it
afterwards. The remaining seven terminating programs create a list and then
traverse it, search for a value, or append lists and compute the length
afterwards. Only few tools in \emph{SV-COMP} 
produced correct
termination proofs for programs from this set:
\tool{HipTNT+} and \tool{2LS} failed for all of them.
\tool{CPAchecker} \cite{CPAchecker} and \tool{PeSCo} \cite{PeSCo}
proved termination and non-termination
for one of these programs in 2020.  \tool{UAutomizer} \cite{UAutomizer}
proved termination for two
and non-termination for seven programs.
The termination proofs
of \tool{CPAchecker}, \tool{PeSCo}, and \tool{UAutomizer}
only concern the
programs that just
create a list. 
Our new version of \aprove{} is the only termination prover\footnote{We did not compare
with the tool \tool{VeriFuzz} \cite{Met+23}, since it
does not prove termination but only tests
for non-termination and thus, it is unsound for  inferring termination.}
that succeeds if
termination depends on the shape or contents of a list after its
creation.
Note that for non-termination, a proof is a single non-terminating
program path, so here list invariants are less helpful.

For the \emph{Termination Competition} \cite{TermComp} 2022, we submitted 18
terminating \sfC{}\linebreak programs on
lists\footnote{\url{https://github.com/TermCOMP/TPDB/tree/master/C/Hensel_22}} (different
from the ones at \emph{SV-COMP}),
where two of them just create a list. Three  traverse it
afterwards (by a loop or recursion),
and ten search for a value, where
 for nine, also the list contents 
are relevant for termination.
Three programs perform common
operations like inserting or deleting an element.
\tool{UAutomizer} proves termination for a program that just creates a list but not for programs operating on the list afterwards.
With our approach, \aprove{} succeeds on 17 of the 18 programs.
 Overall, \aprove{}  and \tool{UAutomizer} were the two most powerful tools 
for termination
of \sfC{}  in  \emph{SV-COMP} 2022 and the \emph{Termination Competition} 2022, with
\tool{UAutomizer}  winning the former and \aprove{}  winning
\begin{wrapfigure}[3]{r}{7cm}
 \footnotesize
 \renewcommand{\arraystretch}{.9}
\vspace*{-.75cm}
 \begin{tabular}{|c|c|c|c|}
    \cline{2-4}
    \multicolumn{1}{c|}{} &  {\scriptsize \textbf{SV-C T.}} 
                          &  {\scriptsize \textbf{SV-C Non-T.}} 
                          &  {\scriptsize \textbf{TermCmp T.}}\\
    \hline
    \aprove{}             & 7 (of 9) & 5 (of 9) & 17 (of 18)\\
    \hline
    \tool{UAutomizer}     & 2 (of 9) & 7 (of 9) &  1 (of 18)\\
    \hline
   \end{tabular}
 \renewcommand{\arraystretch}{1}
\end{wrapfigure}
 the
latter competition. To download  \aprove{}, run it via its web interface,
 and for
 details on our experiments, see \url{https://aprove-developers.github.io/recursive_structs}.\vspace*{-.3cm}\pagebreak

\bibliographystyle{plainurl}
\paper{\bibliography{references}}

\begin{thebibliography}{10}

\bibitem{Mutant}
J.\ Berdine, B.\ Cook, D.\ Distefano, and P.~W.\ O'Hearn.
\newblock Automatic termination proofs for programs with shape-shifting heaps.
\newblock In {\em Proc.\ CAV~'06}, LNCS 4144, pages 386--400, 2006.
\newblock \href {https://doi.org/10.1007/11817963_35}
  {\path{doi:10.1007/11817963_35}}.

\bibitem{CPAchecker}
D.\ Beyer and M.~E.\ Keremoglu.
\newblock {\tool{CPAchecker}}: A tool for configurable software verification.
\newblock In {\em Proc.\ CAV~'11}, LNCS 6806, pages 184--190, 2011.
\newblock \href {https://doi.org/10.1007/978-3-642-22110-1_16}
  {\path{doi:10.1007/978-3-642-22110-1_16}}.

\bibitem{SVCOMP22}
D.~Beyer\noopsort{1}.
\newblock Progress on software verification: \emph{SV-COMP 2022}.
\newblock In {\em Proc.\ TACAS~'22}, LNCS 13244, pages 375--402, 2022.
\newblock For the results of \emph{SV-COMP}~'22, see
  \url{https://sv-comp.sosy-lab.org/2022/}.
\newblock \href {https://doi.org/10.1007/978-3-030-99527-0_20}
  {\path{doi:10.1007/978-3-030-99527-0_20}}.

\bibitem{QSL}
M.\ Bozga, R.\ Iosif, and S.\ Perarnau.
\newblock Quantitative separation logic and programs with lists.
\newblock {\em J.\ Aut.\ Reasoning}, 45(2):131--156, 2010.
\newblock \href {https://doi.org/10.1007/s10817-010-9179-9}
  {\path{doi:10.1007/s10817-010-9179-9}}.

\bibitem{RTA11}
M.\ Brockschmidt\noopsort{4}, C.\ Otto, and J.\ Giesl.
\newblock Modular termination proofs of recursive {{\tool{Java Bytecode}}}
  programs by term rewriting.
\newblock In {\em Proc.\ RTA~'11}, LIPIcs 10, pages 155--170, 2011.
\newblock \href {https://doi.org/10.4230/LIPIcs.RTA.2011.155}
  {\path{doi:10.4230/LIPIcs.RTA.2011.155}}.

\bibitem{CAV12}
M.\ Brockschmidt\noopsort{5}, R.\ Musiol, C.\ Otto, and J.\ Giesl.
\newblock Automated termination proofs for {{\tool{Java}}} programs with cyclic
  data.
\newblock In {\em Proc.\ CAV~'12}, LNCS 7358, pages 105--122, 2012.
\newblock \href {https://doi.org/10.1007/978-3-642-31424-7_13}
  {\path{doi:10.1007/978-3-642-31424-7_13}}.

\bibitem{T2}
M.\ Brockschmidt\noopsort{9}, B.\ Cook, S.\ Ishtiaq, H.\ Khlaaf, and N.\
  Piterman.
\newblock \textsf{T2}: Temporal property verification.
\newblock In {\em Proc.\ TACAS~'16}, LNCS 9636, pages 387--393, 2016.
\newblock \href {https://doi.org/10.1007/978-3-662-49674-9_22}
  {\path{doi:10.1007/978-3-662-49674-9_22}}.

\bibitem{UAutomizer}
Y.-F.\ Chen, M.\ Heizmann, O.\ Leng\'{a}l, Y.\ Li, M.-H.\ Tsai, A.\ Turrini,
  and L.\ Zhang.
\newblock Advanced automata-based algorithms for program termination checking.
\newblock In {\em Proc.\ PLDI~'18}, pages 135--150, 2018.
\newblock \href {https://doi.org/10.1145/3192366.3192405}
  {\path{doi:10.1145/3192366.3192405}}.

\bibitem{Clang}
\textsf{Clang}: \url{https://clang.llvm.org}.

\bibitem{David15}
C.\ David, D.\ Kroening, M.\ Lewis, and J.\ Vitek.
\newblock Propositional reasoning about safety and termination of
  heap-manipulating programs.
\newblock In {\em Proc.\ ESOP~'15}, LNCS 9032, pages 661--684, 2015.
\newblock \href {https://doi.org/10.1007/978-3-662-46669-8_27}
  {\path{doi:10.1007/978-3-662-46669-8_27}}.

\bibitem{ASV19}
F.\ Emrich, J.\ Hensel, and J.\ Giesl.
\newblock {{\textsf{AProVE}}}: Modular termination analysis of
  memory-manipulating {{\textsf{C}}} programs.
\newblock {\em CoRR}, abs/2302.02382, 2023.
\newblock URL: \url{https://arxiv.org/abs/2302.02382}.

\bibitem{loat}
F.\ Frohn and J.\ Giesl.
\newblock Proving non-termination via loop acceleration.
\newblock In {\em Proc.\ {FMCAD}~'19}, pages 221--230, 2019.
\newblock \href {https://doi.org/10.23919/FMCAD.2019.8894271}
  {\path{doi:10.23919/FMCAD.2019.8894271}}.

\bibitem{LoATIJCAR22}
F.\ Frohn and J.\ Giesl.
\newblock Proving non-termination and lower runtime bounds with \tool{LoAT}
  (system description).
\newblock In {\em Proc.\ IJCAR~'22}, LNCS 13385, pages 712--722, 2022.
\newblock \href {https://doi.org/10.1007/978-3-031-10769-6_41}
  {\path{doi:10.1007/978-3-031-10769-6_41}}.

\bibitem{JAR17-AProVE}
J.\ Giesl, C.\ Aschermann, M.\ Brockschmidt, F.\ Emmes, F.\ Frohn, C.\ Fuhs,
  C.\ Otto, M.~Pl\"ucker, P.\ Schneider-Kamp, T.\ Str\"oder, S.\ Swiderski, and
  R.\ Thiemann.
\newblock Analyzing program termination and complexity automatically with
  \textsf{AProVE}.
\newblock {\em J.\ Aut.\ Reasoning}, 58(1):3--31, 2017.
\newblock \href {https://doi.org/10.1007/s10817-016-9389-x}
  {\path{doi:10.1007/s10817-016-9389-x}}.

\bibitem{TermComp}
J.\ Giesl, A.\ Rubio, C.\ Sternagel, J.\ Waldmann, and A.\ Yamada.
\newblock The termination and complexity competition.
\newblock In {\em Proc.\ TACAS~'19}, LNCS 11429, pages 156--166, 2019.
\newblock For the results of \emph{TermComp}~'22, see
  \url{https://termination-portal.org/wiki/Termination_Competition_2022}.
\newblock \href {https://doi.org/10.1007/978-3-030-17502-3_10}
  {\path{doi:10.1007/978-3-030-17502-3_10}}.

\bibitem{JLAMP18}
J.\ Hensel, J.\ Giesl, F.\ Frohn, and T.\ Str\"oder.
\newblock Termination and complexity analysis for programs with bitvector
  arithmetic by symbolic execution.
\newblock {\em Journal of Logical and Algebraic Methods in Programming},
  97:105--130, 2018.
\newblock \href {https://doi.org/10.1016/j.jlamp.2018.02.004}
  {\path{doi:10.1016/j.jlamp.2018.02.004}}.

\bibitem{TACAS22}
J.\ Hensel, C.\ Mensendiek, and J.\ Giesl.
\newblock \tool{AProVE}: Non-termination witnesses for \tool{C} programs
  (competition contribution).
\newblock In {\em Proc.\ TACAS~'22}, LNCS 13244, pages 403--407, 2022.
\newblock \href {https://doi.org/10.1007/978-3-030-99527-0_21}
  {\path{doi:10.1007/978-3-030-99527-0_21}}.

\bibitem{HipTNT}
T.~C.\ Le, S.\ Qin, and W.\ Chin.
\newblock Termination and non-termination specification inference.
\newblock In {\em Proc.\ PLDI~'15}, pages 489--498, 2015.
\newblock \href {https://doi.org/10.1145/2737924.2737993}
  {\path{doi:10.1145/2737924.2737993}}.

\bibitem{2LS}
V.\ Mal{\'i}k, {\v{S}}.\ Marti{\v{c}}ek, P.\ Schrammel, M.\ Srivas, T.\ Vojnar,
  and J.\ Wahlang.
\newblock \textsf{2LS}: Memory safety and non-termination.
\newblock In {\em Proc.\ TACAS~'18}, LNCS 10806, pages 417--421, 2018.
\newblock \href {https://doi.org/10.1007/978-3-319-89963-3_24}
  {\path{doi:10.1007/978-3-319-89963-3_24}}.

\bibitem{Met+23}
R.\ Metta, P.\ Yeduru, H.\ Karmarkar, and R.~K.\ Medicherla.
\newblock {{\tool{{VeriFuzz 1.4}}}}: Checking for (non-)termination
  (competition contribution).
\newblock In {\em Proc.\ TACAS~'23}, LNCS 13994, pages 594--599, 2023.
\newblock \href {https://doi.org/10.1007/978-3-031-30820-8_42}
  {\path{doi:10.1007/978-3-031-30820-8_42}}.

\bibitem{RTA10}
C.\ Otto, M.\ Brockschmidt, C.\ von Essen, and J.\ Giesl.
\newblock Automated termination analysis of \tool{Java Bytecode} by term
  rewriting.
\newblock In {\em Proc.\ RTA~'10}, LIPIcs 6, pages 259--276, 2010.
\newblock \href {https://doi.org/10.4230/LIPIcs.RTA.2010.259}
  {\path{doi:10.4230/LIPIcs.RTA.2010.259}}.

\bibitem{PeSCo}
C.\ Richter and H.\ Wehrheim.
\newblock \textsf{PeSCo}: Predicting sequential combinations of verifiers.
\newblock In {\em Proc.\ TACAS~'19}, LNCS 11429, pages 229--233, 2019.
\newblock \href {https://doi.org/10.1007/978-3-030-17502-3_19}
  {\path{doi:10.1007/978-3-030-17502-3_19}}.

\bibitem{Cyclist}
R.~N.~S.\ Rowe and J.\ Brotherston.
\newblock Automatic cyclic termination proofs for recursive procedures in
  separation logic.
\newblock In {\em Proc.\ CPP~'17}, pages 53--65, 2017.
\newblock \href {https://doi.org/10.1145/3018610.3018623}
  {\path{doi:10.1145/3018610.3018623}}.

\bibitem{LLVM-JAR17}
T.\ Str\"oder, J.\ Giesl, M.\ Brockschmidt, F.\ Frohn, C.\ Fuhs, J.\ Hensel,
  P.\ Schneider-Kamp, and C.\ Aschermann.
\newblock Automatically proving termination and memory safety for programs with
  pointer arithmetic.
\newblock {\em J.\ Aut.\ Reasoning}, 58(1):33--65, 2017.
\newblock \href {https://doi.org/10.1007/s10817-016-9389-x}
  {\path{doi:10.1007/s10817-016-9389-x}}.

\bibitem{Vellvm}
J.\ Zhao, S.\ Nagarakatte, M.~M.~K.\ Martin, and S.\ Zdancewic.
\newblock Formalizing the {{\tool{{LLVM}}}} intermediate representation for
  verified program transformations.
\newblock In {\em Proc.\ POPL~'12}, pages 427--440, 2012.
\newblock \href {https://doi.org/10.1145/2103656.2103709}
  {\path{doi:10.1145/2103656.2103709}}.

\end{thebibliography}
\arxiv{\providecommand{\noopsort}[1]{}\providecommand{\apro}{}\providecommand{\ecli}{}

}

\arxiv{

\section*{Appendix}

This appendix contains the formal definition of the semantics of our abstract states using
separation logic in App.\ \ref{Separation Logic Semantics of Abstract
  States}. App.\ \ref{app:rules} presents the additional variants of the 
symbolic execution rule for \code{getelementptr} to handle list traversal (which are all
similar to the list traversal rule in \cref{sect:List Traversal}). In App.\ \ref{app:leadingExample} we present additional
details of the symbolic execution graph of our leading example. Finally, App.\ \ref{app:proofs}
contains the detailed new rule for generalization and the
soundness proofs for all new symbolic execution rules in our paper.

\setcounter{section}{0}
\renewcommand{\thesection}{\Alph{section}}
\section{Separation Logic Semantics of Abstract States}
\label{Separation Logic Semantics of Abstract States}

As described in Sect.\ \ref{sect:domain}, we use the following definition of abstract states.

\begin{definition}[States]\label{LLVM states}
\normalfont{\LLVM{}}  \emph{states} have the form $(\Pos, \LV, \AL, \PT, \LI, \KB)$  where
    $\Pos \in \PPos$, $\LV \colon \Ids \partialfunctionmap \Vsym$,
    $\AL \subseteq \{\alloc{v_1}{v_2} \: | \: v_1, v_2 \in \Vsym\}$,
    $\PT \subseteq \{ (v_1 \pointsto[\codevar{ty}] v_2) \: | \: v_1, v_2 \in \Vsym, \text{\codevar{ty} is an \normalfont{\LLVM{}} type} \}$,
    $\KB \subseteq \QFIA(\Vsym)$,  \pagebreak and
\[\begin{array}{r@{\hspace{3.7cm}}l}
	\multicolumn{2}{l}{\LI \subseteq \{ v_{\mathit{ad}} \pointstorec[\codevar{ty}]{v_\ell}
            [(\mathit{off}_i: \codevar{ty}_i: {v_i..\hat{v}_i})]_{i=1}^n \: | \: n \in
            \N \mbox{ and}}\\
             &\mbox{for all $1 \leq i \leq n$  we have }
        v_{\mathit{ad}},v_\ell,v_i,\hat{v}_i \in \Vsym, 
         \mathit{off}_i \in \N_{>0},\\
         &\text{and } \codevar{ty}, \codevar{ty}_i \text{ are \normalfont{\LLVM{}} types such that}\\
         &\mbox{there is exactly one $1 \leq j \leq n$ with $\codevar{ty}_j = \codevar{ty*}$}
     \}
\end{array}\]
In addition, there is a state $\ERROR$ for undefined behavior.
For any state $s$, let $\Vsym(s)$ consist of  all symbolic variables  occurring in $s$.
\end{definition}

Every state $s$ is represented by a corresponding first-order formula $\stateformula{s}$.

\begin{definition}[Representing States by $\FOL$ Formulas]
  \label{def:StateFOLFormula} Given a state $s = (\Pos, \LV, \AL, \PT, \LI, \KB)$,
the set
$\stateformula{s}$ is the smallest set with
$\stateformula{s} =$

\vspace*{-.5cm}

\[\mbox{ $\begin{array}{@{\hspace*{-.15cm}}l}
\{v_1 \geq 1
\wedge v_1 \leq v_2 \mid \alloc{v_1}{v_2} \in \AL\} \; \cup\\
 \{ v_2 < w_1 \vee w_2 < v_1 \mid
   \alloc{v_1}{v_2},\alloc{w_1}{w_2} \in \AL, \; (v_1,v_2) \neq (w_1,w_2)\} \; 
   \cup\\
 \{v_1 \geq 1 
   \mid (v_1 \hookrightarrow_{\codevar{ty}} v_2) \in \PT \} \; \cup \\
 \{v_2 = w_2 \mid (v_1 \hookrightarrow_{\codevar{ty}} v_2), 
   (w_1 \hookrightarrow_{\codevar{ty}} w_2) \in \PT \mbox{ and } \models \, 
   \stateformula{s} \implies v_1 = w_1\} \; \cup\\
 \{v_1
  \neq w_1 \mid (v_1 \hookrightarrow_{\codevar{ty}} v_2),  
   (w_1 \hookrightarrow_{\codevar{ty}} w_2) \in \PT \mbox{ and } \models \, 
  \stateformula{s} \implies v_2 \neq w_2\} \; \cup\\ 
 \{v_\ell \geq 1 \wedge v_{\mathit{ad}} \geq 1 
  \mid
(v_{\mathit{ad}} \pointstorec[\codevar{ty}]{v_\ell}
  [(\mathit{off}_i: \codevar{ty}_i: {v_i..\hat{v}_i})]_{i=1}^n) \in \LI \} \; \cup \\
 \{\bigwedge_{i=1}^n v_i = \hat{v}_i 
   \mid (v_{\mathit{ad}} \pointstorec[\codevar{ty}]{v_\ell}
  [(\mathit{off}_i: \codevar{ty}_i: {v_i..\hat{v}_i})]_{i=1}^n) \in \LI \mbox{ and } \models \, 
   \stateformula{s} \implies v_\ell = 1\} \; \cup\\
 \mbox{\small $\{v_j \geq 1 
   \mid
(v_{\mathit{ad}} \pointstorec[\codevar{ty}]{v_\ell}
       [(\mathit{off}_i: \codevar{ty}_i: {v_i..\hat{v}_i})]_{i=1}^n) \in \LI
       \mbox{ with }  \codevar{ty}_j = \codevar{ty*} \mbox{ and } \models \, 
   \stateformula{s} \implies v_\ell \geq 2\} \; \cup$}\\
 \mbox{\small $\{v_\ell \geq 2 
   \mid (v_{\mathit{ad}} \pointstorec[\codevar{ty}]{v_\ell}
  [(\mathit{off}_i: \codevar{ty}_i: {v_i..\hat{v}_i})]_{i=1}^n) \in \LI \mbox{ and
   } \exists k \in \Nplus, k \leq n, \, \mbox{s.t.} \, \models \, 
   \stateformula{s} \implies v_k \neq \hat{v}_k\}$}
\end{array}$}
\]
\end{definition}

As mentioned in Sect.\ \ref{sect:domain}, \emph{concrete} states are states where all values of variables and
 memory contents are determined uniquely.  $\PT$ only contains information about
allocated addresses and to enforce a uniform representation,  for concrete states we only allow statements of
the form
$w_1 \hookrightarrow_{\codevar{i8}} w_2$ in $\PT$, i.e.,
contents are represented byte-wise and each address stores a value from
$[0,2^8-1]$. So as mentioned, to ease the formalization we assume
that all integers are unsigned and refer to \cite{JLAMP18} for the general
case.
 Recall that for concrete states, the set $\LI$ is empty since list invariants express abstract
 knowledge about the memory, which is unnecessary
 if the memory contents are known.

 \begin{definition}[Concrete States]
\label{def:concrete_state}
A \normalfont{\LLVM{}} state $c$ is \emph{concrete} iff $c = \ERROR$ or $c  = (\Pos, \LV,
\AL, \PT, \emptyset, \KB)$ such that the following holds:
\begin{itemize}
\item[$\bullet$] $\stateformula{c}$ is satisfiable
\item[$\bullet$] for all $v \in \Vsym(c)$  there exists an
  $n \in \N$ such that $\models \stateformula{c} \implies v = n$
\item[$\bullet$] there is no $(w_1 \hookrightarrow_{\codevar{ty}} w_2) \in \PT$ for $\codevar{ty} \neq \codevar{i8}$
\item[$\bullet$] for all $\alloc{v_1}{v_2} \in \AL$ and for all numbers $n$ with $\models \stateformula{c}  \implies v_1 \leq n \land n \leq v_2$, there exists $(w_1 \hookrightarrow_{\codevar{i8}} w_2) \in \PT$ for some $w_1,w_2 \in \Vsym$ such that $\models\stateformula{c} \implies  w_1 = n$ and $\models\stateformula{c} \implies w_2 = m$ for some $m \in [0,2^8-1]$
\item[$\bullet$] for every  $(w_1 \pointsto[\codevar{i8}] w_2) \in \PT$, there is a $ \alloc{v_1}{v_2} \in \AL$ such that $\models \stateformula{c} \implies v_1 \leq w_1 \leq v_2$
\end{itemize}
\end{definition}

 To define which concrete states are represented by a state
$s$, we introduced a \emph{separation logic} formula $\stateformula{s}_{\SL}$
in \cite{LLVM-JAR17}. 
We used a fragment of separation logic that combines \pagebreak first-order
logic formulas with the predicate symbols ``$\pointsto$'' and ``${\pointsto}_n$'' with $n \in \N$ for information
from $\PT$ and the separating conjunction ``$*$'' to express that memory blocks from $\AL$
are disjoint. We now lift this to a fragment similar to
\emph{quantitative} separation logic \cite{QSL}, extending conventional separation logic
by list predicates.
As in \cite{QSL}, these predicates represent list segments where all
elements are in disjoint parts of the heap. However, we introduce predicates
$\listpred_{\listsize,j}$ that express more details about the structure and the contents
of the list elements than the predicates of \cite{QSL}.
Here, $\listsize$ is the size of a list element in bytes and $j$ is the index of the
``recursive field''. Note that due to data structure padding resp.\ alignment in \sfC{}, the size of a list
element may be greater than the sum of the offset and the type size of
the last field.

As 
in \cite{LLVM-JAR17}, we use \emph{interpretations} $(\slassignment,\slmemory)$ to
define the semantics of separation logic. The \underline{as}signment function
$\slassignment : \Ids \to \N$ maps the program variables to concrete values. The function
$\slmemory : \N_{> 0} \partialfunctionmap \{0,\ldots,2^8-1\}$ describes the \underline{mem}ory
contents at allocated addresses as unsigned bytes. The
following definition uses possibly non-concrete instantiations $\sigma : \Vsym \to
\mathcal{T}(\Vsym)$, where $\mathcal{T}(\Vsym)$ is the set of arithmetic terms containing
only variables from $\Vsym$.

\begin{definition}[Semantics of Separation Logic]
\label{def:semantics-of-sl}
Let $\slassignment : \Ids \to \N$, $\slmemory : \N_{> 0} \partialfunctionmap
\{0,\ldots,2^8-1\}$, and let $\varphi$ be a formula. Let $\slassignment(\varphi)$ result
from replacing all program variables $\codevar{x}$ in $\varphi$ by the value $\slassignment(\codevar{x})$. Note that by construction, program variables $\codevar{x}$ are never quantified in our formulas. Then we define $(\slassignment,\slmemory) \models \varphi$ iff $\slmemory \models \slassignment(\varphi)$.

We now define $\slmemory \models \psi$ for formulas $\psi$ that may contain symbolic variables from $\Vsym$. As usual, all free variables $v_1,\ldots,v_n$ in $\psi$ are implicitly universally quantified, i.e., $\slmemory \models \psi$ iff $\slmemory \models \forall v_1,\ldots, v_n. \, \psi$.

The semantics of arithmetic operations and predicates as well as of first-order
connectives and quantifiers are as usual. In particular, we define $\slmemory \models
\forall v. \, \psi$ iff $\slmemory \models \sigma(\psi)$ holds for all instantiations
$\sigma$ where
  $\sigma(v) \in \N$ and $\sigma(w) = w$ for all $w \in \Vsym \backslash \{v\}$.

The semantics of $\pointstosl$ and $*$ for variable-free formulas are as follows:
For $n_1 \in \N_{>0}$ and $n_2 \in \{0, \ldots, 2^8-1\}$, let $\slmemory \models n_1 \pointstosl n_2$ hold iff $\slmemory(n_1) = n_2$.\footnote{We use ``$\pointstosl$'' instead of ``$\mapsto$'' in separation logic, since $\slmemory \models n_1 \mapsto n_2$ would imply that $\slmemory(n)$ is undefined for all $n \neq n_1$. This would be inconvenient in our formalization, since $\PT$ usually only contains information about a \emph{part} of the allocated memory.}
We assume a little-endian data layout (where least significant bytes are stored in the
lowest address).
Then for $n_1 \in \N_{>0}$ and
$n_2,n\in \N$, $n_1 \hookrightarrow_n n_2$ expresses that
$n_2$ is stored in the $n$ bytes starting at the address $n_1$. Hence, we always have 
$\slmemory \models n_1
\hookrightarrow_0 n_2$ and we define $\slmemory \models n_1 \hookrightarrow_{n+1} n_2$ iff
$\slmemory \models n_1 \hookrightarrow (n_2 \; \mathrm{mod} \; 2^8)$ and $\slmemory
\models (n_1+1) \hookrightarrow_{n} (n_2\;\mathrm{div}\;2^8)$.

The semantics of $*$ is defined as usual in separation logic: For two partial functions $\slmemory_1,\slmemory_2 : \N_{>0} \partialfunctionmap  \{0, \ldots, 2^8-1\}$, we write $\slmemory_1 \bot \slmemory_2$ to indicate that the domains of $\slmemory_1$ and $\slmemory_2$ are disjoint. If $\slmemory_1 \bot \slmemory_2$, then $\slmemory_1 \uplus \slmemory_2$ denotes the union of $\slmemory_1$ and $\slmemory_2$.
Now $\slmemory \models \varphi_1 * \varphi_2$ holds iff there exist $\slmemory_1 \bot \slmemory_2$ such that $\slmemory = \slmemory_1 \uplus \slmemory_2$ where $\slmemory_1 \models \varphi_1$ and $\slmemory_2 \models \varphi_2$.
We define \pagebreak  the empty separating conjunction to be $\mathit{true}$, i.e.,
$\bigast\nolimits_{\varphi \in \varnothing} \, \varphi \; = \; \mathit{true}$. 

Finally,
for $n \in \N_{>0}$ and all $1 \leq i \leq n$, let
$\listsize,  \ell, n_\mathit{ad}, s_i \in \N_{>0}$,
$\mathit{off}_i, m_i, \hat{m}_i \in
\N$,  and $1 \leq j \leq
n$. Then we define the semantics of the list predicate $\listpred_{\listsize,j}$ such that
$\slmemory \models \listpred_{\listsize,j}(\ell, n_\mathit{ad}, [(\mathit{off}_i: s_i:
  {m_i..\hat{m}_i})]_{i=1}^n)$ if
$\slmemory$ is defined on $n_\mathit{ad}$, \ldots,
$n_\mathit{ad}+\listsize-1$ and maps $n_\mathit{ad} + \mathit{off}_i$ to $m_i$ (taking the
type sizes $s_i$ into account). Furthermore, if $\ell = 1$, we must have $m_i = \hat{m}_i$ for
all $1 \leq i \leq n$, since the first values correspond to the last values for lists
consisting of a single element. In contrast, if $\ell > 1$, we require that
$\listpred_{\listsize,j}(\ell-1, m_j, [(\mathit{off}_i:
  s_i:{\overline{m_i}..\hat{m}_i})]_{i=1}^n)$ holds for a different part of the domain
of $\slmemory$ for some $\overline{m_i} \in \N$:
\begin{itemize}
	\item[(1)] $\slmemory \models \listpred_{\listsize,j}(1, n_\mathit{ad},
          [(\mathit{off}_i: s_i: {m_i..\hat{m}_i})]_{i=1}^n)$ iff
            $\slmemory \models \forall x. \exists y. \; (n_\mathit{ad} \leq x \leq
          n_\mathit{ad}+\listsize-1) \implies (x \pointstosl y)$, and for all $1 \leq i
          \leq n$ we have
            $m_i = \hat{m}_i$ and $\slmemory \models n_\mathit{ad}+\mathit{off}_i \pointstosl_{s_i} m_i$
	\item[(2)] $\slmemory \models \listpred_{\listsize,j}(\ell, n_\mathit{ad},
          [(\mathit{off}_i: s_i: {m_i..\hat{m}_i})]_{i=1}^n)$ for $\ell > 1$ iff
          $\slmemory \models (\forall x. \exists y. \; (n_\mathit{ad} \leq x \leq
          n_\mathit{ad}+\listsize-1) \implies (x \pointstosl y)) \,*\,
          \listpred_{\listsize,j}(\ell-1, m_j, [(\mathit{off}_i: s_i:
             {\overline{m_i}..\hat{m}_i})]_{i=1}^n)$ for some $\overline{m_i} \in \N$,
          and
            $\slmemory \models \,
           n_\mathit{ad}+\mathit{off}_i \pointstosl_{s_i} m_i$ for all $1 \leq i \leq n$
          \end{itemize}
\end{definition}

For example, $\slmemory \models \listpred_{16,2}(2, 1408,[(0: 4: 5..0), (8: 8: 1216..0)])$ holds for a function $\slmemory$ iff there exist $\slmemory_1 \bot \slmemory_2$ with $\slmemory = \slmemory_1 \uplus \slmemory_2$ and
\begin{itemize}
	\item for all $1408 \leq x \leq 1423$ there exists $y$ such that $\slmemory_1 \models x \pointstosl y$,
	\item $\slmemory_1 \models 1408 \pointstosl_{4} 5$, $\slmemory_1 \models 1416 \pointstosl_{8} 1216$,
	\item $\slmemory_2 \models \listpred_{16,2}(1, 1216,[(0: 4: \overline{m_1}..0), (8: 8: \overline{m_2}..0)])$ for some $\overline{m_1}, \overline{m_2} \in \N$:
    \begin{itemize}
        \item for all $1216 \leq x \leq 1231$ there exists $z$ such that $\slmemory_2 \models x \pointstosl z$,
        \item $\slmemory_2 \models 1216 \pointstosl_{4} \overline{m_1}$, $\slmemory \models 1224 \pointstosl_{8} \overline{m_2}$, where $\overline{m_1} = 0$ and $\overline{m_2} = 0$.
    \end{itemize}
\end{itemize}

Now we extend the formula $\stateformula{s}_{\SL}$ for a state $s$ from \cite{LLVM-JAR17}
by list predicates for
entries from $\LI$. Note that $\PT$ may contain entries referring to overlapping parts of
the memory and thus, the formulas resulting from $\PT$ are combined by the ordinary
conjunction ``$\wedge$''. In contrast,
all entries from $\AL$ and $\LI$ must operate on disjoint memory parts in order to exclude
sharing
and thus, their formulas are combined with ``$*$''. Moreover, we require disjointness of
the memory domains of $\PT$ and $\LI$
and $\sizeOf(\codevar{ty})$ denotes the size of \codevar{ty} in bytes for any \LLVM{} type \codevar{ty}.

\begin{definition}[$\SL$ Formulas for States]
\label{def:StateFormula}
For $v_1,v_2 \in \Vsym$, let  $\stateformula{\alloc{v_1}{v_2}}_\SL =(\forall x. \exists y. \; (v_1 \leq x \leq v_2) \implies (x \pointstosl y))$.
For any \normalfont{\LLVM{}} type $\codevarx{ty}$ we define
$\stateformula{v_1 \hookrightarrow_{\codevarx{ty}} v_2}_\SL\linebreak = (v_1
\hookrightarrow_{\sizeOf(\codevar{ty})} v_2)$.
For $n \in \N_{>0}$, $1 \leq i,j \leq n$, $\mathit{off}_i \in \N$, $v_\mathit{ad}, v_\ell, v_i,
\hat{v}_i \in \Vsym$, and \normalfont{\LLVM{}} types $\codevar{ty},
\codevar{ty}_i$ with $\codevar{ty}_j = \codevar{ty*}$ and $\codevar{ty}_m \neq
\codevar{ty*}$ for all $m \neq j$, we define
$\stateformula{v_{\mathit{ad}} \pointstorec[\codevar{ty}]{v_\ell} [(\mathit{off}_i: \codevar{ty}_i: {v_i..\hat{v}_i})]_{i=1}^n}_\SL =$
$$\listpred_{\sizeOf(\codevar{ty}),j}(v_\ell, v_\mathit{ad}, [(\mathit{off}_i: \sizeOf(\codevar{ty}_i): {v_i..\hat{v}_i})]_{i=1}^n).$$

\noindent
A state $s = (\Pos,\LV,\AL,\PT,\LI,\KB)$  is then represented in separation logic
by
\[\stateformula{s}_\SL = \stateformula{s}\;   \wedge \;  \LV \; \wedge \;
\left(
\left((\bigast\nolimits_{\varphi \in \AL} \; \; \stateformula{\varphi}_\SL)
\; \wedge \;
(\bigwedge\nolimits_{\varphi \in \PT} \; \; \stateformula{\varphi}_\SL )\right)
\; \ast \;
(\bigast\nolimits_{\varphi \in \LI} \; \; \stateformula{\varphi}_\SL)\right).\]
\end{definition}

While the weaker formula $\stateformula{s}$ is used for all reasoning during the
construction of the symbolic execution graph,  we use $\stateformula{s}_\SL$ to define when
a concrete state $c \neq \ERROR$ is \emph{represented} by a state $s$, see \cref{def:representation}.
To this end, one has to  extract an interpretation $(\slassignment^c,\slmemory^c)$ from a
concrete \pagebreak  state $c$.

\begin{definition}[Interpretations $\slassignment^c$, $\slmemory^c$]
\label{def:Interpretations}
Let $c \neq \ERROR$ be a concrete state $c = (\Pos,\allowbreak \LV^c,\allowbreak \AL^c,\allowbreak \PT^c,\allowbreak \emptyset,\allowbreak \KB^c)$.
For every $\codevar{x} \in \Ids$ where $\codevar{x} \in \domain(\LV^{c})$, let
$\slassignment^c(\codevar{x}) = n$ for the number $n \in \N$ with $\models \stateformula{c} \implies \LV^c(\codevar{x}) = n$.

For $n \in \N_{>0}$, the  function $\slmemory^{c}(n)$ is defined iff there exists a $(w_1 \pointsto[\codevar{i8}] w_2) \in \PT^c$  such that $\models  \stateformula{c}  \implies  w_1 = n$.
Let $\models \stateformula{c} \implies  w_2 = m$ for  $m \in [0,2^8-1]$.
Then we have $\slmemory^c(n) = m$.
\end{definition}

Now we can define the representation of states formally.
A state $s$
\emph{represents} a concrete state $c$ if there exists a concrete instantiation $\sigma : \Vsym \to \N$ which assign numbers to symbolic variables,
such that $(\slassignment^c,\slmemory^c)$ is a model of
$\sigma(\stateformula{s}_\SL)$. Moreover, 
for each allocation in $s$ (explicitly given by $\AL$ or implicitly given by $\LI$)
there must be a corresponding allocation in $c$. The fourth condition of
\cref{def:representation} requires that for each list invariant of length $v_\ell$, if
$\sigma(v_\ell)$ is the concrete value for $v_\ell$, then there exist $\sigma(v_\ell)$
allocations of the size of a list element,
  where the first allocation
begins at the start address $\sigma(v_\mathit{ad})$ of the
invariant and the $k$-th \code{next} pointer points to the $(k+1)$-th allocation for all $1 \leq k \leq
\sigma(v_\ell)-1$. Here, we extend $\sigma$ to formulas as
usual, i.e., $\sigma(\varphi)$ instantiates all free occurrences of $v \in \Vsym$ in
$\varphi$ by $\sigma(v)$.

\begin{definition}[Representation of States]
\label{def:representation}
A concrete state $c = (\Pos,\allowbreak \LV^c,\allowbreak \AL^c,\allowbreak \PT^c,\allowbreak \emptyset,\allowbreak \KB^c)$ is \emph{represented} by a state $s = (\Pos, \LV^s, \AL^s, \PT^s, \LI^s, \KB^s)$ iff
\begin{enumerate}
    \item $\domain(\LV^c) =  \domain(\LV^s)$
    \item $(\slassignment^c,\slmemory^c)$ is a model of $\sigma(\stateformula{s}_\SL)$ for a concrete instantiation $\sigma : \Vsym \to \N$
    \item for all {\small $\alloc{v_1}{v_2} \in \AL^s$} there is {\small $\alloc{w_1}{w_2}
      \in \AL^c$} with {\small $\models \stateformula{c} \Rightarrow w_1 = \sigma(v_1) \wedge w_2 = \sigma(v_2)$}
    \item for each $v_{\mathit{ad}} \pointstorec[\codevar{ty}]{v_\ell} [(\mathit{off}_i:
      \codevar{ty}_i: {v_i..\hat{v}_i})]_{i=1}^n \in \LI^s$ with $\codevar{ty}_j =
      \codevar{ty*}$ there exist allocations
      $\alloc{w^\mathit{start}_{k}}{w^\mathit{end}_{k}} \in \AL^c$ for all $1 \leq k \leq \sigma(v_\ell)$, such that $(\slassignment^c,\slmemory^c)$ is a model of
    $(\bigwedge_{k=1}^{\sigma(v_\ell)} m^\mathit{start}_{k} - m^\mathit{end}_{k} = \sizeOf(\codevar{ty})-1)
      \,\wedge\, \sigma(v_\mathit{ad}) = m^\mathit{start}_{1}
      \,\wedge\,\bigwedge_{k=1}^{\sigma(v_\ell)-1} m^\mathit{start}_{k} + \mathit{off}_j
      \hookrightarrow_{\sizeOf(\codevar{ty*})} m^\mathit{start}_{k+1}$, where
      $m^\mathit{start}_k$ and  $m^\mathit{end}_k$ are the numbers
      with $\models \stateformula{c} \implies w^\mathit{start}_k = m^\mathit{start}_k$
      and $\models \stateformula{c} \implies w^\mathit{end}_k = m^\mathit{end}_k$.
\end{enumerate}
The error state $\ERROR$ is only represented by $\ERROR$ itself.
\end{definition}

So as mentioned in  Sect.\ \ref{sect:domain},
all concrete states $c = ((\cmpFor,0),\LV,\AL,\PT,\emptyset,\linebreak \KB)$ represented by the
example state \eqref{example state}
contain $n$ allocations of 16 bytes for some $n = \sigma(x_\ell) \geq 1$,
where in the first four bytes a 32-bit integer is stored and in the last eight
bytes the address of the next allocation (or 0, in case of the last allocation) is stored. Here, $\slassignment^c(\codevar{curr})$ is the address of the first allocation.

\section{Additional Symbolic Execution Rules for List Traversal}
\label{app:rules}

In \cref{sect:List Traversal} we presented a new symbolic execution rule for list traversal via pointer
arithmetic.  
This appendix  contains the remaining variants of the symbolic execution rule for list
traversal that are needed in addition to the variant in  \rSC{sect:List Traversal}.

%%%%%%%%%%%%%%%%%%%%%%%%%%%%%%%%%%%%%%%%%%%%%%%%%%%%%%%%%%%%%%%%%%

We start with the rule for list traversal via field access, i.e.,
where the next element is accessed using \code{curr'->next} as in the 
\code{for} loop of our leading example. In \LLVM{}, this also corresponds to a
\code{getelementptr}  \pagebreak instruction,
but then one uses the type of the list instead of \code{i8} and indices are used to specify the field to
be accessed. Such a traversal can be handled without considering field offsets. The
corresponding rule
is very similar to the rule in \cref{sect:List Traversal}
(one simply replaces  $\models \, \stateformula{s} \implies \LV(t) = \mathit{off}_j$ in
the first condition by  $\models \, \stateformula{s} \implies \LV(t) = j$).

\vspace*{.2cm}
\noindent
\fbox{
\begin{minipage}{11.8cm}
\label{rule:traversalfieldaccess}
\small
\mbox{\small \textbf{\hspace*{-.15cm}list trav.}}  {\scriptsize($p:$ ``\code{pb = getelementptr ty,
  ty* pa, i$m_1$ $0$, i$m_2$ $t$}'', $t \in \Ids \cup \N$, $\code{pa},\code{pb}\in \Ids$)}\\
\mbox{\small \textbf{\hspace*{-.15cm}via field}}\\
\mbox{\small \textbf{\hspace*{-.15cm}access}}\\
\vspace*{-6mm}\\
\centerline{$\frac{\parbox{6.6cm}{\centerline{
$s = (\Pos, \; \LV, \; \AL, \; \PT, \; \LI, \; \KB)$}\vspace*{.1cm}}}
{\parbox{10.2cm}{\vspace*{.1cm} \centerline{
$s' = (\Pos^+, \; \LV[\code{pb}:=w^\mathit{start}_j], \; \AL \cup \alloc{v^\mathit{start}}{v^\mathit{end}}, \; \PT', \; \LI\backslash\{\listinv\}\cup\{\listinv'\}, \; \KB')$}}}\;\;\;\;$
\mbox{if}} \vspace*{-.05cm}
{\small
\begin{itemize}
\item[$\bullet$] there is $\listinv = (v_\mathit{ad} \pointstorec[\codevar{ty}]{v_\ell}
  [(\mathit{off}_i: \codevar{ty}_i: {v_i..\hat{v}_i})]_{i=1}^n) \in \LI$ with
$\codevar{ty}_j = \codevar{ty*}$,\\
  $\models
  \, \stateformula{s} \implies \LV(\codevar{pa}) = v_j$,  $\models \, \stateformula{s} \implies \LV(t) = j$, and $\models \, \stateformula{s} \implies v_\ell \geq 2$
\item[$\bullet$] $\PT' = \PT \cup \{(v^\mathit{start}_i \hookrightarrow_{\codevar{ty}_i} v_i) \mid 1 \leq i \leq n \}$
\item[$\bullet$] $\listinv' = (w^\mathit{start} \pointstorec[\codevar{ty}]{w_\ell} [(\mathit{off}_i: \codevar{ty}_i: {w_i..\hat{v}_i})]_{i=1}^n)$
\item[$\bullet$] $\KB' = \KB \cup \{v^\mathit{start} = v_\mathit{ad},\, v^\mathit{end} = v^\mathit{start} + \sizeOf(\code{ty}) -1,\, w^\mathit{start} = v_j,\, w_\ell = v_\ell-1,\\ w^\mathit{start}_j = w^\mathit{start} + \mathit{off}_j\} \cup \{v^\mathit{start}_i = v_\mathit{ad} + \mathit{off}_i \mid 1 \leq i \leq n \}$
\item[$\bullet$] $v^\mathit{start}, v^\mathit{end}, v^\mathit{start}_1, \ldots, v^\mathit{start}_n, w^\mathit{start}, w_\ell, w^\mathit{start}_j, w_1, \ldots, w_n \in \Vsym$ are fresh
\end{itemize}}
\end{minipage}}
\vspace*{.2cm}

In the following, we only give rules for the case where pointer arithmetic is used. The
respective rules for list traversal via field access are obtained using the small adaption shown above.

%%%%%%%%%%%%%%%%%%%%%%%%%%%%%%%%%%%%%%%%%%%%%%%%%%%%%%%%%%%%%%%%%%

If the list invariant $\listinv$ to be traversed only contains one element ($\models \,
\stateformula{s} \implies v_\ell = 1$),
but the rest of the first condition of the rule in \cref{sect:List Traversal} holds,
then we apply
a rule which computes the components of $s'$ as in \cref{sect:List Traversal} (i.e., the list element is
transformed to $\AL$ and $\PT$ entries), but we do not add a new invariant
$\listinv'$ to $\LI$. So we
remove the list invariant since the pointer $\code{pb}$ reaches the final
element.
Moreover, we add the atoms $v_i = \hat{v}_i$ for $1 \leq i \leq n$ to $\KB'$.
If we cannot decide whether $v_\ell \geq 2$ holds, then we refine the state such that
$v_\ell \geq 2$ holds
in the first refined
state,  and we have $v_\ell = 1$ in the second state. For our leading
example such a refinement is not needed since the null check ensures that we only enter
the loop body if the length of the list invariant is at least two.

\vspace*{.2cm}
\noindent
\fbox{
\begin{minipage}{11.8cm}
\label{rule:traversal}
\small
\mbox{\small \textbf{\hspace*{-.15cm}list traversal}}  ($p:$ ``\code{pb = getelementptr
  i8, i8* pa, i$m$ $t$}'', {\scriptsize $t \in \Ids \cup \N$, $\code{pa},\code{pb}\in \Ids$})\\
\mbox{\small \textbf{\hspace*{-.15cm}of last element}}\\
\mbox{\small \textbf{\hspace*{-.15cm}via pointer}}\\
\mbox{\small \textbf{\hspace*{-.15cm}arithmetic}}\\
\vspace*{-6mm}\\
\centerline{$\frac{\parbox{6.6cm}{\centerline{
$s = (\Pos, \; \LV, \; \AL, \; \PT, \; \LI, \; \KB)$}\vspace*{.1cm}}}
{\parbox{9.3cm}{\vspace*{.1cm} \centerline{
$s' = (\Pos^+, \; \LV[\code{pb}:=w^\mathit{start}_j], \; \AL \cup \alloc{v^\mathit{start}}{v^\mathit{end}}, \; \PT', \; \LI\backslash\{\listinv\}, \; \KB')$}}}\;\;\;\;$
\mbox{if}} \vspace*{-.05cm}
{\small
\begin{itemize}
\item[$\bullet$] there is $\listinv = (v_\mathit{ad} \pointstorec[\codevar{ty}]{v_\ell}
  [(\mathit{off}_i: \codevar{ty}_i: {v_i..\hat{v}_i})]_{i=1}^n) \in \LI$ with
$\codevar{ty}_j = \codevar{ty*}$,\\
  $\models
  \, \stateformula{s} \implies \LV(\codevar{pa}) = v_j$,  $\models \, \stateformula{s} \implies \LV(t) = \mathit{off}_j$, and $\models \, \stateformula{s} \implies v_\ell = 1$
\item[$\bullet$] $\PT' = \PT \cup \{(v^\mathit{start}_i \hookrightarrow_{\codevar{ty}_i} v_i) \mid 1 \leq i \leq n \}$
\item[$\bullet$] $\KB' = \KB \cup \{v^\mathit{start} = v_\mathit{ad},\, v^\mathit{end} = v^\mathit{start} + \sizeOf(\code{ty}) -1,\, w^\mathit{start}_j = v_j + \mathit{off}_j\} \cup$\\$\{v^\mathit{start}_i = v_\mathit{ad} + \mathit{off}_i \mid 1 \leq i \leq n \} \cup \{v_i = \hat{v}_i \mid 1 \leq i \leq n \}$
\item[$\bullet$] $v^\mathit{start}, v^\mathit{end}, v^\mathit{start}_1, \ldots, v^\mathit{start}_n, w^\mathit{start}_j \in \Vsym$ are fresh
\end{itemize}}
\end{minipage}}
\vspace*{.2cm}

%%%%%%%%%%%%%%%%%%%%%%%%%%%%%%%%%%%%%%%%%%%%%%%%%%%%%%%%%%%%%%%%%%

If the list invariant $\listinv_2$ to be traversed starts at an address $v_\mathit{ad}$
that is equal to the value $\hat{u}_j$ of the last ``recursive'' pointer of another list
invariant $\listinv_1$, we have a split list. As above, the invariant $\listinv_2$ is
replaced by a new list invariant $\listinv'_2$ with decremented length that starts at the
second element of $\listinv_2$. But instead of transforming the first element to entries
from $\AL$ and $\PT$, it is ``appended'' to the end of $\listinv_1$, resulting in a new
list invariant $\listinv'_1$ with incremented length. 
In the new resulting state, again the last recursive pointer $v_j$ of $\listinv'_1$ is equal to
the start address $w^\mathit{start}$ of $\listinv'_2$.
The soundness proof is similar to
the proof for 
the corresponding rule above but uses disjointness
of the individual list invariants instead of
disjointness between the allocations and the list invariants.

\vspace*{.2cm}
\noindent
\fbox{
\begin{minipage}{11.8cm}
\label{rule:traversal}
\small
\mbox{\small \textbf{\hspace*{-.15cm}list traversal}} ($p:$ ``\code{pb = getelementptr i8,
  i8* pa, i$m$ $t$}'', {\scriptsize $t \in \Ids \cup \N$, $\code{pa},\code{pb}\in \Ids$})\\
\mbox{\small \textbf{\hspace*{-.15cm}via pointer}}\\
\mbox{\small \textbf{\hspace*{-.15cm}arithmetic}}\\
\mbox{\small \textbf{\hspace*{-.15cm}(split lists)}}\\
\vspace*{-6mm}\\
\centerline{$\frac{\parbox{6.6cm}{\centerline{
$s = (\Pos, \; \LV, \; \AL, \; \PT, \; \LI, \; \KB)$}\vspace*{.1cm}}}
{\parbox{8.9cm}{\vspace*{.1cm} \centerline{
$s' = (\Pos^+, \; \LV[\code{pb}:=w^\mathit{start}_j], \; \AL, \; \PT, \; \LI\backslash\{\listinv_1,\listinv_2\}\cup\{\listinv'_1,\listinv'_2\}, \; \KB')$}}}\;\;\;\;$
\mbox{if}} \vspace*{-.05cm}
{\small
\begin{itemize}
\item[$\bullet$] there is $\listinv_2 = (v_\mathit{ad} \pointstorec[\codevar{ty}]{v_\ell}
  [(\mathit{off}_i: \codevar{ty}_i: {v_i..\hat{v}_i})]_{i=1}^n) \in \LI$ with
$\codevar{ty}_j = \codevar{ty*}$,\\
  $\models
  \, \stateformula{s} \implies \LV(\codevar{pa}) = v_j$,  $\models \, \stateformula{s} \implies \LV(t) = \mathit{off}_j$, and $\models \, \stateformula{s} \implies v_\ell \geq 2$
\item[$\bullet$] there is $\listinv_1 = (u_\mathit{ad} \pointstorec[\codevar{ty}]{u_\ell}
  [(\mathit{off}_i: \codevar{ty}_i: {u_i..\hat{u}_i})]_{i=1}^n) \in \LI$ with
$\models \, \stateformula{s} \implies \hat{u}_j = v_\mathit{ad}$
\item[$\bullet$] $\listinv'_1 = (u_\mathit{ad} \pointstorec[\codevar{ty}]{{u}'_\ell}
  [(\mathit{off}_i: \codevar{ty}_i: {u_i..v_i})]_{i=1}^n)$
\item[$\bullet$] $\listinv'_2 = (w^\mathit{start} \pointstorec[\codevar{ty}]{w_\ell} [(\mathit{off}_i: \codevar{ty}_i: {w_i..\hat{v}_i})]_{i=1}^n)$
\item[$\bullet$] $\KB' = \KB \cup \{{u}'_\ell = u_\ell+1,\, w^\mathit{start} = v_j,\, w_\ell = v_\ell-1,\, w^\mathit{start}_j = w^\mathit{start} + \mathit{off}_j\}$
\item[$\bullet$] ${u}'_\ell, w^\mathit{start}, w_\ell, w^\mathit{start}_j, w_1, \ldots, w_n \in \Vsym$ are fresh
\end{itemize}}
\end{minipage}}
\vspace*{.2cm}

%%%%%%%%%%%%%%%%%%%%%%%%%%%%%%%%%%%%%%%%%%%%%%%%%%%%%%%%%%%%%%%%%%

If the list invariant $\listinv_2$ to be traversed only contains one element
($\models \, \stateformula{s} \implies v_\ell = 1$) and there is a list invariant
$\listinv_1$ whose last element points to the only element of $\listinv_2$, then $\listinv_2$
is ``appended'' to the end of $\listinv_1$.

\vspace*{.2cm}
\noindent
\fbox{
\begin{minipage}{11.8cm}
\label{rule:traversal}
\small
\mbox{\small \textbf{\hspace*{-.15cm}list traversal}}  ($p:$ ``\code{pb = getelementptr i8, i8* pa, i$m$ $t$}'', {\scriptsize $t \in \Ids \cup \N$, $\code{pa},\code{pb}\in \Ids$})\\
\mbox{\small \textbf{\hspace*{-.15cm}of last element}}\\
\mbox{\small \textbf{\hspace*{-.15cm}via pointer}}\\
\mbox{\small \textbf{\hspace*{-.15cm}arithmetic}}\\
\mbox{\small \textbf{\hspace*{-.15cm}(split lists)}}\\
\vspace*{-6mm}\\
\centerline{$\frac{\parbox{6.6cm}{\centerline{
$s = (\Pos, \; \LV, \; \AL, \; \PT, \; \LI, \; \KB)$}\vspace*{.1cm}}}
{\parbox{8.6cm}{\vspace*{.1cm} \centerline{
$s' = (\Pos^+, \; \LV[\code{pb}:=w^\mathit{start}_j], \; \AL, \; \PT, \; \LI\backslash\{\listinv_1,\listinv_2\}\cup\{\listinv'_1\}, \; \KB')$}}}\;\;\;\;$
\mbox{if}} \vspace*{-.05cm}
{\small
\begin{itemize}
\item[$\bullet$] there is $\listinv_2 = (v_\mathit{ad} \pointstorec[\codevar{ty}]{v_\ell}
  [(\mathit{off}_i: \codevar{ty}_i: {v_i..\hat{v}_i})]_{i=1}^n) \in \LI$ with
$\codevar{ty}_j = \codevar{ty*}$,\\
  $\models
  \, \stateformula{s} \implies \LV(\codevar{pa}) = v_j$,  $\models \, \stateformula{s} \implies \LV(t) = \mathit{off}_j$, and $\models \, \stateformula{s} \implies v_\ell = 1$
\item[$\bullet$] there is $\listinv_1 = (u_\mathit{ad} \pointstorec[\codevar{ty}]{u_\ell}
  [(\mathit{off}_i: \codevar{ty}_i: {u_i..\hat{u}_i})]_{i=1}^n) \in \LI$ with
$\models \, \stateformula{s} \implies \hat{u}_j = v_\mathit{ad}$
\item[$\bullet$] $\listinv'_1 = (u_\mathit{ad} \pointstorec[\codevar{ty}]{{u}'_\ell}
  [(\mathit{off}_i: \codevar{ty}_i: {u_i..v_i})]_{i=1}^n)$
\item[$\bullet$] $\KB' = \KB \cup \{{u}'_\ell = u_\ell+1,\, w^\mathit{start}_j = v_j + \mathit{off}_j\} \cup \{v_i = \hat{v}_i \mid 1 \leq i \leq n \}$
\item[$\bullet$] ${u}'_\ell, w^\mathit{start}_j \in \Vsym$ are fresh
\end{itemize}}
\end{minipage}}
\vspace*{.2cm}

\clearpage
\setcounter{section}{2}
\renewcommand{\thesection}{\Alph{section}}
\section{Symbolic Execution Graph of our Leading Example}
\label{app:leadingExample}

In this appendix, we present the remaining construction of the SEG for our leading example.

The first cycle of the SEG contains the two states 
$\GraphInitCmpZeroGen$  and $\GraphInitCmpZeroC$ at the program position $(\cmpFor, 0)$,
with a generalization edge from  $\GraphInitCmpZeroC$ to $\GraphInitCmpZeroGen$, see
\cref{fig:ForLoopGen}.

\footnotesize

\begin{figure}[t]
\centering
\begin{tikzpicture}[node distance = \ydist and \xdist]
\scriptsize
\def\widetwidth{6.2cm}
\def\smalltwidth{4.75cm}
\def\fulllinewidhth{11cm}
\def\edgenodedist{0.2cm}
\def\FirstIndentwidth{3cm}
\def\SecondIndentwidth{2cm}
\tikzset{invisible/.style={opacity=0}}
\tikzstyle{state}=[
						   %minimum size=10pt,
                           %fill=white,
                           %shape=rectangle,
                           %text=black,
                           inner sep=2pt,
                           font=\scriptsize,
                           draw]

\node[state, align=left, label={[xshift=-\labelxshift]180:$\GraphInitCmpZeroGen$}] (3) 
{$(\cmpFor, 0),\;
  \{ \code{n} = x_{\code{n}}, \,
   \code{tail\_ptr} = x_{\code{tp}}, \,
   \code{mem} = x_{\code{mem}}, \,
   \code{curr} = x_{\code{mem}}, \,
   \code{nondet} = x_{\code{nd}}, \,
   \code{curr\_val} = x_{\code{mem}}, \,
   \code{curr\_next} = x_{\code{cn}},$\\
$\,\code{k} = x_{\code{k}}, \,
   \code{kinc} = x_{\code{kinc}}, \, 
   ...\},\;
  \{
   \alloc{x_{\code{tp}}}{x_\code{tp}^\mathit{end}}, \,
   \alloc{x_{\code{k\_ad}}}{x_\code{k\_ad}^\mathit{end}}
   \},$\\
 $\{
   x_{\code{tp}} \pointsto[\code{list*}] x_{\code{mem}}, \,
   x_{\code{k\_ad}} \pointsto[\code{i32}] x_{\code{kinc}}
  \}, \;
  \{
   x_{\code{mem}} \pointstorec[\code{list}]{x_{\ell}}
     [(0: \code{i32}: x_{\code{nd}}..\hat{x}_{\code{nd}}),
      (8: \code{list*}: x_{\code{next}}..0)]
   \},$\\
 $\{  x_{\code{n}} > x_{\code{k}}, \,
x_{\code{k\_ad}}^\mathit{end} = x_{\code{k\_ad}} + 3, \,
   x_\code{tp}^\mathit{end} = x_{\code{tp}} + 7, \,
%   v_{\code{k\_ad}_\mathit{end}} = v_{\code{k\_ad}} + 3, \,
%   x_{\code{mem}_\mathit{end}} = x_{\code{mem}} + 15, \,
   x_{\code{cn}} = x_{\code{mem}} + 8, \,
%   v_{\code{n}} > 0, \,
   x_{\code{kinc}} = x_{\code{k}} + 1, \,
%   x_{\ell} = x_{\code{kinc}}, \,
   ... \}$
};

\node[state, below left=0.15cm and -4.5cm of 3, align=left] (4) {};

\node[below left=-0.13cm and -0.05cm of 4] (label) {$\GraphTravEntryZero$};

\node[left=0.6cm of 4] (dots) {\ldots};

\node[state, below right=0.25cm and -4.5cm of 4, align=left, label={[xshift=-\labelxshift]183:$\GraphInitCmpZeroC$}] (5) 
{$(\cmpFor, 0),\;
  \{ \code{n} = x_{\code{n}}, \,
   \code{tail\_ptr} = x_{\code{tp}}, \,
   \code{mem} = y_{\code{mem}}, \,
   \code{curr} = y_{\code{mem}}, \,
   \code{nondet} = y_{\code{nd}}, \,
   \code{curr\_val} = y_{\code{mem}}, \,
   \code{curr\_next} = y_{\code{cn}},$\\
$\,\code{k} = x_{\code{kinc}}, \,
   \code{kinc} = y_{\code{kinc}}, \,
     ...\},\;
  \{
   \alloc{x_{\code{tp}}}{x_\code{tp}^\mathit{end}}, \,
   \alloc{x_{\code{k\_ad}}}{x_\code{k\_ad}^\mathit{end}}
   \},$\\
 $\{
   x_{\code{tp}} \pointsto[\code{list*}] y_{\code{mem}}, \,
   x_{\code{k\_ad}} \pointsto[\code{i32}] y_{\code{kinc}}
  \}, \;
  \{
   y_{\code{mem}} \pointstorec[\code{list}]{y_{\ell}}
     [(0: \code{i32}: y_{\code{nd}}..\hat{x}_{\code{nd}}),
      (8: \code{list*}: x_{\code{mem}}..0)]
   \},$\\
 $\{
    x_{\code{n}} > x_{\code{kinc}}, \,
  x_{\code{k\_ad}}^\mathit{end} = x_{\code{k\_ad}} + 3, \,
   x_\code{tp}^\mathit{end} = x_{\code{tp}} + 7, \,
%   v_{\code{k\_ad}_\mathit{end}} = v_{\code{k\_ad}} + 3, \,
%   v_{\code{n}} > 0, \,
%   y_{\code{mem}_\mathit{end}} = y_{\code{mem}} + 15, \,
   y_{\code{cn}} = y_{\code{mem}} + 8, \,
   y_{\code{kinc}} = x_{\code{kinc}} + 1, \,
%   x_{\ell} = x_{\code{kinc}}, \,
   y_{\ell} = x_{\ell} + 1, \,
   ... \}$
};

\draw[omit-edge] ($(4)+(1.65,0)$) -- (4);
\draw[omit-edge] (4) -- (dots);
\draw[omit-edge] (3) -- (5);
\draw[gen-edge] (5.west) to[bend left=80] (3.south west);

\end{tikzpicture}
\caption{Generalization of States in the First Cycle}
\label{fig:ForLoopGen}
\end{figure}
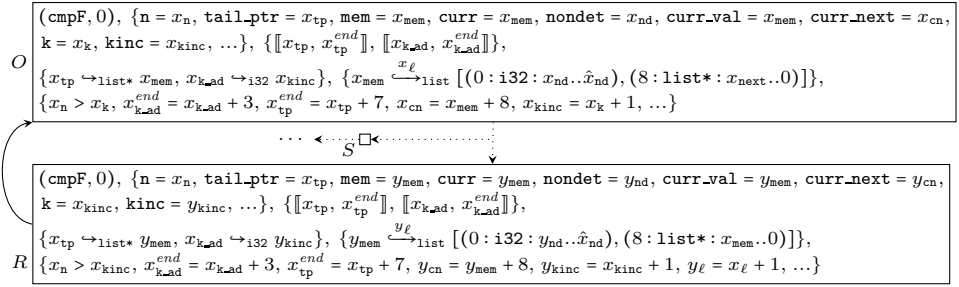
\normalsize

For the second cycle of the SEG, we continue the construction from \cref{sect:List Traversal}.
After reaching  State $\GraphTravBodyGetelemRes$ by the list traversal rule in \cref{sect:List Traversal}, we
continue the symbolic execution.
When we reach
position $(\cmpWhile, 0)$ again, our merging heuristic from
\cite{LLVM-JAR17}
would merge the resulting
state with the previous state $\GraphTravCmpZeroTwo$
at this
position.
Then the list invariants would be merged to a
generalized invariant but in that way we would lose the information about the first
part of the list,
which has already been traversed and transformed to other components. Therefore,
then we would not be able to prove memory safety anymore if there
were a subsequent loop where one iterates through the whole list again.
  For this reason, we forbid
merging of two states $s_1, s_2$ if $s_1$ contains a list invariant $\listinv$ and there are no list
elements resulting from $\AL^{s_1}$ and $\PT^{s_1}$ pointing to the first element of
$\listinv$, and the other state $s_2$ contains a corresponding list invariant $\listinv'$
and a list element resulting from $\AL^{s_2}$ and $\PT^{s_2}$ pointing to the first element
of $\listinv'$. In our example, the successor of $\GraphTravBodyGetelemRes$ at position $(\code{cmpW},0)$ contains the list
element at $\alloc{x_{\code{mem}}}{x_\code{mem}^\mathit{end}}$, which points to the list
invariant starting at $x_{\code{next}}$
(since $(x_{\code{np}} = x_{\code{mem}} + 8)$,
$(x_\code{mem}^\mathit{end} = x_{\code{mem}} + 15) \in \KB^{\GraphTravBodyGetelemRes}$, and
 $(x_{\code{np}} \pointsto[\code{list*}] x_{\code{next}}) \in
\PT^{\GraphTravBodyGetelemRes}$).
Its predecessor at the same position (i.e., $\GraphTravCmpZeroTwo$)
  does not contain any list element
pointing to the start of its corresponding list invariant.
Using this merging heuristic, we iterate the loop a third time and apply the list traversal
rule a second time (Fig. \ref{fig:Trav2}).

\footnotesize

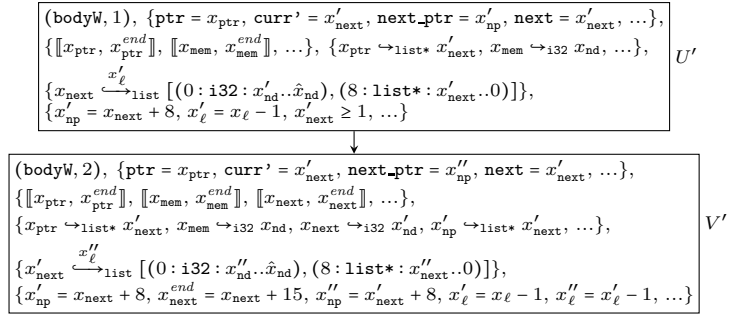
\begin{figure}[t]
\centering
\begin{tikzpicture}[node distance = \ydist and \xdist]
\scriptsize
\def\widetwidth{6.2cm}
\def\smalltwidth{4.75cm}
\def\fulllinewidhth{11cm}
\def\edgenodedist{0.2cm}
\def\FirstIndentwidth{3cm}
\def\SecondIndentwidth{2cm}
\tikzset{invisible/.style={opacity=0}}
\tikzstyle{state}=[
						   %minimum size=10pt,
                           %fill=white,
                           %shape=rectangle,
                           %text=black,
                           inner sep=2pt,
                           font=\scriptsize,
                           draw]

\node[state, align=left, label={[xshift=\labelxshift]2:$\GraphTravBodyGetelemStartP$}] (5) 
{$(\bodyWhile, 1),\;
  \{
   \code{ptr} = x_{\code{ptr}}, \,
   \code{curr'} = x'_{\code{next}}, \,
   \code{next\_ptr} = x'_{\code{np}}, \,
   \code{next} = x'_{\code{next}}, \,
   ...\},$\\
 $\{
   \alloc{x_{\code{ptr}}}{x_\code{ptr}^\mathit{end}}, \,
   \alloc{x_{\code{mem}}}{x_\code{mem}^\mathit{end}}, \,
   ...\}, \;
  \{
   x_{\code{ptr}} \pointsto[\code{list*}] x'_{\code{next}}, \,
   x_{\code{mem}} \pointsto[\code{i32}] x_{\code{nd}}, \,
   ...\},$\\
 $\{
   x_{\code{next}} \pointstorec[\code{list}]{x'_{\ell}}
     [(0: \code{i32}: x'_{\code{nd}}..\hat{x}_{\code{nd}}),
      (8: \code{list*}: x'_{\code{next}}..0)]
   \},$\\
 $\{
%   x_{\code{kinc}} = x_{\code{k}} + 1, \,
%   x_{\ell} = x_{\code{kinc}}, \,
%   z_{\code{next}_\mathit{end}} = z_{\code{next}} + 15, \,
   x'_{\code{np}} = x_{\code{next}} + 8, \,
   x'_{\ell} = x_{\ell} - 1, \,
   x'_\code{next} \geq 1, \,
   ... \}$
};

\node[state, below=of 5, align=left, label={[xshift=\labelxshift]2:$\GraphTravBodyGetelemResP$}] (6) 
{$(\bodyWhile, 2),\;
  \{
   \code{ptr} = x_{\code{ptr}}, \,
   \code{curr'} = x'_{\code{next}}, \,
   \code{next\_ptr} = {x}_{\code{np}}'', \,
   \code{next} = x'_{\code{next}}, \,
   ...\},$\\
 $\{
   \alloc{x_{\code{ptr}}}{x_\code{ptr}^\mathit{end}}, \,
   \alloc{x_{\code{mem}}}{x_\code{mem}^\mathit{end}}, \,
   \alloc{x_{\code{next}}}{x_\code{next}^\mathit{end}}, \,
   ...\},$\\
  $\{
   x_{\code{ptr}} \pointsto[\code{list*}] x'_{\code{next}}, \,
   x_{\code{mem}} \pointsto[\code{i32}] x_{\code{nd}}, \,
   x_{\code{next}} \pointsto[\code{i32}] x'_{\code{nd}}, \,
   x'_{\code{np}} \pointsto[\code{list*}] x'_{\code{next}}, \,
   ...\},$\\
 $\{
   x'_{\code{next}} \pointstorec[\code{list}]{{x}_{\ell}''}
     [(0: \code{i32}: {x}_{\code{nd}}''..\hat{x}_{\code{nd}}),
      (8: \code{list*}: {x}_{\code{next}}''..0)]
   \},$\\
 $\{
%   x_{\code{kinc}} = x_{\code{k}} + 1, \,
%   x_{\ell} = x_{\code{kinc}}, \,
%   z_{\code{next}_\mathit{end}} = z_{\code{next}} + 15, \,
   x'_{\code{np}} = x_{\code{next}} + 8, \,
x_{\code{next}}^{\mathit{end}} = x_{\code{next}} + 15, \,
   {x}_{\code{np}}'' = x'_{\code{next}} + 8, \,
   x'_{\ell} = x_{\ell} - 1, \,
   {x}_{\ell}'' = x'_{\ell} - 1, \,
   ... \}$
};

\draw[eval-edge] (5)  --  (6);

\end{tikzpicture}
\caption{Traversing a List Invariant (Part 2)}
\label{fig:Trav2}
\end{figure}
\normalsize

The
next state at position $(\cmpWhile, 0)$ has
a concrete list
with two elements (at the addresses $x_\code{mem}$ and $x_\code{next}$) before the element
at the start address $x'_\code{next}$ of the list invariant.
This state is merged with the state at position
$(\cmpWhile, 0)$ from the previous iteration, which has a concrete list
with one element (at the address  $x_\code{mem}$) before the element 
at the start address $x_\code{next}$ of its list invariant.
The merging results in the state $\GraphTravCmpZeroThree$
in Fig. \ref{fig:Trav3}. Here,
according to our merging technique in \cref{sect:inferring_inv},
the concrete lists consisting of the traversed list elements are
merged to a new list invariant. We call such a list that is represented by more than one
list invariant a \emph{split list}. To evaluate a \code{getelementptr} instruction on a split
list
  (i.e., the base address $\code{pa}$
of this instruction is known to be the \code{next} pointer of a list
invariant
 $\listinv_2$, and there is another list invariant $\listinv_1$ whose last element points to the
start address of $\listinv_2$), we use a variant of the above list traversal rule where
the first element of $\listinv_2$ is added as a last element to $\listinv_1$ and the 
length of $\listinv_1$ is incremented, see App.\ \ref{app:rules}. Using this rule, State
$\GraphTravBodyGetelemStartFinal$ is evaluated to State
$\GraphTravBodyGetelemResFinal$.
Note that the adaption of the merging heuristic for list traversal adds at most one
further loop iteration to the graph, so it does not change the termination behavior of the
graph construction. 
We reach position $(\cmpWhile, 0)$ again at state $\GraphTravCmpZeroFour$. This state is
an instance of $\GraphTravCmpZeroThree$ with
$\mu_{\GraphTravCmpZeroFour}(z_\code{mem}) = z_\code{mem}$,
$\mu_{\GraphTravCmpZeroFour}(z'_\code{next}) = z''_\code{next}$,
$\mu_{\GraphTravCmpZeroFour}(z''_\code{next}) = z'''_\code{next}$,
$\mu_{\GraphTravCmpZeroFour}(z_{\ell 1}) = z'_{\ell 1}$, and
$\mu_{\GraphTravCmpZeroFour}(z_{\ell}) = z'_{\ell}$. The generalization edge from
$\GraphTravCmpZeroFour$ to $\GraphTravCmpZeroThree$ finishes the graph construction.

Note that currently our approach does not support modifying or freeing list elements while
they are part of a list invariant. In such cases, using the rules from \cite{LLVM-JAR17}
we reach the state $\ERROR$. However, this can easily be improved by adapting the
corresponding rules for \code{store} and \code{free}.

\footnotesize

\begin{figure}
\centering
\begin{tikzpicture}[node distance = \ydist and \xdist]
\scriptsize
\def\widetwidth{6.2cm}
\def\smalltwidth{4.75cm}
\def\fulllinewidhth{11cm}
\def\edgenodedist{0.3cm}
\def\FirstIndentwidth{3cm}
\def\SecondIndentwidth{2cm}
\tikzset{invisible/.style={opacity=0}}
\tikzstyle{state}=[
						   %minimum size=10pt,
                           %fill=white,
                           %shape=rectangle,
                           %text=black,
                           inner sep=2pt,
                           font=\scriptsize,
                           draw]

\node[state,label={[xshift=\labelxshift]92:$\GraphInitCmpZeroGen$}] (B1) {};

\node[state, right=0.7cm of B1,label={[xshift=\labelxshift]92:$\GraphTravEntryZero$}] (B2) {};

\node[state, right=0.7cm of B2,label={[xshift=\labelxshift]92:$\GraphTravCmpZeroTwo$}] (B3) {};

\node[state, right=0.7cm of B3,label={[xshift=\labelxshift]92:$\GraphTravBodyGetelemStart$}] (B4) {};

\node[state, right=0.7cm of B4,label={[xshift=\labelxshift]92:$\GraphTravBodyGetelemRes$}] (B5) {};

\node[state, below=of B5, align=left, label={[xshift=\labelxshift]2:$\GraphTravCmpZeroThree$}] (3) 
{$(\cmpWhile, 0),\;
  \{
   \code{ptr} = z_{\code{ptr}}, \,
   \code{curr'} = z'_{\code{next}}, \,
   \code{next\_ptr} = z''_{\code{np}}, \,
   \code{next} = z''_{\code{next}}, \,
   ...\},$\\
 $\{
   \alloc{z_{\code{ptr}}}{z_\code{ptr}^\mathit{end}},
   ...\},
   \{
   z_{\code{ptr}} \pointsto[\code{list*}] z''_{\code{next}}, \,
   ...\},$\\
 $\{
   z_{\code{mem}} \pointstorec[\code{list}]{z_{\ell 1}}
     [(0: \code{i32}: z_{\code{nd}}..{z'_{\code{nd}}}),
      (8: \code{list*}: z_{\code{next}}..{z'_\code{next}})], \,$\\
  \;\,$z'_{\code{next}} \pointstorec[\code{list}]{z_{\ell}}
     [(0: \code{i32}: z''_\code{nd}..\hat{z}_\code{nd}),
      (8: \code{list*}: z''_{\code{next}}..0)]
   \},$\\
 $\{
   z''_{\code{np}} = z'_{\code{next}} + 8, \,
 %  z_{\ell 1} \geq 1, \,
   z_{\ell} = z'_{\ell} - 1, \,
 %  z_{\ell} \geq 1, \,
   ... \}$
};

\node[state, below=of 3, align=left, label={[xshift=\labelxshift]2:$\GraphTravBodyGetelemStartFinal$}] (4) 
{$(\bodyWhile, 1),\;
  \{
   \code{ptr} = z_{\code{ptr}}, \,
   \code{curr'} = z''_{\code{next}}, \,
   \code{next\_ptr} = z''_{\code{np}}, \,
   \code{next} = z''_{\code{next}}, \,
   ...\},$\\
 $\{
   \alloc{z_{\code{ptr}}}{z_\code{ptr}^\mathit{end}},
   ...\},
   \{
   z_{\code{ptr}} \pointsto[\code{list*}] z''_{\code{next}}, \,
   ...\},$\\
 $\{
   z_{\code{mem}} \pointstorec[\code{list}]{z_{\ell 1}}
     [(0: \code{i32}: z_{\code{nd}}..{z'_{\code{nd}}}),
      (8: \code{list*}: z_{\code{next}}..{z'_\code{next}})], \,$\\
  \;\,$z'_{\code{next}} \pointstorec[\code{list}]{z_{\ell}}
     [(0: \code{i32}: z''_\code{nd}..\hat{z}_\code{nd}),
      (8: \code{list*}: z''_{\code{next}}..0)]
   \},$\\
 $\{
   z''_{\code{np}} = z'_{\code{next}} + 8, \,
 %  z_{\ell 1} \geq 1, \,
   z_{\ell} = z'_{\ell} - 1, \,
 %  z_{\ell} \geq 1, \,
   ... \}$
};

\node[state, below=of 4, align=left, label={[xshift=\labelxshift]2:$\GraphTravBodyGetelemResFinal$}] (5) 
{$(\bodyWhile, 2),\;
  \{
   \code{ptr} = z_{\code{ptr}}, \,
   \code{curr'} = z''_{\code{next}}, \,
   \code{next\_ptr} = z'''_{\code{np}}, \,
   \code{next} = z''_{\code{next}}, \,
   ...\},$\\
 $\{
   \alloc{z_{\code{ptr}}}{z_\code{ptr}^\mathit{end}},
   ...\},
   \{
   z_{\code{ptr}} \pointsto[\code{list*}] z''_{\code{next}}, \,
   ...\},$\\
 $\{
   z_{\code{mem}} \pointstorec[\code{list}]{z'_{\ell 1}}
     [(0: \code{i32}: z_{\code{nd}}..{z''_\code{nd}}),
      (8: \code{list*}: z_{\code{next}}..{z''_\code{next}})], \,$\\
  \;\,$z''_{\code{next}} \pointstorec[\code{list}]{z'_{\ell}}
     [(0: \code{i32}: z'''_\code{nd}..\hat{z}_\code{nd}),
      (8: \code{list*}: z'''_{\code{next}}..0)]
   \},$\\
 $\{
 %  z_{\ell 1} \geq 1, \,
 %  z_{\ell} \geq 1, \,
   z'''_{\code{np}} = z''_{\code{next}} + 8, \,
   z'_{\ell 1} = z_{\ell 1} + 1, \,
   z'_{\ell} = z_{\ell} - 1, \,
   ... \}$
};

\node[state, below=of 5, align=left, label={[xshift=\labelxshift]2:$\GraphTravCmpZeroFour$}] (6) 
{$(\cmpWhile, 0),\;
  \{
   \code{ptr} = z_{\code{ptr}}, \,
   \code{curr'} = z''_{\code{next}}, \,
   \code{next\_ptr} = z'''_{\code{np}}, \,
   \code{next} = z'''_{\code{next}}, \,
   ...\},$\\
 $\{
   \alloc{z_{\code{ptr}}}{z_\code{ptr}^\mathit{end}},
   ...\},
   \{
   z_{\code{ptr}} \pointsto[\code{list*}] z'''_{\code{next}}, \,
   ...\},$\\
 $\{
   z_{\code{mem}} \pointstorec[\code{list}]{z'_{\ell 1}}
     [(0: \code{i32}: z_{\code{nd}}..{z''_\code{nd}}),
      (8: \code{list*}: z_{\code{next}}..{z''_\code{next}})], \,$\\
  \;\,$z''_{\code{next}} \pointstorec[\code{list}]{z'_{\ell}}
     [(0: \code{i32}: z'''_\code{nd}..\hat{z}_\code{nd}),
      (8: \code{list*}: z'''_{\code{next}}..0)]
   \},$\\
 $\{
 %  z_{\ell 1} \geq 1, \,
 %  z_{\ell} \geq 1, \,
   z'''_{\code{np}} = z''_{\code{next}} + 8, \,
   z'_{\ell 1} = z_{\ell 1} + 1, \,
   z'_{\ell} = z_{\ell} - 1, \,
   ... \}$
};

\draw[omit-edge] (B1)  --  (B2);
\draw[omit-edge] (B2)  --  (B3);
\draw[omit-edge] (B3)  --  (B4);
\draw[eval-edge] (B4)  --  (B5);
\draw[omit-edge] (B5)  --  (3);
\draw[omit-edge] (3)  --  (4);
\draw[eval-edge] (4)  --  (5);
\draw[omit-edge] (5)  --  (6);
\coordinate (gen-edge) at ($(6.west)+(-\edgenodedist,0)$);
\draw[gen-edge,rounded corners=5pt] (6.west) -- (gen-edge) -- (gen-edge |- 3.west) -- (3.west);

\end{tikzpicture}
\caption{Traversing a List Invariant (Part 3)}
\label{fig:Trav3}
\end{figure}
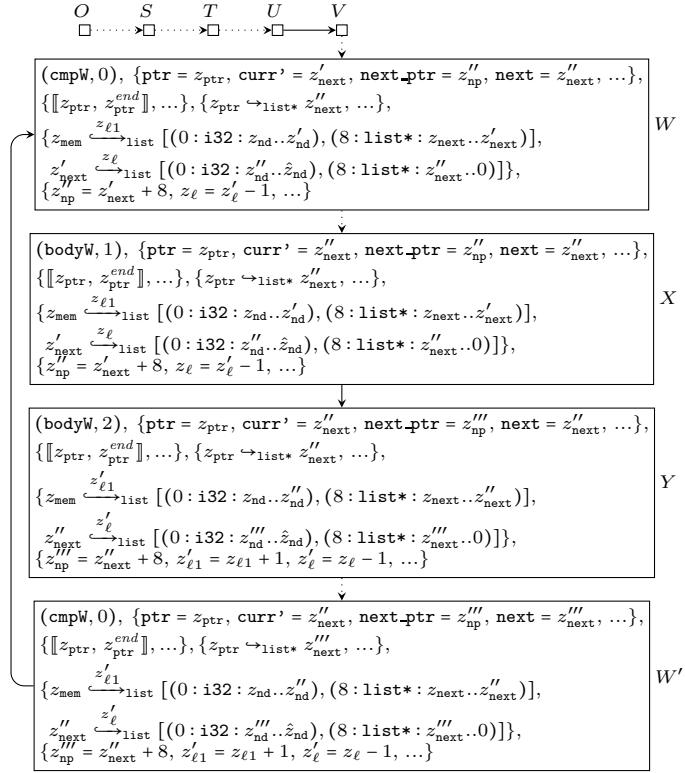
\normalsize

\clearpage

\setcounter{section}{3}
\section{Proofs}
\label{app:proofs}

This appendix contains all proofs for the results of the paper.

\section*{Soundness of Generalization}

\begin{theorem}[Soundness of Generalization]\label{thm-soundness-gen}
Let
$s = (\Pos, \; \LV, \; \AL, \; \PT, \; \LI,\linebreak
 \KB)$ be a state and let $\overline{s} = (\Pos, \; \overline{\LV}, \; \overline{\AL}, \; \overline{\PT}, \; \overline{\LI}, \; \overline{\KB})$ be a generalization of $s$ with instantiation $\mu$ according to the generalization rule. Then every concrete state $c = (\Pos, \; \LV^c, \; \AL^c, \; \PT^c, \; \varnothing, \; \KB^c)$ that is represented by $s$ is also represented by $\overline{s}$.
\end{theorem}

\begin{proof}
We prove that $c$ is represented by $\overline{s}$ as defined in Def. \ref{def:representation}.
Since $c$ is represented by $s$, we have $\domain(\LV^c) = \domain(\LV)$ and with $(g2)$
we have $\domain(\LV) = \domain(\overline{\LV})$, so $\domain(\LV^c)
= \domain(\overline{\LV})$.

Let $\sigma : \Vsym \to \Z$ such that $(\slassignment^c,\slmemory^c)$ is a model of
$\sigma(\stateformula{s}_\SL)$. We show that $(\slassignment^c,\slmemory^c)$ is also a
model of $\sigma(\mu(\stateformula{\overline{s}}_\SL))$
(i.e., for $\overline{\sigma} = \sigma \circ \mu$,  $(\slassignment^c,\slmemory^c)$ is a
model of $\overline{\sigma}(\stateformula{\overline{s}}_\SL)$.
For the subformula
$\sigma(\mu((\bigast\nolimits_{\varphi \in \overline{\AL}} \, \stateformula{\varphi}_\SL) \,\wedge\,
(\bigwedge\nolimits_{\varphi \in \overline{\PT}} \, \stateformula{\varphi}_\SL)))$, this
follows from soundness of the generalization rule of our earlier approach
in \cite{LLVM-JAR17}.

To prove that $(\slassignment^c,\slmemory^c)$ is a model of
    $\sigma(\mu(\stateformula{\overline{s}}))$ for the newly introduced conjuncts from
    $\stateformula{\overline{s}}$, assume that there is an $\overline{\listinv}:=
    (x_{\codevar{ad}} \pointstorec[\codevar{ty}]{x_\ell} [(\mathit{off}_i: \codevar{ty}_i:
    {x_i..\hat{x}_i})]_{i=1}^n) \in \overline{\LI}$. With $(g6)$, there
    either exists a corresponding invariant in $\LI$ or a corresponding concrete list in
    $s$. In the first case, there is an $\listinv:= (v_{\codevar{ad}} \pointstorec[\codevar{ty}]{v_\ell} [(\mathit{off}_i: \codevar{ty}_i: {v_i..\hat{v}_i})]_{i=1}^n) \in \LI$ with
    $\models \stateformula{s} \implies v_\codevar{ad} = \mu(x_\codevar{ad}) \wedge v_\ell = \mu(x_\ell)$ and
    $\models \stateformula{s} \implies v_i = \mu(x_i) \wedge \hat{v}_i = \mu(\hat{x}_i)$ for all $1 \leq i \leq n$.
This means that $\overline\listinv$ only renames the variables of $\listinv$ but the assignment of the variables has to be equal. Since $(\slassignment^c,\slmemory^c)$ is a model of $\sigma(\stateformula{s})$ and $\listinv$ yields the same conjuncts in $\stateformula{s}$ where corresponding variables are replaced, $(\slassignment^c,\slmemory^c)$ is also a model of the conjuncts added to $\sigma(\mu(\stateformula{\overline{s}}))$ for the list invariant $\overline{\listinv}$.
In the second case of $(g6)$ we have $\ell \geq 1$ and
    $\models \stateformula{s} \implies v^{\mathit{start}}_{1}
    = \mu(x_\codevar{ad}) \wedge \ell = \mu(x_\ell)$, where $v^{\mathit{start}}_{1}$ is the start
    address of an allocation from $\AL$ and thus $v^{\mathit{start}}_{1} \geq 1$ is a conjunct of
    $\stateformula{s}$. Since $(\slassignment^c,\slmemory^c)$ is a model of
    $\sigma(\stateformula{s})$, it is also a model of $\sigma(\mu(x_\ell \geq 1))$ and
    $\sigma(\mu(x_{\codevar{ad}} \geq 1))$. If
    $\models \stateformula{\overline{s}} \implies x_\ell = 1$, we have $\ell = 1$, so
    $v_{{\ell},i} = v_{1,i}$ and hence $\sigma(\mu(\bigwedge_{i=1}^n x_i = \hat{x}_i))
    = \sigma(\bigwedge_{i=1}^n v_{1,i} = v_{\ell,i})$ holds. If
    $\models \stateformula{\overline{s}} \implies x_\ell \geq 2$, by condition $(d)$ we
    have $\models \stateformula{s} \implies v_{1,j} = v^{\mathit{start}}_{2}$, where $v^{\mathit{start}}_{2}$ is the
    start address of an allocation from $\AL$ and thus $v^{\mathit{start}}_{2} \geq 1$ is a conjunct of
    $\stateformula{s}$. Since $(\slassignment^c,\slmemory^c)$ is a model of
    $\sigma(\stateformula{s})$ and $\models \stateformula{s} \implies v_{1j}
    = \mu(x_j)$, $(\slassignment^c,\slmemory^c)$ is also a model of
    $\sigma(\mu(x_j \geq 1))$. If there exists an $1 \leq i \leq n$ such that
    $\models \stateformula{\overline{s}} \implies x_k \neq \hat{x}_i$, in a similar way we can
    infer that 
    $\ell \geq 2$ and hence, $\sigma(\mu(x_\ell \geq 2))$ holds.

To show that $(\slassignment^c,\slmemory^c)$ is a model of
    $\sigma(\mu(\bigast\nolimits_{\varphi \in \overline{\LI}} \, \stateformula{\varphi}_\SL))$,
    we again assume $\overline{\listinv} \in \overline{\LI}$. Then there either exists a
    corresponding invariant in $\LI$ or a corresponding concrete list in $s$. In the first case, $\overline\listinv$ only renames the variables
    of $\listinv$ but the assignment of the variables has to be equal. Therefore, with
    $(\slassignment^c,\slmemory^c) \models\sigma(\stateformula{\listinv}_\SL)$ it follows
    directly that $(\slassignment^c,\slmemory^c)$ is also a model of
    $\sigma(\mu(\stateformula{\overline{\listinv}}_\SL))$.
    We now consider the second case of $(g6)$, where
    $v_{1}^\mathit{start}$ points to a concrete list of type \codevar{ty} and length $\ell$ with
    values $v_{k,i}$ (i.e., $(\star)$ the conditions $(a)-(d)$ hold),
    $\models \stateformula{s} \implies v^{\mathit{start}}_{1} = \mu(x_\codevar{ad}) \wedge \ell
    = \mu(x_\ell)$, and $\models \stateformula{s} \implies v_{1,i} = \mu(x_i) \wedge
    v_{\ell,i} = \mu(\hat{x}_i)$ for all $1 \leq i \leq n$. 
    We have
    \begin{align*}
    	&\sigma(\mu(\stateformula{\overline{\listinv}}_\SL))\\
    	=\;&\listpred_{\sizeOf(\code{ty}),j}(\sigma(\mu(x_\ell)), \sigma(\mu(x_\codevar{ad})), [(\mathit{off}_i: \sizeOf(\codevar{ty}_i): {\sigma(\mu(x_i))..\sigma(\mu(\hat{x}_i))})]_{i=1}^n)\\
    	=\;&\listpred_{\sizeOf(\code{ty}),j}(\ell, \sigma(v^{\mathit{start}}_{1}), [(\mathit{off}_i: \sizeOf(\codevar{ty}_i): {\sigma(v_{1,i})..\sigma(v_{\ell,i})})]_{i=1}^n)
    \end{align*}
    and show by induction over the length $\ell$ of the list that if
    $(\slassignment^c,\slmemory^c)$ is a model of $\sigma(\stateformula{s}_\SL)$ and
    $v^{\mathit{start}}_{1}$ points to a concrete list of length $\ell$ and
    type \codevar{ty} such that $(\star)$ holds, then
    $(\slassignment^c,\slmemory^c)$ is a model of
    $\listpred_{\sizeOf(\code{ty}),j}(\ell, \sigma(v^{\mathit{start}}_{1}),
    [(\mathit{off}_i: \sizeOf(\codevar{ty}_i):
    {\sigma(v_{1,i})..\sigma(v_{\ell,i})})]_{i=1}^n)$.
    
    If $\ell = 1$, the list consists of a single allocation $\alloc{v^{\mathit{start}}_{1}}{v^{\mathit{end}}_{1}}$ of the size of \codevar{ty}.
    Then $v^{\mathit{start}}_{1} \geq 1$ and $\stateformula{\alloc{v^{\mathit{start}}_{1}}{v^{\mathit{end}}_{1}}}_\SL =(\forall x. \exists y. \; (v^{\mathit{start}}_{1} \leq x \leq v^{\mathit{end}}_{1}) \implies (x \pointstosl y))$ are already conjuncts of $\stateformula{s}_\SL$.
       Trivially, $\sigma(v_{1,i}) = \sigma(v_{\ell,i})$ for all $1 \leq i \leq n$.
    With $(c)$ there exists $(v^{\mathit{start}}_{1,i} \pointsto[\codevar{ty}_i]
    v_{1,i}) \in \PT$ with $\models \, \stateformula{s} \implies v^{\mathit{start}}_{1,i}
    = v^{\mathit{start}}_{1} + \mathit{off}_i$ for all $1 \leq i \leq n$. From
    $\slmemory^c \models \sigma(\stateformula{v^{\mathit{start}}_{1,i}
 \pointsto[\codevar{ty}_i] v_{1,i}}_\SL)$ it follows that
    $\slmemory^c \models \sigma(\stateformula{v^{\mathit{start}}_{1}
    + \mathit{off}_i \pointsto[\codevar{ty}_i] v_{1,i}}_\SL)$.

    Now assume that the hypothesis holds for a fixed length $\ell' \geq 1$. We show that it
    holds for $\ell = \ell'+1$. Now, the list consists of $\ell'+1$ allocations. Again
    we have conjuncts $v^{\mathit{start}}_1 \geq 1$ and
    $\stateformula{\alloc{v^{\mathit{start}}_1}{v^{\mathit{end}}_1}}_\SL 
    =(\forall x. \exists y. \; (v^{\mathit{start}}_1 \leq x \leq v^{\mathit{end}}_1) \implies (x \pointstosl y))$
    in $\stateformula{s}_\SL$.
    Since $\sigma(v_{1,j}) = \sigma(v^{\mathit{start}}_{2})$ points to a concrete list of length $\ell'$ and type \codevar{ty} such that $(\star)$ holds, we can use the induction hypothesis to conclude that
    \[\slmemory^c \models\; \listpred_{\sizeOf(\code{ty}),j}(\ell', \sigma(v_{1,j}),
    [(\mathit{off}_i: \sizeOf(\codevar{ty}_i):
    {m_i..\sigma(v_{\ell,i})})]_{i=1}^n),\]
      where $\slmemory^c \models \sigma(v_{1j}) + \mathit{off}_i \pointstosl_{\sizeOf(\codevar{ty}_i)} m_i$ for all $1 \leq i \leq n$.
    Since $\alloc{v^{\mathit{start}}_{1}}{v^{\mathit{end}}_{1}}$ is different from all
    allocations $\alloc{v^{\mathit{start}}_{k}}{v^{\mathit{end}}_{k}}$ for $2 \leq k \leq \ell$ and all these allocations are in $\AL$, with $\slmemory^c \models\, (\bigast\nolimits_{\varphi \in \AL} \, \stateformula{\varphi}_\SL)$ there exist $\slmemory^c_1 \bot \slmemory^c_2$ such that $\slmemory^c = \slmemory^c_1 \uplus \slmemory^c_2$ where $\slmemory^c_1 \models (\forall x. \exists y. \; (v^{\mathit{start}}_1 \leq x \leq v^{\mathit{end}}_1) \implies (x \pointstosl y))$ and $\slmemory^c_2 \models \listpred_{\sizeOf(\code{ty}),j}(\ell', \sigma(v_{1,j}), [(\mathit{off}_i: \sizeOf(\codevar{ty}_i): {m_i..\sigma(v_{\ell,i})})]_{i=1}^n)$. Therefore,
    \begin{align*}
    	\slmemory^c \models\; &(\forall x. \exists y. \; (v^{\mathit{start}}_1 \leq x \leq v^{\mathit{end}}_1) \implies (x \pointstosl y)) \;*\,\\ 
    	&\listpred_{\sizeOf(\code{ty}),j}(\ell', \sigma(v_{1,j}), [(\mathit{off}_i: \sizeOf(\codevar{ty}_i): {m_i..\sigma(v_{\ell,i})})]_{i=1}^n)
    \end{align*}
    for some $m_1, \ldots, m_n \in \Z$.
    Furthermore, with $(c)$ there exists $(v^{\mathit{start}}_{1,i} \pointsto[\codevar{ty}_i] v_{1,i}) \in \PT$, so we have
    \begin{align*}
    	&\;\slmemory^c \models\; \sigma(\stateformula{v^{\mathit{start}}_{1,i} \pointsto[\codevar{ty}_i] v_{1,i}}_\SL)\\
    	\Leftrightarrow &\;\slmemory^c \models\; \sigma(v^{\mathit{start}}_{1,i}) \pointsto[\sizeOf(\codevar{ty}_i)] \sigma(v_{1,i})\\
    	\overset{(c)}{\Leftrightarrow} &\;\slmemory^c \models\; \sigma(v^{\mathit{start}}_{1}) + \mathit{off}_i \pointsto[\sizeOf(\codevar{ty}_i)] \sigma(v_{1,i})
    \end{align*}
    for all $1 \leq i \leq n$.
    Disjointness between the memory domain of $\overline{\listinv}$ and the
    memory domains of all other list invariants and allocation subformulas is preserved
    from $\stateformula{s}_\SL$.
    Similarly, in the first case of $(g6)$, disjointness between the memory domain of
    $\overline{\listinv}$ and the memory domain of $\overline{PT}$ is preserved from
    $\stateformula{s}_\SL$. In the second case of $(g6)$, this disjointness
    follows from the last condition.
    This concludes our induction and the proof that $(\slassignment^c,\slmemory^c)$ is a
    model of
\[\left(
\left((\bigast\nolimits_{\varphi \in \AL} \; \; \stateformula{\varphi}_\SL)
\; \wedge \;
(\bigwedge\nolimits_{\varphi \in \PT} \; \; \stateformula{\varphi}_\SL )\right)
\; \ast \;
(\bigast\nolimits_{\varphi \in \LI} \; \; \stateformula{\varphi}_\SL)\right).\]

Now, we prove the third condition of Def. \ref{def:representation}. By $(g4)$ there exists $\alloc{v_1}{v_2} \in \AL$ for all $\alloc{x_1}{x_2} \in \overline{\AL}$ such that $\models \stateformula{s} \implies  v_1 = \mu(x_1) \wedge v_2 = \mu(x_2)$. Since $c$ is represented by $s$, there exists $\alloc{w_1}{w_2} \in \AL^c$ for all $\alloc{v_1}{v_2} \in \AL$ such that $\models \stateformula{c} \Rightarrow w_1 = \sigma(v_1) \wedge w_2 = \sigma(v_2)$. Hence there exists $\alloc{w_1}{w_2} \in \AL^c$ for all $\alloc{x_1}{x_2} \in \overline{\AL}$ such that $\models \stateformula{c} \Rightarrow w_1 = \sigma(\mu(x_1)) \wedge w_2 = \sigma(\mu(x_2))$.

Finally, we show that the fourth condition of Def. \ref{def:representation} holds.
Again, let $\overline{\listinv}:= (x_{\codevar{ad}} \pointstorec[\codevar{ty}]{x_\ell}
[(\mathit{off}_i: \codevar{ty}_i:
{x_i..\hat{x}_i})]_{i=1}^n) \in \overline{\LI}$. With $(g6)$, there either
exists a corresponding invariant in $\LI$ (which is equal except for variable renamings) or a corresponding concrete list in $s$. The first case is
trivial. In the second case, the condition is again implied by $(b)$, $(c)$, and
$(d)$.
\qed
\end{proof}

%%%%%%%%%%%%%%%%%%%%%%%%%%%%%%%%%%%%%%%%%%%%%%%%%%%%%%%%%%%%%%%%%%

\section*{Soundness of List Extension}

Next we prove the soundness of the new list extension rule. To this end, we have to show
that if a concrete state is represented by an
abstract state with a corresponding list invariant
and the concrete state is evaluated via a 
\code{store} instruction, then this evaluation can also be represented by applying the
list extension rule on the abstract state.

In \cite{LLVM-JAR17},
we proved the soundness of our
 symbolic execution rules with respect to the
formal \LLVM{} semantics defined by the \tool{Vellvm} project \cite{Vellvm}. 
Hence, the evaluation of the concrete state via the \code{store} instruction can be
modeled by the general \codevar{store} rule for concrete
states from \cite{LLVM-JAR17}, which we recapitulate below, 
adapted to our definition of states.
The value  $t$
 of type \code{ty} has to be converted to
 \code{i8}
values for the byte-wise memory representation of concrete states.
Since we interpret
all integers stored in the memory as unsigned (as opposed to signed) integers, this
conversion can now be simplified.

\vspace*{.3cm}
\noindent
\fbox{
\begin{minipage}{11.8cm}
\label{rule:generalization}
\small
\mbox{\small \textbf{\hspace*{-.15cm}\codevar{store} to allocated memory $(p:$
``\code{store ty $t$, ty* pa}'', $t \in \Ids \cup \Z$, \code{pa} $\in \Ids$)}}\\
\vspace*{1mm}\\
\centerline{$\frac{\parbox{4.5cm}{\centerline{
$s = (\Pos, \; \LV, \; \AL, \; \PT, \; \emptyset, \; \KB)$}\vspace*{.1cm}}}
{\parbox{4.7cm}{\vspace*{.1cm} \centerline{
$s' = (\Pos^+, \; \LV, \; \AL, \; \PT', \; \emptyset, \; \KB')$}}}\;\;\;\;$
\mbox{if}} \vspace*{-.05cm}
{\small
\begin{itemize}
\item[$\bullet$] there is $\alloc{w_1}{w_2} \in \AL$ with $\models \, \stateformula{s} \implies (w_1 \leq \LV(\codevar{pa}) \wedge \LV(\codevar{pa})+\sizeOf(\code{ty})-1 \leq w_2)$
\item[$\bullet$] there are $(y_0 \pointsto[\codevar{i8}] z_0), \ldots, (y_{\sizeOf(\codevar{ty})-1} \pointsto[\codevar{i8}] z_{\sizeOf(\codevar{ty})-1}) \in \PT$ such that $\models \, \stateformula{s} \implies \LV(\codevar{pa}) = y_0 \wedge \bigwedge_{1 \leq i \leq \sizeOf(\codevar{ty})-1} y_i = y_0 + i$
\item[$\bullet$] $\PT' = (\PT\backslash\{(y_0 \pointsto[\codevar{i8}] z_0),\ldots,(y_{\sizeOf(\codevar{ty})-1} \pointsto[\codevar{i8}] z_{\sizeOf(\codevar{ty})-1})\}) \,\cup\, \{(y_0 \pointsto[\codevar{i8}] v_0),\ldots,(y_{\sizeOf(\codevar{ty})-1} \pointsto[\codevar{i8}] v_{\sizeOf(\codevar{ty})-1})\}$
\item[$\bullet$] $\KB' = \KB \,\cup\, \{v_i = u_i \,\mid\, 0 \leq i \leq \sizeOf(\codevar{ty})-1\}$
\item[$\bullet$] for $0 \leq i \leq \sizeOf(\codevar{ty})-1$ let $u_i = (t \operatorname{div} 2^{8 \cdot i}) \mod 2^8$
\item[$\bullet$] $v_0, \ldots, v_{\sizeOf(\codevar{ty})-1} \in \Vsym$ are fresh
\end{itemize}}
\end{minipage}}
\vspace*{.2cm}

\begin{theorem}[Soundness of List Extension]\label{thm-soundness-ext}
Let
$s$ be a state and let $c$ be a concrete state that is represented by $s$. Moreover, let
$s'$ be the state resulting from $s$
by applying the list extension rule and let $c'$ be the concrete state that results from
$c$ by evaluating the general \codevar{store} rule for concrete states. Then $c'$ is represented by $s'$.
\end{theorem}

\begin{proof}
Neither $s$ nor $c$ is $\ERROR$ since $\ERROR$ has no successor states, so we have $s = (\Pos, \; \LV^s, \; \AL^s, \; \PT^s, \; \LI^s, \; \KB^s)$ and $c = (\Pos, \; \LV^c, \; \AL^c, \; \PT^c, \; \emptyset, \; \KB^c)$.
By construction of the rules we get
\[ \begin{array}{rcl}
s' &=& (\Pos^+, \; \LV^{s'}, \; \AL^{s'}, \; \PT^{s'}, \; \LI^{s'}, \; \KB^{s'})\\
c' &=& (\Pos^+, \; \LV^{c'}, \; \AL^{c'}, \; \PT^{c'}, \; \emptyset, \; \KB^{c'})
\end{array}\] 
where\\
\\
\begin{tabular}{lll}
	$\LV^{s'}$&\!\!$=$&$\LV^s$\\
	$\AL^{s'}$&\!\!$=$&$\AL^s\backslash\{\alloc{v^\mathit{start}}{v^\mathit{end}}\}$\\
	$\PT^{s'}$&\!\!$=$&$\{(x_1 \hookrightarrow_{\codevar{sy}} x_2) \in \PT^s \mid\, \models \stateformula{s} \implies$\\
	&&$(v^\mathit{end} < x_1) \, \lor \,
        (x_1+\sizeOf(\codevar{sy})-1 < v^\mathit{start})$\\
	$\LI^{s'}$&\!\!$=$&$(\LI^s\backslash\{\listinv\}) \cup \{\listinv'\}$\\
	$\KB^{s'}$&\!\!$=$&$\KB^s \cup \{v_k = \LV^s(t),\, v'_\ell =
	v_\ell+1\}$\\        
        \\
  $\LV^{c'}$&\!\!$=$&$\LV^c$\\
$\AL^{c'}$&\!\!$=$&$\AL^c$\\
$\PT^{c'}$&\!\!$=$&$(\PT^c\backslash\{(y_0 \pointsto[\codevar{i8}] z_0),\ldots,$\\
&&$(y_{\sizeOf(\codevar{ty})-1} \pointsto[\codevar{i8}] z_{\sizeOf(\codevar{ty})-1})\})$\\
&&$\cup\, \{(y_0 \pointsto[\codevar{i8}] v_0),\ldots,$\\
&&$(y_{\sizeOf(\codevar{ty})-1} \pointsto[\codevar{i8}] v_{\sizeOf(\codevar{ty})-1})\}$\\
$\KB^{c'}$&\!\!$=$&$\KB^c \cup \{v_i = u_i \,\mid\,
	0 \leq i \leq \sizeOf(\codevar{ty})-1\}$         
\end{tabular}

\vspace{0.25cm}
\noindent
Here, $\listinv = (v_\mathit{ad} \pointstorec[\codevar{lty}]{v_\ell} [(\mathit{off}_i: \codevar{lty}_i: {w_i..\hat{w}_i})]_{i=1}^n)$ and $\listinv' = (v^\mathit{start} \pointstorec[\codevar{lty}]{v'_\ell}[(\mathit{off}_i: \codevar{lty}_i:{v_i..\hat{w}_i})]_{i=1}^n)$.
We prove that $c'$ is represented by $s'$ as defined in Def. \ref{def:representation}.
Since $c$ is represented by $s$, we have $\domain(\LV^c) = \domain(\LV^s)$ and thus $\domain(\LV^{c'}) = \domain(\LV^{s'})$.

Let $\sigma : \Vsym \to \Z$ such that $(\slassignment^c,\slmemory^c)$ is a model of $\sigma(\stateformula{s}_\SL)$. We show that there exists $\sigma' : \Vsym \to \Z$ such that $(\slassignment^{c'},\slmemory^{c'})$ is a model of $\sigma'(\stateformula{s'}_\SL)$. We define $\sigma' = \sigma[v_k:=\sigma(\LV^s(t)), v'_\ell:=\sigma(v_\ell)+1]$. By construction of $\LV^{c'}$ and $\PT^{c'}$ we have $\slassignment^{c'} = \slassignment^c$ and $\slmemory^{c'}(n) = \slmemory^c(n)$ for all $n \in \Nplus\backslash\{n_0,\ldots,n_{\sizeOf(\codevar{ty})-1}\}$, where $n_i$ is the number with $\models \stateformula{c} \implies y_i = n_i$. Moreover, let $\slmemory^{c'}(n_i) = u_i$ for all $i \in \{0,\ldots,\sizeOf(\codevar{ty})-1\}$.
The formula $\sigma'(\stateformula{s'}_\SL)$ differs from $\sigma(\stateformula{s}_\SL)$ as follows:
\begin{itemize}
	\item[(1)] The conjuncts for the allocation
	$\alloc{v^\mathit{start}}{v^\mathit{end}}$ are missing to prevent further reasoning about parts of the heap that are included in the list invariant.
	\item[(2)] Some points-to conjuncts are missing such that for all conjuncts $x_1 \hookrightarrow x_2$ that are kept we know $\sigma(x_1) \not\in \{n_0,\ldots,n_{\sizeOf(\codevar{ty})-1}\}$.
	\item[(3)] Instead of the subformula
\[\sigma(\stateformula{\listinv}_\SL)
	= \sigma(\listpred_{\sizeOf(\code{lty}),j}(v_\ell, v_\mathit{ad},
	[(\mathit{off}_i: \sizeOf(\codevar{lty}_i): {w_i..\hat{w}_i})]_{i=1}^n)),\] we
	now have \[\sigma'(\stateformula{\listinv'}_\SL) = \sigma'(\listpred_{\sizeOf(\code{lty}),j}(v'_\ell, v^\mathit{start}, [(\mathit{off}_i: \sizeOf(\codevar{lty}_i): {v_i..\hat{w}_i})]_{i=1}^n)).\]
	\item[(4)] The conjuncts that are added for $l$ to $\stateformula{s}$ are
	missing. As new conjuncts for $l'$, $\sigma'(v'_\ell) \geq 1$,
	$\sigma'(v^\mathit{start}) \geq 1$, $\sigma'(v_j) \geq 1$, and
	$\sigma'(v'_\ell) \geq 2$ have been added.
	\item[(5)] The conjuncts $\sigma'(v_k) = \sigma'(\LV^s(t))$ and $\sigma'(v'_\ell)
	= \sigma'(v_\ell)+1$ have been added, i.e., $\sigma'(v_k) = \sigma(\LV^s(t))$ and $\sigma'(v'_\ell) = \sigma(v_\ell)+1$.
\end{itemize}
Since $\slmemory^{c'}$ behaves like $\slmemory^c$ on all addresses except $\{n_0,\ldots,n_{\sizeOf(\codevar{ty})-1}\}$, the conjuncts that were kept from $\sigma(\stateformula{s}_\SL)$ are satisfied by $(\slassignment^{c'},\slmemory^{c'})$. Moreover, we can ignore subformulas that were removed since this only makes the formula weaker.
\begin{itemize}
\item[(3)]
Since $\sigma'(v'_\ell) = \sigma(v_\ell) + 1 \geq 2$, $\slmemory^{c'} \models \sigma'(\stateformula{\listinv'}_\SL)$ iff
\begin{itemize}
	\item $\sigma'(v^\mathit{start}) \geq 1$, which is true as with $\alloc{v^\mathit{start}}{v^\mathit{end}} \in \AL^s$, $\sigma(v^\mathit{start}) \geq 1$ is a conjunct of $\sigma(\stateformula{s})$ and we have $\sigma'(v^\mathit{start}) = \sigma(v^\mathit{start})$,
	\item $\slmemory^{c'} \models (\forall x. \exists y. \; (v^\mathit{start} \leq
	x \leq v^\mathit{start}+\sizeOf(\code{lty})-1) \implies (x \pointstosl y)) \,*\, \listpred_{\sizeOf(\code{lty}),j}(\sigma'(v'_\ell)-1, \sigma'(v_j), [(\mathit{off}_i: \sizeOf(\codevar{lty_i}): {\sigma(w_i)..\sigma'(\hat{w}_i)})]_{i=1}^n)$.
    We have
	    \begin{align*}
		    &\listpred_{\sizeOf(\code{lty}),j}(\sigma'(v'_\ell)-1, \sigma'(v_j), [(\mathit{off}_i: \sizeOf(\codevar{lty_i}): {\sigma(w_i)..\sigma'(\hat{w}_i)})]_{i=1}^n)\\
		    =\;&\listpred_{\sizeOf(\code{lty}),j}(\sigma(v_\ell), \sigma(v_\mathit{ad}), [(\mathit{off}_i: \sizeOf(\codevar{lty_i}): {\sigma(w_i)..\sigma(\hat{w}_i)})]_{i=1}^n)
	    \end{align*}
	    since $\sigma'(v'_\ell) = \sigma(v_\ell) + 1$ and $\sigma(v_\mathit{ad})
    = \sigma'(v_j)$.
    This is equal
    to the removed list predicate $\sigma(\stateformula{\listinv}_\SL)$, for which $\slmemory^{c}$ (and hence $\slmemory^{c'}$) is a model. Furthermore, $(\forall x. \exists y. \; (v^\mathit{start} \leq
	x \leq v^\mathit{start}+\sizeOf(\code{lty})-1) \implies (x \pointstosl y))$ holds since $\sigma(\stateformula{\alloc{v^\mathit{start}}{v^\mathit{end}}}_\SL)$
	is a conjunct of $\sigma(\stateformula{s}_\SL)$. It indeed operates on a different
	domain of $\slmemory^{c'}$ since $\stateformula{s}_\SL$ ensures that
	$\sigma(\stateformula{\alloc{v^\mathit{start}}{v^\mathit{end}}}_\SL)$ operates on
	a domain disjoint from the domain of the removed list predicate
	$\sigma(\stateformula{\listinv}_\SL)$ and since
	$\sigma'(v^\mathit{start})+\mathit{off}_i \geq \sigma'(v^\mathit{start})$ and
	$\sigma'(v^\mathit{start})+\mathit{off}_i+\sizeOf(\codevar{lty}_i)-1 \leq \sigma'(v^\mathit{end})$
	for all $1 \leq i \leq n$.
        \item for all $1 \leq i \leq n$ we have $\slmemory^{c'} \models \, (\sigma'(v^\mathit{start})+\mathit{off}_i \pointstosl_{\sizeOf(\codevar{lty_i})} \sigma'(v_i))$.
	If $i \neq k$ this follows from the condition of
    the list extension rule that there exist $v^\mathit{start}_i,v_i \in \Vsym$ with
    $\models \, \stateformula{s} \implies v^\mathit{start}_i = v^\mathit{start}
    + \mathit{off}_i$ and $(v^\mathit{start}_i \pointsto[\codevar{lty}_i] v_i) \in \PT$,
    so $\stateformula{\sigma(v^\mathit{start})
    + \mathit{off}_i \pointsto[\codevar{lty}_i] \sigma(v_i)}_\SL$ is a conjunct of
    $\sigma(\stateformula{s}_\SL)$.
    If $i = k$ we have to show $\slmemory^{c'} \models \,
    (\sigma(\LV^s(\codevar{pa})) \pointstosl_{\sizeOf(\codevar{lty}_i)} \sigma(\LV^s(t)))$,
    i.e.,
    $\slmemory^{c'} \models \, (\sigma(\LV^s(\codevar{pa})) +
    m \pointstosl \lfloor\frac{\sigma(\LV^s(t))}{2^{8 \cdot m}}\rfloor \mod 2^8)$
for all $0 \leq m \leq \sizeOf(\codevar{lty}_i) -1$,
which holds by construction with $u_i$ as given in the \codevar{store} rule for concrete states.
\end{itemize}
 Disjointness of the memory domain of $\sigma'(\stateformula{\listinv'}_\SL)$ and the memory domain of $\sigma'(\bigast\nolimits_{\varphi \in \AL^s\backslash\{\alloc{v^\mathit{start}}{v^\mathit{end}}\}} \; \; \stateformula{\varphi}_\SL)$ follows directly from $\stateformula{s}_\SL$. Disjointness of the memory domain of $\sigma'(\stateformula{\listinv'}_\SL)$ and the memory domain of $\sigma'(\bigwedge\nolimits_{\varphi \in \PT^{s'}} \; \; \stateformula{\varphi}_\SL)$ follows from $\stateformula{s}_\SL$ and (2).

\item[(4)]
The conjunct $\sigma'(v'_\ell) \geq 2$ holds with $\sigma'(v'_\ell) = \sigma(v_\ell) + 1 \geq 2$.
Since $\alloc{v^\mathit{start}}{v^\mathit{end}} \in \AL^s$, the conjunct $\sigma(v^\mathit{start}) \geq 1$ is part of $\sigma(\stateformula{s})$, so $\sigma'(v^\mathit{start}) \geq 1$ also holds.
We have $\sigma'(v_j) = \sigma(v_\mathit{ad})$ and $v_\mathit{ad}$ is the start address of $l$. Therefore, $\sigma(v_\mathit{ad}) \geq 1$ is a conjunct of $\sigma(\stateformula{s})$ and thus we have $\sigma'(v_j) \geq 1$.

\item[(5)]
The conjuncts $\sigma'(v_k) = \sigma(\LV^s(t))$ and $\sigma'(v'_\ell) = \sigma(v_\ell)+1$ are trivially satisfied by construction of $\sigma'$.
\end{itemize}

The third condition of Def. \ref{def:representation} holds as $\AL^{s'} \subseteq \AL^s$ and $\AL^{c'} = \AL^c$.

Finally, we show that for the new list invariant $v^\mathit{start} \pointstorec[\codevar{lty}]{v'_\ell} [(\mathit{off}_i: \codevar{lty}_i: {v_i..\hat{w}_i})]_{i=1}^n$ with $\codevar{lty}_j = \codevar{lty*}$ there exist allocations $\alloc{w^\mathit{start}_r}{w^\mathit{end}_r} \in \AL^{c'}$, $1 \leq r \leq \sigma(v'_\ell)$, such that $(\slassignment^{c'},\slmemory^{c'})$ is a model of
    $(\bigwedge_{r=1}^{\sigma'(v'_\ell)} m^\mathit{end}_r - m^\mathit{start}_r
    = \sizeOf(\codevar{lty})-1) \,\wedge\, \sigma'(v^\mathit{start}) =
    m^\mathit{start}_1 \,\wedge\,\bigwedge_{r=1}^{\sigma'(v'_\ell)-1} m^\mathit{start}_r
    + \mathit{off}_j \hookrightarrow_{\sizeOf(\codevar{lty*})} m^\mathit{start}_{r+1}$,
where
$m^\mathit{start}_k$ resp.\ $m^\mathit{end}_k$
is the number with $\models \stateformula{c} \implies w^\mathit{start}_k =
    m^\mathit{start}_k$
resp.\ 
 $\models \stateformula{c} \implies w^\mathit{end}_k =
    m^\mathit{end}_k$.
    For $2 \leq r \leq \sigma'(v'_\ell)$ this follows from the fact that for the former list invariant from $\LI^s$ such allocations exist in $\AL^c$ and $\AL^{c'} = \AL^c$. We choose $m^\mathit{start}_1 = \sigma(v^\mathit{start})$, $m^\mathit{end}_1 = \sigma(v^\mathit{end})$, and $m^\mathit{start}_2 = \sigma(v_\mathit{ad})$. Since $\models \stateformula{s} \implies v^\mathit{end} = v^\mathit{start} + \sizeOf(\codevar{lty})-1$ is a condition of the \codevar{store} rule and $(\slassignment^c,\slmemory^c)$ is a model of $\stateformula{s}$, $m^\mathit{end}_1 - m^\mathit{start}_1 = \sizeOf(\codevar{lty})-1$ holds. Furthermore, we have $\sigma'(v^\mathit{start}) = m^\mathit{start}_1$ and $\slmemory^{c'} \models m^\mathit{start}_1 + \mathit{off}_j \hookrightarrow_{\sizeOf(\codevar{lty*})} m^\mathit{start}_2$.
\qed
\end{proof}

%%%%%%%%%%%%%%%%%%%%%%%%%%%%%%%%%%%%%%%%%%%%%%%%%%%%%%%%%%%%%%%%%%

\section*{Soundness of List Traversal}

To prove soundness of the symbolic execution rule for list traversal, we
first recapitulate the general rule for \code{getelementptr} from \cite{LLVM-JAR17} for the case
of \code{i8} pointers and concrete states.

\vspace*{.3cm}
\noindent
\fbox{
\begin{minipage}{11.8cm}
\label{rule:generalization}
\small
\mbox{\small \textbf{\hspace*{-.15cm}\codevar{getelementptr} on \code{i8} pointer}}\\
\mbox{\small \textbf{\hspace*{-.15cm}$(p:$ ``\code{pb = getelementptr i8, i8* pa, i$m$ $t$}'', $t \in \Ids \cup \N$, $\code{pa}, \code{pb} \in \Ids$)}}\\
\vspace*{1mm}\\
\centerline{$\frac{\parbox{5cm}{\centerline{
$s = (\Pos, \; \LV, \; \AL, \; \PT, \; \emptyset, \; \KB)$}\vspace*{.1cm}}}
{\parbox{9.4cm}{\vspace*{.1cm} \centerline{
$s' = (\Pos^+, \; \LV[\code{pb}:=w], \; \AL, \; \PT, \; \emptyset, \; \KB \cup \{w = \LV(\code{pa})+\LV(t)\})$}}}\;\;\;\;$
\mbox{if}} \vspace*{-.05cm}
{\small
\begin{itemize}
\item[$\bullet$] $w \in \Vsym$ is fresh
\end{itemize}}
\end{minipage}}
\vspace*{.2cm}

\begin{theorem}[Soundness of List Traversal]\label{thm-soundness-trav}
Let $s$ be a state and let $c$ be a concrete state that is represented by $s$. Moreover,
let $s'$ be the state resulting from $s$ by applying the list traversal rule and let $c'$
be the concrete state that results from $c$ by evaluating the general
\codevar{getelementptr} rule for concrete states. Then $c'$ is represented by $s'$.
\end{theorem}

\begin{proof}
We only prove the soundness for the variant of the list traversal rule from \rSC{sect:List
Traversal}.
The proof for traversal via field access is basically the same and
  the proofs for the cases where the list invariant has length $1$ are analogous, but
  easier, 
  since we do not have the new list invariant $\listinv'$, while everything else says the
  same. The proofs for the variants with split lists are similar but use disjointness of the individual list invariants instead of disjointness between the allocations and the list invariants.

Let $s = (\Pos, \; \LV^s, \; \AL^s, \; \PT^s, \; \LI^s, \; \KB^s)$ and $c = (\Pos, \; \LV^c, \; \AL^c, \; \PT^c, \; \emptyset, \; \KB^c)$.
By construction of the rules we get
\[ \begin{array}{rcl}
s' &=& (\Pos^+, \; \LV^{s'}, \; \AL^{s'}, \; \PT^{s'}, \; \LI^{s'}, \; \KB^{s'})\\
c' &=& (\Pos^+, \; \LV^{c'}, \; \AL^{c'}, \; \PT^{c'}, \; \emptyset, \; \KB^{c'})
\end{array}\]
where\\
\\
\begin{tabular}{lll}
	$\LV^{s'}$&\!\!$=$&$\LV^s[\code{pb}:=w^{\mathit{start}}]$\\
	$\AL^{s'}$&\!\!$=$&$\AL^s\cup\{\alloc{v^{\mathit{start}}}{v^{\mathit{end}}}\}$\\
	$\PT^{s'}$&\!\!$=$&$\PT^s\cup\{(v_i^{\mathit{start}} \hookrightarrow_{\codevar{ty}_i} v_i) \mid\, 1 \leq i \leq n\}$\\
	$\LI^{s'}$&\!\!$=$&$(\LI^s\backslash\{\listinv\}) \cup \{\listinv'\}$\\
	$\KB^{s'}$&\!\!$=$&$\KB^s \cup \{v^{\mathit{start}} =
	v_\mathit{ad},\,w^{\mathit{start}} = v_j \}\,\cup$\\
	&&$\{v^{\mathit{end}} = v^{\mathit{start}} + \sizeOf(\code{ty}) -1 \} \, \cup$\\
	&&$\{w_\ell = v_\ell-1,\, w_j^{\mathit{start}} = w^{\mathit{start}} + \mathit{off}_j\} \, \cup$\\
	&&$\{v_i^{\mathit{start}} = v_\mathit{ad} + \mathit{off}_i \mid 1 \leq i \leq
	n \}$\\
        \\
$\LV^{c'}$&\!\!$=$&$\LV^c[\code{pb}:=w]$\\
$\AL^{c'}$&\!\!$=$&$\AL^c$\\
$\PT^{c'}$&\!\!$=$&$\PT^c$\\
$\KB^{c'}$&\!\!$=$&$\KB^c \cup \{w = \LV(\code{pa})+\LV(t)\}$
\end{tabular}

\vspace{0.25cm}
\noindent
Here, $\listinv = (v_\mathit{ad} \pointstorec[\codevar{ty}]{v_\ell} [(\mathit{off}_i: \codevar{ty}_i:{v_i..\hat{v}_i})]_{i=1}^n)$ and $\listinv' = (w^{\mathit{start}} \pointstorec[\codevar{ty}]{w_\ell}[(\mathit{off}_i: \codevar{ty}_i:{w_i..\hat{v}_i})]_{i=1}^n)$.
We prove that $c'$ is represented by $s'$ as defined in Def. \ref{def:representation}.
Since $c$ is represented by $s$, we have $\domain(\LV^c) = \domain(\LV^s)$ and thus $\domain(\LV^{c'}) = \domain(\LV^{c}) \cup \{\code{pb}\} = \domain(\LV^{s'})$.

Let $\sigma : \Vsym \to \Z$ such that $(\slassignment^c,\slmemory^c)$ is a model of $\sigma(\stateformula{s}_\SL)$. We show that there exists $\sigma' : \Vsym \to \Z$ such that $(\slassignment^{c'},\slmemory^{c'})$ is a model of $\sigma'(\stateformula{s'}_\SL)$.
From
$\slmemory^c \models \listpred_{\sizeOf(\code{ty}),j}(\sigma(v_\ell), \sigma(v_\mathit{ad}),
[(\mathit{off}_i: \sizeOf(\codevar{ty}_i): {\sigma(v_i)..\sigma(\hat{v}_i)})]_{i=1}^n)$
and $\models \stateformula{s} \implies v_\ell \geq 2$ it follows that for all $1 \leq
i \leq n$, $\slmemory^c$ is a model of
\begin{equation}
\label{listFormulaProof}\tag{$\star\star$}
(\ldots) \,*\, \listpred_{\sizeOf(\code{ty}),j}(\sigma(v_\ell)-1, \sigma(v_j),
[(\mathit{off}_i: \sizeOf(\codevar{ty}_i): {u_i..\sigma(\hat{v}_i)})]_{i=1}^n)
\end{equation}
with $u_i \in \Z$ such that $\slmemory^c \models \sigma(v_j) + \mathit{off}_i \pointstosl_{\sizeOf(\codevar{ty}_i)} u_i$.
Then, we define $\sigma' = \sigma[v^{\mathit{start}}:=\sigma(v_\mathit{ad}),
v^{\mathit{end}}:=\sigma(v_\mathit{ad})+\sizeOf(\code{ty})-1,
v_1^{\mathit{start}}:=\sigma(v_\mathit{ad})+\mathit{off}_1, \ldots,
v_n^{\mathit{start}}:=\sigma(v_\mathit{ad})+\mathit{off}_n,
w^{\mathit{start}}:=\sigma(v_j), w_\ell:=\sigma(v_\ell)-1,
w_j^{\mathit{start}}:=\sigma(v_j)+\mathit{off}_j, w_1:=u_1, \ldots,$\linebreak $w_n:=u_n]$ for $u_1, \ldots, u_n$ as above.
By construction of $\LV^{c'}$ and $\PT^{c'}$ we have $\slassignment^{c'}(\code{x}) = \slassignment^c(\code{x})$ for all $\code{x} \in \Ids\backslash\{\code{pb}\}$ and $\slmemory^{c'} = \slmemory^c$. Moreover, let $\slassignment^{c'}(\code{pb}) = \slassignment^{c}(\code{pa})+u$ for the number $u \in \N$ with $\models \stateformula{c} \implies \LV^c(t) = u$.
Recall that $\slmemory^{c'}=\slmemory^c$ and that $\slassignment^{c'}$ only differs
on \code{pb} from $\slassignment^{c}$.
Moreover, if $\code{pb} = \LV^s(\code{pb})$ is in $\stateformula{s}_\SL$, then it is
removed in  $\stateformula{s'}_\SL$ (and replaced by $\code{pb} = w_j^{\mathit{start}}$). So
the conjuncts that are kept from $\sigma(\stateformula{s}_\SL)$ are satisfied by $(\slassignment^{c'},\slmemory^{c'})$.
Therefore, we now focus on the subformulas of $\sigma'(\stateformula{s'}_\SL)$ that are new compared to $\sigma(\stateformula{s}_\SL)$. We have to show that $\slmemory^{c'}$ is a model of
\begin{itemize}
    \item[(1)] $\slassignment^{c'}(\sigma'(\code{pb} = w_j^{\mathit{start}}))$,
    \item[(2)] $\sigma'(\stateformula{\alloc{v^{\mathit{start}}}{v^{\mathit{end}}}}_\SL)$,
    \item[(3)] the formulas in $\stateformula{s'}$ for the new allocation:
    \begin{itemize}
        \item[(a)] $\sigma'(v^{\mathit{start}} \geq 1)$,
        \item[(b)] $\sigma'(v^{\mathit{start}} \leq v^{\mathit{end}})$,
        \item[(c)] $\sigma'(v^{\mathit{end}} < x_1 \vee x_2 < v^{\mathit{start}})$ for all $\alloc{x_1}{x_2} \in \AL^s$,
    \end{itemize}
    \item[(4)] $\sigma'(\stateformula{v_i^{\mathit{start}} \hookrightarrow_{\codevar{ty}_i} v_i}_\SL)$ for all $1 \leq i \leq n$,
    \item[(5)] the formulas in $\stateformula{s'}$ for the new points-to entries:
    \begin{itemize}
        \item[(a)] $\sigma'(v_i^{\mathit{start}} \geq 1)$ for all $1 \leq i \leq n$,
        \item[(b)] $\sigma'(v_i = y)$ if $\models \stateformula{s'} \implies
    v_i^{\mathit{start}} = x$ for all $1 \leq i \leq n$ and
    $(x \hookrightarrow_{\code{ty}_i} y) \in \PT^{s'}$,
        \item[(c)] $\sigma'(v_i^{\mathit{start}} \neq x)$ if $\models \stateformula{s'} \implies v_i \neq y$ for all $1 \leq i \leq n$ and $(x \hookrightarrow_{\code{ty}_i} y) \in \PT^{s'}$,
    \end{itemize}
	\item[(6)]
            $\sigma'(\stateformula{\listinv'}_\SL) = \sigma'(\listpred_{\sizeOf(\code{ty}),j}(w_\ell, w^{\mathit{start}}, [(\mathit{off}_i: \sizeOf(\codevar{ty}_i): {w_i..\hat{v}_i})]_{i=1}^n))$,
	\item[(7)] the formulas in $\stateformula{s'}$ for $\listinv'$:
	\begin{itemize}
	    \item[(a)] $\sigma'(w_\ell \geq 1)$, $\sigma'(w^{\mathit{start}} \geq 1)$,
	    \item[(b)] $\sigma'(\bigwedge_{i=1}^n w_i = \hat{v}_i)$ if $\models \stateformula{s'} \implies w_\ell = 1$,
	    \item[(c)] $\sigma'(w_j \geq 1)$ if $\models \stateformula{s'} \implies w_\ell \geq 2$,
	    \item[(d)] $\sigma'(w_\ell \geq 2)$ if there exists $1 \leq i \leq n$ such
    that $\models \stateformula{s'} \implies w_i \neq \hat{v}_i$, 
	\end{itemize}
	\item[(8)] the formulas in $\stateformula{s'}$ for $\KB^{s'}$:
	\begin{itemize}
	    \item[(a)] $\sigma'(v^{\mathit{start}} = v_\mathit{ad})$,
	    \item[(b)] $\sigma'(w^{\mathit{start}} = v_j)$,
	    \item[(c)] $\sigma'(v^{\mathit{end}} = v^{\mathit{start}} + \sizeOf(\code{ty}) - 1)$,
	    \item[(d)] $\sigma'(w_\ell = v_\ell-1)$,
	    \item[(e)] $\sigma'(w_j^{\mathit{start}} = w^{\mathit{start}} + \mathit{off}_j)$,
	    \item[(f)] $\sigma'(v_i^{\mathit{start}} = v_\mathit{ad} + \mathit{off}_i)$ for all $1 \leq i \leq n$.
	\end{itemize}
\end{itemize}
In the following we prove that $\slmemory^{c}$ is a model of all added conjuncts. It
directly follows that $\slmemory^{c'}$ is also a model of the conjuncts, since
$\slmemory^{c'} = \slmemory^{c}$.
\begin{itemize}
\item[(1)] We have $\slassignment^{c'}(\code{pb}) = \slassignment^{c}(\code{pa}) + u$ for
the number $u \in \N$ with $\models \stateformula{c} \implies \LV^c(t) = u$. With
$\models \stateformula{s} \implies \LV^s(\code{pa}) = v_j$
and $(\slassignment^{c}, \slmemory^{c}) \models \sigma(\stateformula{s})$, we must have
$\slassignment^{c}(\code{pa}) = \sigma(v_j)$, 
and with $\models \stateformula{s} \implies \LV^s(t) = \mathit{off}_j$ we obtain
$u = \mathit{off}_j$. Hence, we get $\slassignment^{c'}(\code{pb})
= \slassignment^{c}(\code{pa}) + u = \sigma(v_j) + \mathit{off}_j
= \sigma'(w_j^{\mathit{start}})$.
\item[(2)] Since $\slmemory^{c} \models \sigma(\stateformula{\listinv}_\SL)$, we have $\slmemory^{c} \models (\forall x. \exists y. \; (\sigma(v_\mathit{ad}) \leq x \leq \sigma(v_\mathit{ad}) + \sizeOf(\code{ty}) - 1) \implies (x \pointstosl y))$.
    With $\sigma'(v^{\mathit{start}}) = \sigma(v_\mathit{ad})$ and $\sigma'(v^{\mathit{end}}) = \sigma(v_\mathit{ad}) + \sizeOf(\code{ty}) - 1$, we have $\slmemory^{c} \models (\forall x. \exists y. \; (\sigma'(v^{\mathit{start}}) \leq x \leq \sigma'(v^{\mathit{end}})) \implies (x \pointstosl y))$.
\item[(3)] (a) Since $\slmemory^{c} \models \sigma(\stateformula{\listinv}_\SL)$, we have
$\sigma(v_\mathit{ad}) \geq 1$ and hence $\sigma'(v^{\mathit{start}}) \geq 1$.
    (b) We have $\sigma'(v^{\mathit{end}}) = \sigma'(v^{\mathit{start}}) + \sizeOf(\code{ty}) - 1$. Since the size of an \LLVM{} type is always positive, $\sigma'(v^{\mathit{start}}) \leq \sigma'(v^{\mathit{end}})$ holds.
    The conjuncts from (c) follow from the fact that for each allocation $\alloc{x_1}{x_2}
    \in \AL^s$, $\slmemory^c$ can be split into partial functions $\slmemory_1^c$, $\slmemory_2^c$ with disjoint domains such that $\slmemory_1^c \models \sigma(\stateformula{\listinv}_\SL)$ and $\slmemory_2^c \models \sigma(\stateformula{\alloc{x_1}{x_2}}_\SL)$.
\item[(4)] Since $\slmemory^{c} \models \sigma(\stateformula{\listinv}_\SL)$, we have
    $\slmemory^{c} \models \sigma(v_\mathit{ad})
    + \mathit{off}_i \hookrightarrow_{\sizeOf(\codevar{ty}_i)} \sigma(v_i)$ for all
    $1 \leq i \leq n$.
This implies $\slmemory^{c} \models \sigma'(\stateformula{v_i^{\mathit{start}} \hookrightarrow_{\codevar{ty}_i} v_i}_\SL)$, since
$\sigma'(v_i^{\mathit{start}}) =
\sigma(v_{\mathit{ad}}) + \mathit{off}_i$ and $\sigma'(v_i) = \sigma(v_i)$. 
\item[(5)] (a) Since $\slmemory^{c} \models \sigma(\stateformula{\listinv}_\SL)$, we have
    $\sigma(v_\mathit{ad}) \geq 1$.
 Therefore, $\sigma'(v_i^{\mathit{start}})
= \sigma(v_\mathit{ad}) + \mathit{off}_i \geq 1$ holds
for all $1 \leq i \leq n$.
The premises of (b) and (c) never hold because
 the addresses from the points-to entries are disjoint from the
 addresses of the 
 list invariant. 
 \item[(6)]
    With $\slmemory^{c} \models \sigma(\stateformula{\listinv}_\SL)$, $\slmemory^{c}$ is a model of
    $$\listpred_{\sizeOf(\code{ty}),j}(\sigma(v_\ell)-1, \sigma(v_j),
    [(\mathit{off}_i: \sizeOf(\codevar{ty}_i): {u_i..\sigma(\hat{v}_i)})]_{i=1}^n),$$
see \eqref{listFormulaProof}, where $\slmemory^c \models \sigma(v_j) + \mathit{off}_i \pointstosl_{\sizeOf(\codevar{ty}_i)} u_i$.
    Hence,
    \begin{align*}
    \slmemory^{c} \models\;\; &\listpred_{\sizeOf(\code{ty}),j}(\sigma'(w_\ell), \sigma'(w^{\mathit{start}}), [(\mathit{off}_i: \sizeOf(\codevar{ty}_i): {\sigma'(w_i)..\sigma'(\hat{v}_i)})]_{i=1}^n)\\
        =\;\; &\sigma'(\stateformula{\listinv'}_\SL).
    \end{align*}
\item[(7)] (a) We have $\sigma'(w_\ell) = \sigma(v_\ell)-1 \geq 2-1 = 1$. Furthermore, we have $\sigma'(w^{\mathit{start}}) = \sigma(v_j)$ and since $\slmemory^c \models \sigma(\stateformula{\listinv}_\SL)$ and $\sigma(v_\ell) \geq 2$ it holds that $\sigma(v_j) \geq 1$.
    (b) Let $\models \stateformula{s'} \implies w_\ell = 1$. Since $w_\ell$ is a fresh variable in $s'$ with the only constraint $w_\ell = v_\ell - 1$, we must have $\models \stateformula{s} \implies v_\ell = 2$ and hence $\sigma(v_\ell) = 2$. With $\slmemory^{c} \models \eqref{listFormulaProof}$, we have $\slmemory^{c} \models \listpred_{\sizeOf(\code{ty}),j}(\sigma(1, \sigma(v_j),
[(\mathit{off}_i: \sizeOf(\codevar{ty}_i): {u_i..\sigma(\hat{v}_i)})]_{i=1}^n)$ and by the
    first case of the definition of the semantics of list predicates, we obtain $u_i = \sigma(\hat{v}_i)$ for all $1 \leq i \leq n$. Therefore, with $\sigma'(w_i) = u_i$ and $\sigma'(\hat{v}_i) = \sigma(\hat{v}_i)$ we have $\sigma'(\bigwedge_{i=1}^n w_i = \hat{v}_i)$.
    (c) Let $\models \stateformula{s'} \implies w_\ell \geq 2$. By an analogous reasoning
    as in (b),  we must have $\models \stateformula{s} \implies v_\ell \geq 3$. Then we can apply the second case of the definition of the semantics of list predicates twice to obtain $u_j \geq 1$ and hence, $\sigma'(w_j \geq 1)$.
    (d) We never have $\models \stateformula{s'} \implies w_i \neq \hat{v}_i$ since the $w_i$ are fresh and not constrained, i.e., they are completely unknown values. The only exception is $w_j$, which is $\geq 1$ in case $w_\ell \geq 2$. So if $\models \stateformula{s'} \implies w_\ell \geq 2$ and $\models \stateformula{s'} \implies \hat{v}_j = 0$, we have to show that $\sigma'(w_\ell) \geq 2$ holds. The reason is again that we must have $\models \stateformula{s} \implies v_\ell \geq 3$, so $\sigma(v_\ell) \geq 3$ and hence $\sigma'(w_\ell) = \sigma(v_\ell) - 1 \geq 2$.
\item[(8)] The claims (a)-(f) directly follow from the definition of $\sigma'$.
\end{itemize}
Note that $\sigma'(\stateformula{\alloc{v^{\mathit{start}}}{v^{\mathit{end}}}}_\SL)$ and $\sigma'(\stateformula{\listinv'}_\SL)$ both reason about addresses that are disjoint from all other allocations that are explicitly contained in $\AL^{s'}$ or implicitly contained in $\LI^{s'}$. Moreover, the addresses of $\sigma'(\stateformula{\listinv'}_\SL)$ are disjoint from the addresses on the left-hand sides of all points-to entries contained in $\PT^{s'}$. The reason is that $\slmemory^c \models \{\sigma(
    ((\bigast\nolimits_{\varphi \in \AL^s} \; \; \stateformula{\varphi}_\SL)
    \; \wedge \;
    (\bigwedge\nolimits_{\varphi \in \PT^s} \; \; \stateformula{\varphi}_\SL ))
    \; \ast \;
    (\bigast\nolimits_{\varphi \in \LI^s} \; \; \stateformula{\varphi}_\SL))\}$, where $\LI^s$ contains $\listinv$. Here, $\stateformula{\listinv}_\SL$ ensures that the addresses of the first list element of $\listinv$ (i.e., $\sigma(v_\mathit{ad}) = \sigma'(v^{\mathit{start}}), \ldots, \sigma(v_\mathit{ad})+\sizeOf(\code{ty})-1 = \sigma'(v^{\mathit{end}})$) are disjoint from the addresses of all other list elements of $\listinv$. Therefore,
     $\slmemory^{c} \models \{\sigma'(((\stateformula{\alloc{v^{\mathit{start}}}{v^{\mathit{end}}}}_\SL \ast \bigast\nolimits_{\varphi \in \AL^s} \; \; \stateformula{\varphi}_\SL)
    \; \wedge \;
    (\bigwedge\nolimits_{\varphi \in \PT^{s'}} \; \; \stateformula{\varphi}_\SL ))
    \; \ast \;
    (\stateformula{\listinv'}_\SL \ast \bigast\nolimits_{\varphi \in \LI^s\backslash\{\listinv\}} \; \; \stateformula{\varphi}_\SL))\}$.

Now we show that the third condition of Def.\ \ref{def:representation} holds. 
Since $c$ is represented by $s$, this condition holds for $c$ and $s$ and implies that the
condition holds for all allocations that were already present in $\AL^s$. Moreover, the
fourth condition holds for $c$ and $s$ and implies that there exists an allocation
$\alloc{u^{\mathit{start}}}{u^{\mathit{end}}} \in \AL^c$ such that $(\slassignment^c,\slmemory^c)$
is a model of $m^\mathit{end}_1 - m^\mathit{start}_1
= \sizeOf(\codevar{ty})-1 \,\wedge\, \sigma(v_\mathit{ad}) = m^\mathit{start}_1$, where
$m^\mathit{start}_1$ and $m^\mathit{end}_1$ are the numbers with
$\models \stateformula{c} \implies u^{\mathit{start}} = m^\mathit{start}_1$ and
$\models \stateformula{c} \implies u^{\mathit{end}} = m^\mathit{end}_1$. With the definition of
$\sigma'$ and $\AL^c = \AL^{c'}$ this means that
$\alloc{u^{\mathit{start}}}{u^{\mathit{end}}} \in \AL^{c'}$ such that
$\models \stateformula{c'} \Rightarrow u^{\mathit{start}}
= \sigma'(v^{\mathit{start}}) \wedge u^{\mathit{end}} = \sigma'(v^{\mathit{end}})$. 

The fourth condition of Def.\ \ref{def:representation} for $c'$ and $s'$ directly follows from the fact that this condition holds for $c$ and $s$, where the list invariant is an extension of the new list invariant in $s'$.
\qed
\end{proof}

}
\end{document}